\documentclass{article}

\usepackage[english]{babel}

\usepackage[letterpaper,top=2cm,bottom=2cm,left=3cm,right=3cm,marginparwidth=1.75cm]{geometry}

\usepackage{amsmath}
\usepackage{times}
\usepackage{comment}
\usepackage{makecell}
\usepackage{booktabs}
\usepackage{epigraph}
\usepackage[normalem]{ulem}

\usepackage{authblk}

\usepackage{amssymb}
\usepackage{mathtools}
\usepackage{amsthm}
\usepackage{physics}
\usepackage{bbm}
\usepackage[square,numbers]{natbib}
\usepackage{graphicx}
\usepackage{subcaption}
\usepackage[colorlinks=true, allcolors=teal]{hyperref}
\usepackage[capitalise]{cleveref}
\usepackage[ruled,vlined]{algorithm2e} 

\usepackage{qcircuit}

\newcommand{\defeq}{\vcentcolon=}

\newcommand{\polylog}{\mathrm{polylog}~}

\newcommand{\dual}[1]{\hat{#1}}
\newcommand{\row}{\mathrm{row}}
\newcommand{\Irr}{\mathrm{Irr}}
\newcommand{\bC}{\mathbb{C}}
\newcommand{\N}{\mathbb{N}}

\newcommand{\Hh}{\mathcal{H}}
\newcommand{\mZ}{\mathbb{Z}}
\newcommand{\NOT}{\mathrm{NOT}}
\newcommand{\NOR}{\mathrm{NOR}}

\newcommand{\mul}{\mathrm{mul}}
\newcommand{\C}{\mathrm{C}}
\newcommand{\ceil}[1]{\lceil #1 \rceil}
\newcommand{\floor}[1]{\lfloor #1 \rfloor}
\newcommand{\bo}{\mathbbm{1}}

\newcommand{\bzr}{\mathbf{0}}

\newcommand{\mspan}{\mathrm{span}}
\newcommand{\lcm}{\mathrm{lcm}}
\newcommand{\mend}{\mathsf{END}}

\newtheorem{theorem}{Theorem}[section]
\newtheorem{proposition}[theorem]{Proposition}
\newtheorem{lemma}[theorem]{Lemma}
\newtheorem{fact}[theorem]{Fact}

\theoremstyle{definition}
\newtheorem{definition}[theorem]{Definition}

\theoremstyle{remark}
\newtheorem{remark}[theorem]{Remark}

\usepackage{xcolor}
\definecolor{yuhan}{rgb}{0.9, 0, 0.5}

\definecolor{arm}{rgb}{0.1, 0.4, 0.6}

\definecolor{jose}{rgb}{0.4, .1, 1}

\definecolor{newtext}{rgb}{0.05,0.30,0.65}

\newcommand{\dg}{\dagger}
\newcommand{\ra}{\rightarrow}
\renewcommand{\epsilon}{\varepsilon}

\SetKwRepeat{Do}{do}{while}

\title{The power of oracle access: Optimal sample and query complexity of the abelian state hidden subgroup problem}

\author[1,2]{Yuhan Liu}
\author[3]{Jose Carrasco}
\author[3]{Jens Eisert}
\author[1,2]{Armando Bellante\thanks{armando.bellante@mpq.mpg.de}}

\affil[1]{Max-Planck-Institut für Quantenoptik, Hans-Kopfermann-Straße 1, D-85748 Garching, Germany}
\affil[2]{Munich Center for Quantum Science and Technology (MCQST), Schellingstraße 4, D-80799 Munich, Germany}
\affil[3]{Dahlem Center for Complex Quantum Systems, Freie Universit\"at Berlin, 14195 Berlin, Germany}

\begin{document}
\maketitle

\begin{abstract}
In the quest to identify further quantum algorithms exhibiting superpolynomial speed-ups, a recurring theme is that the complexity of a problem is largely shaped by the input access model.
Here, we study this phenomenon for the \emph{state hidden subgroup problem} (StateHSP), a quantum generalization of the hidden subgroup problem in which the goal is to identify the symmetries of an unknown quantum state. 
For finite abelian groups, existing Fourier-sampling algorithms use $O(\log(|G|)/\epsilon)$ copies of the state, but whether this scaling is optimal has remained open. We settle the complexity of the abelian StateHSP in both the previously studied sample model and a new query model, which is a stronger and operationally natural generalization that provides access to the state-preparation unitary and its inverse. 
In the query model, we give a time-efficient quantum algorithm using $O(\log(|G/H|)/\sqrt{\epsilon})$ forward and inverse queries, and prove a matching $\Omega(\log(|G/H|)/\sqrt{\epsilon})$ lower bound which holds even in the stronger conjugate-query and controlled-query settings. 
By contrast, we show that in the sample model, $\Theta(\log(|G/H|)/\epsilon)$ copies are both sufficient and information-theoretically necessary, even if one allows for arbitrary collective measurements. Thus, the quadratic improvement in $\epsilon$ genuinely arises from coherent access to the preparation circuit. 
As applications, we obtain faster algorithms for learning stabilizer groups, locating unentanglement, and identifying hidden translation symmetries.
\end{abstract}

\epigraph{\textit{Is this a dagger which I see before me,
The handle toward my hand? Come, let me clutch thee.}}{{Shakespeare, Macbeth}}

\hypersetup{linkcolor=black}
\tableofcontents
\hypersetup{linkcolor=teal}

\section{Introduction}
The overarching challenge in quantum algorithms over the last couple of decades has not merely been to find more quantum algorithms, but to understand what makes quantum speedups possible in the first place \cite{MindTheGaps}. 
Which computational primitives are genuinely responsible for an advantage? What constitutes the input to a quantum algorithm, and how may the algorithm access it?
These questions are particularly pressing in quantum learning, where the distinction between classical descriptions, copies of quantum states, and coherent access to state-preparation procedures can fundamentally change the complexity of a task~\cite{biamonte2017quantum,QuantumData}.

We make progress on these broad questions by tightly characterizing the complexity of a fundamental quantum algorithmic problem in two distinct yet natural input models.
We focus on the abelian \emph{state hidden subgroup problem} (StateHSP), a quantum-state generalization of the celebrated abelian \emph{hidden subgroup problem} (HSP), the framework underlying Shor's algorithms for factoring and discrete logarithms~\cite{Shor-1994}. Introduced by Bouland, Giurgi\c{c}\u{a}-Tiron, and Wright~\cite{bouland2024state} and further developed in 
Ref.~\cite{hinsche2025povm}, StateHSP provides a common framework for quantum learning problems in which the 
task is to identify the symmetries of an unknown quantum state. 
More precisely, given a finite abelian group $G$ with a unitary representation $R$ and an unknown state vector $|\varphi\rangle$, the goal is to identify a subgroup $H\leq G$ such that
\begin{equation}
    \begin{aligned}
        R(h)|\varphi\rangle &=|\varphi\rangle \qquad ~~~~ \text{for every }h\in H, \text{ and}
    \\
    |\langle\varphi|R(g)|\varphi\rangle|& \leq 1-\epsilon
    \qquad \text{for every }g\in G\setminus H.
    \end{aligned}
\end{equation}
Thus, $H$ consists precisely of the exact symmetries of the state, while the promise gap $\epsilon$
separates them from every other group element.    
This formulation fully generalizes the standard HSP, which has a constant gap $\epsilon=1$~\cite{bouland2024state}, and encompasses a variety of gapped, symmetry-learning problems, including learning stabilizer groups, locating cuts across which a state is unentangled, and identifying hidden translation symmetries~\cite{hinsche2025povm}.

Existing algorithms based on weak Fourier sampling solve the problem using $O(\log(|G|)/\epsilon)$ copies of the input state~\cite{bouland2024state,hinsche2025povm}.
Although a $\Omega(\log |G|)$ lower bound for constant $\epsilon=1$ follows by reduction from Simon's problem, previous work did not determine the optimal dependence on $\epsilon$, nor did it reveal whether the size of the hidden subgroup $|H|$ influences the complexity.
This leaves two natural questions open: (1) is the $1/\epsilon$ dependence optimal, and (2) can larger hidden subgroups be easier to learn?

We add a third question: (3) can we get better quantum algorithms for StateHSP if we provide access to the input state in a model that goes beyond simple copies? 
We consider access through a state-preparation unitary $U_\varphi$ and its inverse $U_\varphi^\dagger$, where
$U_\varphi|0\rangle=|\varphi\rangle$, rather than only to independent copies of
$|\varphi\rangle$. 
Such access is natural, for instance, when the state is produced by a known unitary circuit, whose gates can be run in reverse. 
Moreover, this model is at least as powerful as having access to the copies\footnote{Indeed, one could always just prepare the copies by running the process forward.}, but the ability to run the state-preparation backwards might enable coherent primitives that are unavailable in the copy model.
Our investigation is therefore guided by the following question:

\begin{center}
    \emph{How much faster can we learn about hidden symmetries of an unknown quantum state when one may coherently prepare and unprepare it, rather than access it only through independent copies?}
\end{center}

We answer this question exactly by tightly characterizing the quantum complexity of StateHSP in both input models, and in all the problem parameters at once: the group size $|G|$, the unknown hidden subgroup size $|H|$, and the promise gap $\epsilon$.

We present algorithms that succeed with constant probability $\geq2/3$ and identify the hidden symmetries using $O({\log(|G/H|)}/{\epsilon})$ copies of the input state or $O({\log(|G/H|)}/{\sqrt{\epsilon}})$ queries to the state-preparation unitary and inverse.
Our upper bounds are both achieved by time-efficient algorithms, of which we carefully study both the quantum and classical additional costs, for any abelian group of the general $G=\mZ_{M_1} \times \mZ_{M_2} \times \dots \mZ_{M_n}$ form.
The copy-based algorithm does not require accessing more than one copy at a time, and the state-preparation-based one does not need controlled access to the state-preparation unitary. We complement these algorithms with strong matching lower bounds that require $\Omega({\log(|G/H|)}/{\epsilon})$ copies, even with collective measurements, and $\Omega({\log(|G/H|)}/{\sqrt{\epsilon}})$ queries, even in the stonger case in which the algorithm has access to conjugate queries, $U_\varphi^*$ and $U_\varphi^T$, and to the controlled versions of all these oracles. 

Taken together, these results settle both the sample and query complexity of StateHSP, tightly. 
They answer all three questions above and isolate both the source and exact extent of the advantage afforded by coherent access to the input state.
Both algorithms recover the hidden subgroup by first accumulating span-increasing generators of its dual group $H^\perp$, whose worst-case number is governed by $\log |G/H|$, and then using them to classically solve for $H$.
With copy access, finding a new and span-increasing generator incurs a $1/\epsilon$ cost.
On the other hand, access to $U_\varphi$ and $U_\varphi^\dagger$ makes the missing generators coherently detectable and hence amplitude-amplifiable, reducing the search cost to $1/\sqrt{\epsilon}$.
Our worst-case matching lower bounds show that this quadratic separation is intrinsic to the access models, and that amplitude amplification is all there is to exploit.
The input model that we introduce completes the view of StateHSP as a proper generalization of HSP, with comparable input access. 
Indeed, at $\epsilon=1$ the bounds recover the ordinary HSP complexity, while the complexities in the two input models pull apart smoothly as the promised gap weakens. 

In the remainder of this introductory section, we comment on the importance of the hidden subgroup problem and its state version, discuss the power of different input models and connect our approach to related work. Finally, we summarize our results.

\subsection{The hidden subgroup problem and its state version}

Quantum computers promise superpolynomial speedups for important computational problems.
However, today, only a few dozen quantum algorithms that exhibit substantial speedups over their classical counterparts are known
\cite{MontanaroOverview,dalzell2023quantum,GrandChallenge, childs2010quantum}, and useful quantum algorithms remain in short supply
\cite{MindTheGaps,king2025quantum,huang2025vast,kapit2025roadblocks}.
This shortfall is becoming increasingly conspicuous as the prospect of building fault-tolerant quantum computers moves closer to technological reality \cite{GoogleDynamicSurfaceCodes}.
Looking back at the roots of our field, two leading directions for advantage emerge: simulating quantum mechanics and solving problems with strong algebraic structure. 
Although the core idea of quantum computing is older and dates back to proposals for efficiently simulating quantum mechanics without the apparent exponential overhead faced by classical machines~\cite{manin1980vychislimoe,Feynman,Lloyd}, Peter Shor largely launched the field of quantum algorithms by showing that factoring and discrete logarithms can be solved in polynomial time, whereas the best known classical algorithms require superpolynomial time \cite{Shor-1994}.

The algorithmic techniques behind Shor's success were tightly linked to the ones used in the work of \citet{bernstein1993quantum} and \citet{Simonproblem94}.
Soon, these techniques were understood to generalize to a bigger problem, the abelian \emph{hidden subgroup problem} (HSP)~\cite{kitaev1995quantum, mosca1998hidden, jozsa2001quantum, NielsenChuang}. 
The formal HSP statement is as follows.

\begin{definition}[Hidden subgroup problem (HSP)] 
    Let $G$ be a finite group and let $H\leq G$ be a subgroup of $G$.
    Let $f:G\to X$ be a function from the group to a finite set $X$, with the  promise that
    \begin{align}
        \forall g_1, g_2 \in G, \quad f(g_1) = f(g_2) \iff \exists h \in H\, \text{ s.t }\, g_1=g_2h.
    \end{align}
    The problem is to identify $H$.
\end{definition}

In words, $f$ is constant on the cosets of $H$ and takes distinct values on distinct cosets, hiding the subgroup this way.
In the quantum setting, algorithms access the function $f$ through the standard reversible oracle $O_f\colon\ket{g}\ket{b}\mapsto \ket{g}\ket{b +f(g)}$, where the elements of $X$ are represented by mutually orthogonal computational-basis states, and the addition is defined by the bitwise $\operatorname{XOR}$.

We say that an algorithm solves the HSP on a group $G$ \emph{efficiently} if, with high probability, it outputs a generating set for $H$ using $\polylog |G| $ queries to the self-inverse oracle $O_f$ and $\polylog |G|$ additional quantum and classical operations.
When $G$ is finite abelian, HSP can be solved efficiently by \emph{weak Fourier sampling}, the algorithmic technique behind Shor, Simon, and Bernstein-Vazirani.
On the other hand, the finite non-abelian case, which contains graph isomorphism and central lattice problems, has resisted three decades of efforts and remains a big open question in general~\cite{ettinger2000quantum, ettinger2004quantum, regev2004quantum, moore2008symmetric, childs2005quantum, childs2025lecture, childs2010quantum}.
Today, HSP remains a useful template to search for quantum advantage, but progress on this framework remains hard-won.

Recently, Bouland, Giurgi\c{c}\u{a}-Tiron, and Wright~\cite{bouland2024state} introduced the \emph{state hidden subgroup problem} (StateHSP). 
This can be seen as a many-body version of HSP, in which a quantum state takes over the role of the hiding function. 
Instead of evaluating a function on group elements, one acts on the state with a unitary representation of the group: acting with an element of the hidden subgroup leaves the state invariant, while acting with any element outside the subgroup perturbs the state by at least $\epsilon$ in fidelity.
The problem is defined as follows.

\begin{definition}[State hidden subgroup problem (StateHSP)~{(formulation of Ref.~\cite{hinsche2025povm}, sample model)}]
\label{def: StateHSP copies}
    Let $G$ be a finite group with a unitary representation $R:G\to \mathrm{U}(\mathcal{H})$ acting on a Hilbert space $\mathcal{H}$, and let $H \leq G$ be a subgroup of $G$. 
    Assume access to copies of an unknown quantum state vector $\ket{\varphi} \in \mathcal{H}$ that is promised to satisfy the following properties, for a known parameter $\epsilon\in(0,1]$:
    \begin{enumerate}
        \item $\forall h \in H, \quad R(h)\ket{\varphi} = \ket{\varphi}.$
        \item $\forall g \not\in H, \quad \abs{\bra{\varphi}R(g)\ket{\varphi}} \leq 1-\epsilon.$
    \end{enumerate}
    The problem is to identify $H$.
\end{definition}

The new formulation generalizes HSP, which reduces to StateHSP.
The reduction proceeds as follows.
The input state vector $\ket{\varphi}={|G|^{-1/2}}\sum_{g \in G} \ket{g}\ket{f(g)}$ can be created with one query to $O_f$, and if we act on the first register with the regular representation $R(g_1)\ket{g} = \ket{g+g_1}$, the promise is satisfied with a constant gap $\epsilon=1$.

At the same time, this new formulation paves the way for efficient algorithms for applications in physics, connecting the two oldest research lines in quantum computing. 
As of today, researchers have shown how StateHSP encompasses several symmetry-learning problems, including learning stabilizer groups, locating hidden tensor-product structure, and identifying translational symmetries~\cite{bouland2024state, hinsche2025povm}.
Its non-abelian version, instead, has recently been leveraged by \citet{lee2025learning} to provide algorithms for learning stabilizers beyond Pauli, and by \citet{gheorghiu2026quantum} to explore the complexity of quantum state isomorphism under a group action.
Any improvement to StateHSP would directly translate to its applications.

While the general non-abelian problem remains hard, \citet{bouland2024state} and \citet{hinsche2025povm} extended weak Fourier sampling to solve finite abelian StateHSP using $O(\log(|G|)/\epsilon)$ copies of the input state, focusing on time-efficient implementations for additive $\mZ_2^n$.
However, the exact complexity of this problem remained open.
Prior to our work, \citet{bouland2024state} gave a sample complexity lower bound of $\Omega(\log(|G|)/\log\log(|G|))$ for constant $\epsilon <1/2$, by reduction from testing for bipartite productness~\cite{jones2025testing}.
For constant $\epsilon=1$, one can establish a stronger $\Omega(\log|G|)$ lower bound by reduction from Simon's problem~\cite{koiran2005quantum}.

Beyond determining the exact dependency on the promise gap $\epsilon$, one can also wonder if the size of the hidden subgroup, relative to the size of the group, $|G/H|$ plays a role in the problem's complexity.
The problem might be easier for larger $|H|$, and light hints come from at least two different places.
First, Simon's classical hardness proof heavily relies on $|H|$ being very small, so that the function could hide one of exponentially many candidates~\cite{Simonproblem94, childs2025lecture}.
Second, \citet{jones2025testing} showed that bipartite productness testing is harder than multipartite productness testing, and in StateHSP, this corresponds to small and large values of $|H|$, respectively.
Integrating $|H|$ in the solution of StateHSP might not look straightforward. 
Indeed, one needs to do so without knowing the cardinality of the hidden subgroup beforehand.
In this work, we succeed in taking $|G/H|$ into account by introducing a new stopping strategy.

\subsection{The power of different access models}

In the previous section, we remarked how HSP reduces to StateHSP.
However, the reduction could be made cleaner by slightly modifying the input access to StateHSP and making the state-preparation unitary and its inverse available to the algorithms.

Indeed, access to a state-preparation unitary is already implicit in the standard formulation of HSP. 
Given an oracle $O_f\colon\ket{g}\ket{b}\mapsto \ket{g}\ket{b +f(g)}$ taking distinct values on the cosets of a hidden subgroup $H$, the canonical quantum algorithm uses the oracle to prepare coset states; it does not receive copies of these states for free. 
From this perspective, the copy-access formulation of StateHSP discards part of the coherent access available in the original HSP. 
The model studied here retains it as a more natural generalization.

\begin{definition}[StateHSP with 
 access to the state-preparation unitaries (query model)]
\label{def: StateHSP with circuit}
    Let $G$ be a finite group with a unitary representation $R:G\to \mathrm{U}(\mathcal{H})$ acting on a Hilbert space $\mathcal{H}$, and let $H \leq G$ be a subgroup of $G$. 
    Assume query access to a unitary $U_\varphi$ that prepares an unknown quantum state vector $\ket{\varphi} \in \mathcal{H}$ (\emph{i.e.,} $|\varphi\rangle=U_\varphi|0\rangle$) and to its inverse $U_\varphi^{-1}$. 
    This state is promised to satisfy the following properties, for a known parameter $\epsilon \in (0,1]$:
    \begin{enumerate}
        \item $\forall h \in H, \quad R(h)\ket{\varphi} = \ket{\varphi}.$
        \item $\forall g \not\in H, \quad \abs{\bra{\varphi}R(g)\ket{\varphi}} \leq 1-\epsilon.$
    \end{enumerate}
    The problem is to identify $H$.
\end{definition}

This is not merely a formal strengthening of the input model. In many quantum-algorithmic and experimental settings, the state is produced by a circuit or device that can be run coherently, and the preparation procedure -- rather than a collection of independently supplied states -- is the natural object to which one has access.
Moreover, the distinction between receiving copies of a state and accessing its preparation circuit is operationally fundamental. Independent copies permit repeated measurements, including arbitrary collective measurements, but they do not allow the algorithm to coherently reverse the preparation process. Access to $U_\varphi$ and $U_\varphi^{\dagger}$, by contrast, enables interference between different calls and makes routines like
amplitude amplification available.
Recent work has begun to reveal the importance of such distinctions. 
Tang and Wright have recently studied how access to a state-preparation unitary $U$ and its inverse $U^\dagger$ changes the complexity of amplitude amplification and estimation~\cite{TangWriteInputModels,RandomPurificationTang,RandomPurificationSimple}.
Along similar lines, previous work by Kothari and O'Donnell has shown how access to the state-preparation circuit can help in mean estimation~\cite{kothari2023mean};  \citet{van2023quantum} showed how state-preparation circuits can help improve quantum state tomography; and Grewal and Liang
have investigated the task of learning unknown quantum channels under different forms of query access~\cite{grewal2025query}.
On a similar spirit, Tang, Wright, and Zhandry have investigated how access to $U^*$ and $U^T$ can change the complexity of a problem, also giving rise to the powerful idea of random purifications~\cite{RandomPurificationTang, RandomPurificationSimple}.
Finally, while it might be impossible to build controlled queries for all black-box oracles~\cite{araujo2014quantum}, Tang and Wright recently showed that control does not help for a large class of problems~\cite{tang2025controlled}.

These results suggest that access models are not merely technical choices in the formulation of a problem: they can determine which quantum algorithmic primitives are available and, ultimately, which speedups are possible.
In the context of HSP, Brassard and H{\o}yer have been the first to exploit this additional structure explicitly: combining Simon's algorithm with amplitude amplification, they have obtained an exact worst-case algorithm for Simon's problem \cite{brassard1997exact}. 

Our work continues this line of thought by extending the underlying amplification strategy from $\mathbb Z_2^n$ and the exact promise $\epsilon=1$ to arbitrary finite abelian groups and the entire range $\epsilon\in(0,1]$. The resulting quadratic improvement in $\epsilon$, together with its matching lower bound, which holds even with access to conjugate queries and their controlled versions, shows precisely how much computational power this coherent access provides.

Further evidence for the naturalness of this access model comes from concurrent work on \emph{quantum state isomorphism} \cite{AlexandruQuantumStateIsomorphism}. There, one is given states $\ket\psi$ and $\ket\varphi$ together with a group action $R$, and asked to identify a hidden element $s\in G$ satisfying $R(s)\ket\psi=\ket\varphi$. Access to the preparation unitaries $U_\psi$ and $U_\varphi$ makes it possible to coherently combine the states into superpositions, like $\frac{1}{\sqrt2}(\ket0\ket\psi+\ket1\ket\varphi)$, that cannot straightforwardly be prepared from independent copies alone. 
Although the algorithmic techniques differ from our amplitude-amplification-based filtering procedure, the two results point towards a common principle: for symmetry-identification problems, the preparation circuit may be the natural quantum input, and retaining coherent access to it may expose algorithmic possibilities hidden by the copy model.

Throughout the remainder of this work, in both models, the representation is also considered part of the problem specification and is provided as the controlled unitary
\begin{align}
    U_R=\sum_{g\in G}\ketbra{g}{g}\otimes R(g).
\end{align}
We assume that the algorithm has access to $U_R$ and $U_R^{-1}$, and the overall time-efficiency of the approach depends on the availability of polylogarithmic-size implementations of the controlled representation action.
For instance, the applications we consider in \cref{sec: applications} admit efficient implementations.

\subsection{Summary of results}

We determine the complexity of the abelian StateHSP in both access models, and simultaneously in all problem parameters, tightly up to constant factors.
Throughout, $G$ is a finite abelian group, written as $G \cong Z_{M_1} \times \dots \times Z_{M_n}$ for arbitrary positive integer $M_1, \dots, M_n$, and $H \leq G$ is the hidden subgroup, whose order is not known to the algorithm. 
\Cref{tab: results summary} summarizes the results concisely.
What follows is a technical overview.

\begin{table}[t]
    \centering
    \resizebox{\linewidth}{!}{%
    \begin{tabular}{|c|c|c|c|}
    \hline
        Access model & Previously known & Upper bound (this work) & Lower bound (this work)  \\ \hline
        \makecell{Copies of the state vector $\ket{\varphi}$\\ (sample complexity)} 
        & \makecell{$O\left(\frac{\log |G|}{\epsilon}\right)$ \small \cite{bouland2024state,hinsche2025povm}\\[4pt] $\Omega \left(\frac{\log |G|}{\log\log |G|}\right)$ \small \cite{bouland2024state}}
        & \makecell{ $O\left(\frac{\log|G/H|}{\epsilon}\right)$ \\[4pt] \small  [\cref{theorem: stateHSP copies z2}]} 
        & \makecell{$\Omega\left(\frac{\log|G/H|}{\epsilon}\right)$ \\ [4pt] \small [\cref{theorem: copy lower bound}]} 
        \\
        \hline
        \makecell{State-preparation unitary $U_\varphi^{\pm1}$\\ (query complexity)} 
        & \makecell{\small previously \\ \small not studied} 
        & \makecell{$O\left(\frac{\log |G/H|}{\sqrt{\epsilon}}\right)$ \\ [4pt] \small [\cref{theorem: stateHSP unified}]} 
        & \makecell{$\Omega\left(\frac{\log |G/H|}{\sqrt{\epsilon}}\right)$ \\ [4pt] \small  [\cref{thm: lower bound}]} 
        \\
        \hline
    \end{tabular}
    }
    \caption{Sample and query complexity of the abelian StateHSP at success probability $\geq \frac{2}{3}$. 
    Previous work established a $O({\log(|G|)}/{\epsilon})$ sample complexity upper bound via Fourier sampling~\cite{bouland2024state,hinsche2025povm} and a $\Omega \left(\frac{\log |G|}{\log\log |G|}\right)$ sample complexity lower bound by reduction to testing bipartite entanglement~\cite{bouland2024state, jones2025testing}. However, the optimal dependency on the $\epsilon$ gap and the hidden subgroup's size $|H|$, as well as the query complexity, were open.
    Our work settles both complexities tightly, up to constants.
    Our query lower bound holds with controlled access to $U_\varphi^{\pm1}$, yet our upper bound does not require controlled access, showing that control does not help further.
    Both our upper bounds are achieved by time-efficient algorithms whenever the representation unitary $U_R$ admits a polylogarithmic-size circuit implementation.}
    \label{tab: results summary}
\end{table}

\subsubsection{Upper bounds}
Our main new algorithmic contributions are (1) an adaptive stopping strategy to make the complexity actually scale with $\log |G/H|$, without previous knowledge of $|H|$, and (2) a Simon-meets-Grover approach that introduces fixed-point amplitude amplification in the weak Fourier sampling scheme to quadratically improve the dependency on the promise gap $\epsilon$ in the query model.
Besides optimizing for sample and query complexity, we carefully bound the additional quantum and classical resources for all of our algorithms.
Our tight bounds became possible through a detailed understanding of the primitives underlying abelian HSP and StateHSP algorithms. 
We provide our summary and intuitions below.

\paragraph{Fourier sampling background}
The standard approach of solving abelian HSP is \emph{weak Fourier sampling}. 
Using the quantum Fourier transform and one query to the input oracle $O_f$, one constructs a quantum circuit to sample from the hidden subgroup's dual subgroup $H^\perp$, namely the subgroup of characters that are trivial on $H$.
After collecting sufficiently many samples from $H^\perp$, a classical algorithm can output a generating set for the hidden subgroup $H$.
\Cref{fig: dual sampling} shows this process.

\begin{figure}[t]
    \centering
    \begin{subfigure}{0.3\linewidth}
        \includegraphics[width=\linewidth]{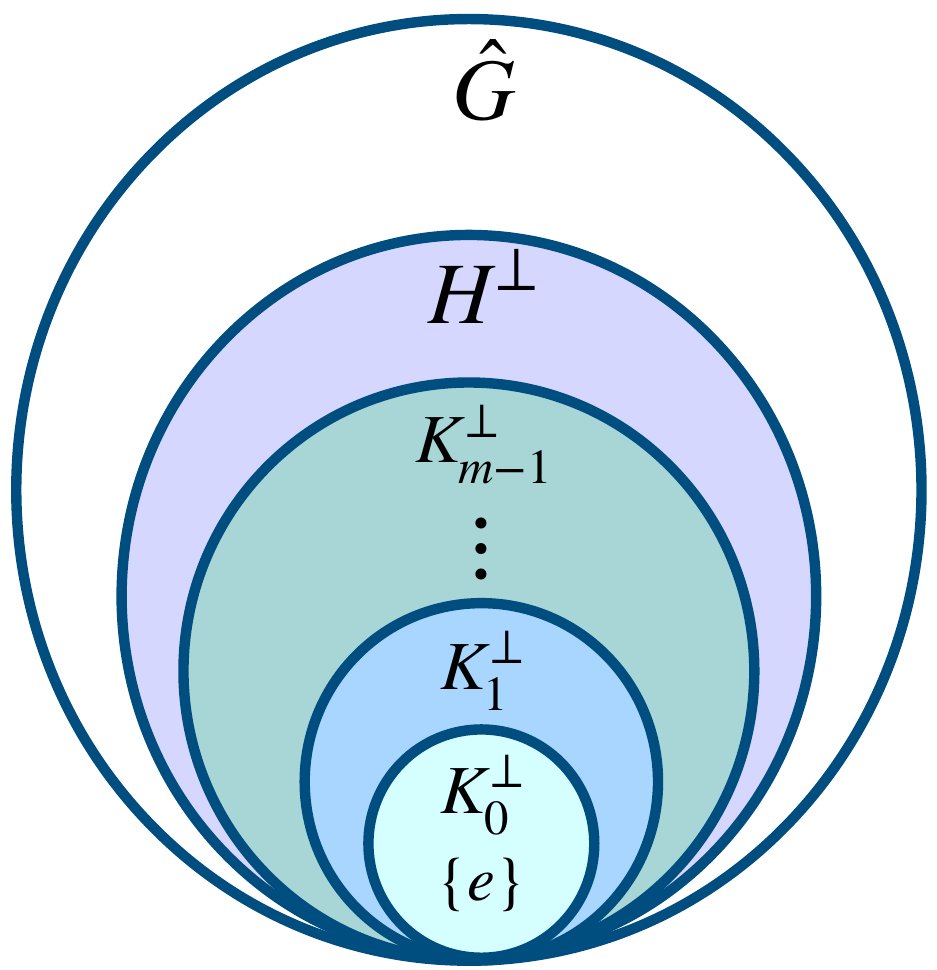}
        \caption{Dual subgroups.}
        \label{fig: dual hierarchy}
    \end{subfigure}
    \hspace{3cm}
    \centering
    \begin{subfigure}{0.3\linewidth}
        \includegraphics[width=\linewidth]{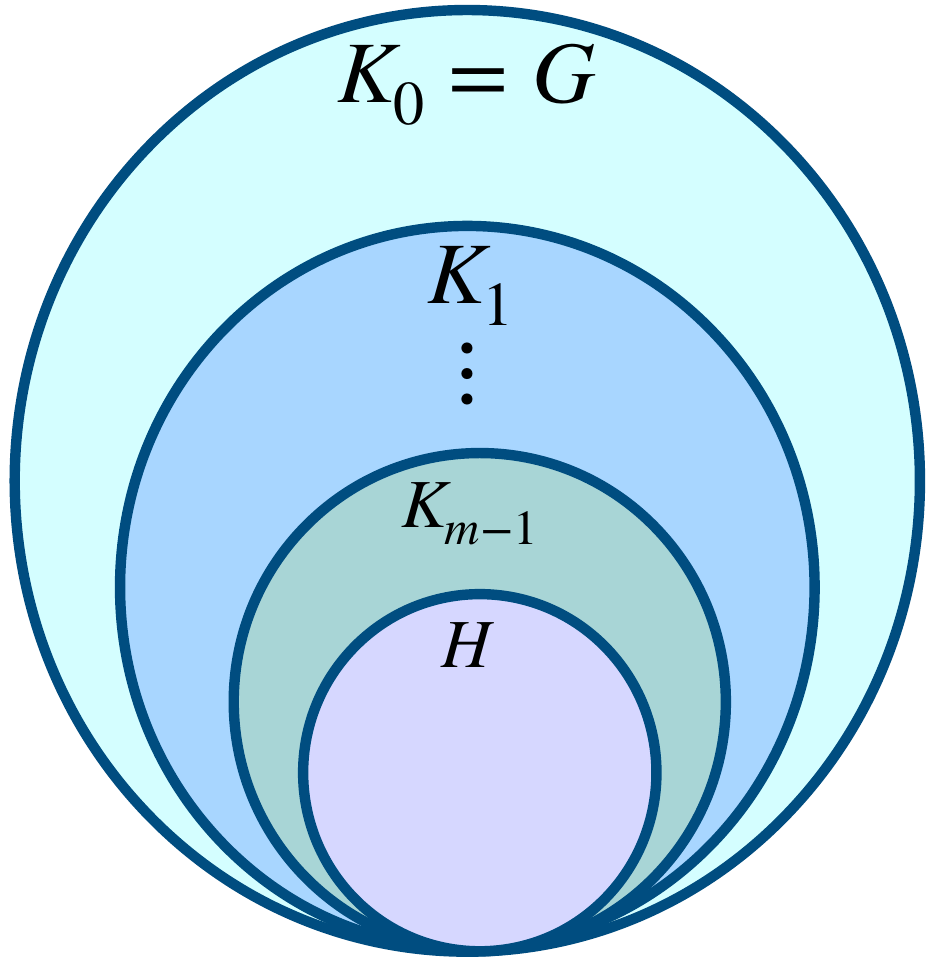}
        \caption{Main subgroups.}
        \label{fig: main hierarchy}
    \end{subfigure} 
    \caption{High-level picture of the Fourier sampling strategy. Fourier sampling enables sampling elements of the dual group $H^\perp$.
    At the beginning, before the first sample, one assumes that $H^\perp$ is just the identity $\{e\}$, and that the hidden subgroup is the whole group $G$. 
    Every successive span-increasing sample updates our beliefs by increasing the size of the current dual $K^\perp$ and reducing the size of the current main $K$.
    Eventually, the samples collectively span the whole dual $H^\perp$ and exactly pinpoint $H$. The figure shows the evolution of the beliefs throughout the algorithm, highlighting corresponding dual and main subgroups with matching colors.
    }    
    \label{fig: dual sampling}
\end{figure}

The reason why $O(\log|G/H|)$ queries to $O_f$ suffice is as follows.
The sampling circuit uses one oracle call and produces samples \emph{uniformly}: each element of $H^\perp$ is observed with probability $|H|/|G|$. 
By Lagrange's theorem, any set of $\floor{\log_2|G/H|}$ independent samples spans $H^\perp$ entirely.
Additionally, by the same theorem, a span-increasing sample appears with probability $\geq 1/2$ until the span is complete.
Therefore, $O(\log{|G/H|})$ samples suffice to obtain an independent set of generators for $H^\perp$ with high constant success probability.

In the abelian StateHSP, one can construct an analogous sampler using the quantum Fourier transform and one copy of the input state.
However, the promise gap $\epsilon \in (0,1]$ perturbs the sampling distribution, which is not necessarily uniform anymore.
The recent work of \citet{hinsche2025povm} strengthens the perturbation analysis of \citet{bouland2024state}, showing that a span-increasing sample appears with probability $\geq \epsilon/2$.
For $\epsilon=1$, this recovers the standard HSP setting, and we offer a qualitative visual interpretation in \Cref{fig: probability distributions}.
Intuitively, the algorithm must collect $O(\log{|G/H|})$ span-increasing generators, but each span-increasing sample appears with probability $\geq \epsilon/2$. 
This leads to a copy complexity of $O(\log{(|G/H|)}/\epsilon)$.

Since the cardinality of $H$ is unknown in advance, previous algorithms aimed for $O(\log |G|)$ samples, fixing a redundant sampling budget before running the algorithm.

\paragraph{Adaptive budgeting: knowing when to stop}
To make the complexity actually scale as $O(\log{|G/H|})$ without knowing $|H|$, we propose a simple yet effective stopping strategy: allocate an \emph{adaptive} sampling budget. 
First, we initialize the budget to a constant. 
Then, we decrease the budget by $1$ after collecting a batch of $\ceil{{2\ln(2)}/{\epsilon}}$ many samples, 
which contains a span-increasing sample with probability $\geq 1/2$, until no more span-increasing samples are available. 
If any of the samples in the batch increases the current span, which can be checked efficiently classically, we encourage further discovery by increasing the remaining budget by $3$. 
Eventually, $H^\perp$ is completely spanned and the budget runs out.
By modeling this process like a random walk, we bound the overall failure probability as a function of the initial budget.
If we want the algorithm to succeed with arbitrary probability $\geq 1-\delta$, the budget shall be initialized to $3\ceil{\log_2 \frac{1}{\delta}}$.

This mechanism allows us to improve the \emph{sample} complexity of Abelian StateHSP.

\begin{theorem}[StateHSP with copies, informal version of~\cref{theorem: stateHSP copies z2}]
\label{thm:stateHSP-c}
    Consider a \textsc{StateHSP} instance as in \cref{def: StateHSP copies}, with $G$ finite abelian and hidden subgroup $H\le G$. Let $\delta\in(0,1]$. 
    Then, there exists a quantum algorithm that identifies $H$ with probability at least $1-\delta$, using 
    \begin{equation}
    t=O\left(\frac{\log|G/H| + \log\frac{1}{\delta}}{\epsilon}\right)
    \end{equation}
    copies of the input state $|\varphi\rangle$, together with $t$ applications of $U_R$ and Quantum Fourier Transforms (QFTs), and $O(t \log^3(|G|)~\polylog M)$ classical operations. 
    $M$ is the least common multiple of the cyclic orders of $G$. 
\end{theorem}

\paragraph{Simon-meets-Grover}
Now turning to the upper bound in the query model. We keep the same adaptive budgeting strategy as in the sample model. 
However, to gain the quadratic advantage in $\epsilon$, we speed up the procedure that obtains a span-increasing sample with probability $\geq 1/2$.
Instead of collecting a batch of $O(1/\epsilon)$ samples, we coherently flag the span-increasing samples with a quantum circuit---a \emph{subspace identifier}---and boost the probability of sampling a new one through fixed-point amplitude amplification~\cite{yoder2014fixed, gilyen2019quantum}. 
The result is that a span-increasing sample can be obtained with a single deep circuit that makes $O(1/\sqrt{\epsilon})$ calls to the state-preparation unitary and its inverse.
We study both the classical and quantum resources required to implement the subspace identifier and execute the fixed-point amplification, ensuring that our algorithms remain time-efficient.
In the process, we had to modify fixed-point amplitude amplification to suit our needs (the details are in \cref{apx: fixed point amp amp}).
The additional classical and quantum resources keep a polylogarithmic scaling in the group size.

\begin{theorem}[StateHSP with state-preparation unitaries, informal version of~\cref{theorem: stateHSP unified}]
\label{thm:stateHSP-u}
    Consider a \textsc{State\-HSP} instance as in \cref{def: StateHSP with circuit}, with $G$ finite abelian and hidden subgroup $H\le G$. Let $\delta\in(0,1]$. 
    Then, there exists a quantum algorithm that identifies $H$ with probability at least $1-\delta$, using 
    \begin{equation}
    T=O\left(\frac{\log|G/H| + \log \frac{1}{\delta}}{\sqrt{\epsilon}}\right)
    \end{equation}
    queries to $U_\varphi^{\pm1}$, together with $T$ applications of $U_R^{\pm1}$ and QFTs,  $O(T\log^2(|G|)~\polylog M)$ additional elementary quantum gates, and $O(T_0\log^3(|G|)~\polylog M)+O(\frac{1}{\sqrt{\epsilon}}~\polylog\frac{1}{\sqrt{\epsilon}})$ classical operations, where $M$ is least common multiple of cyclic orders of $G$ and $T_0=O(\log |G/H| + \log\frac{1}{\delta})$.
\end{theorem}

\begin{figure}[t]
    \centering
    \begin{subfigure}{0.35\linewidth}
        \includegraphics[width=\linewidth]{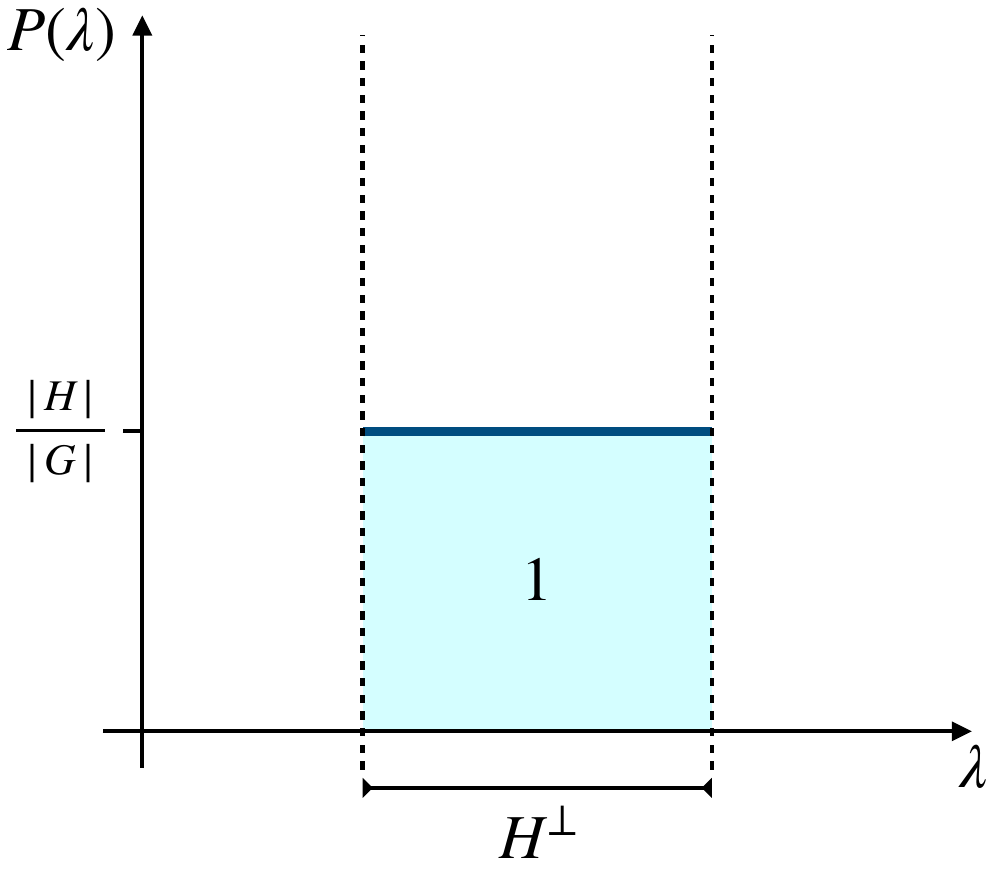}
        \caption{HSP's sampling probability.}
        \label{fig: prob HSP}
    \end{subfigure}
        \qquad\qquad
    \centering
    \begin{subfigure}{0.35\linewidth}
        \includegraphics[width=\linewidth]{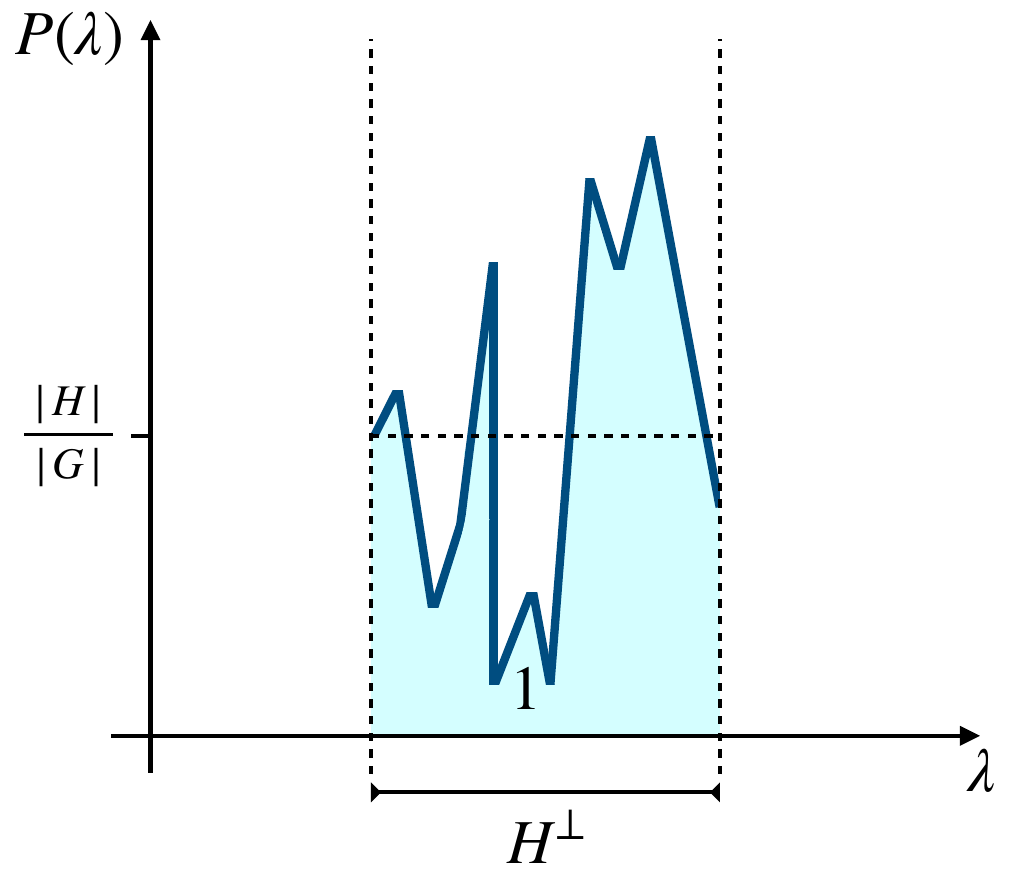}
        \caption{StateHSP's sampling probability.}
        \label{fig: prob SHSP}
    \end{subfigure} 
        \hfill
    \centering
    \begin{subfigure}{0.35\linewidth}
        \includegraphics[width=\linewidth]{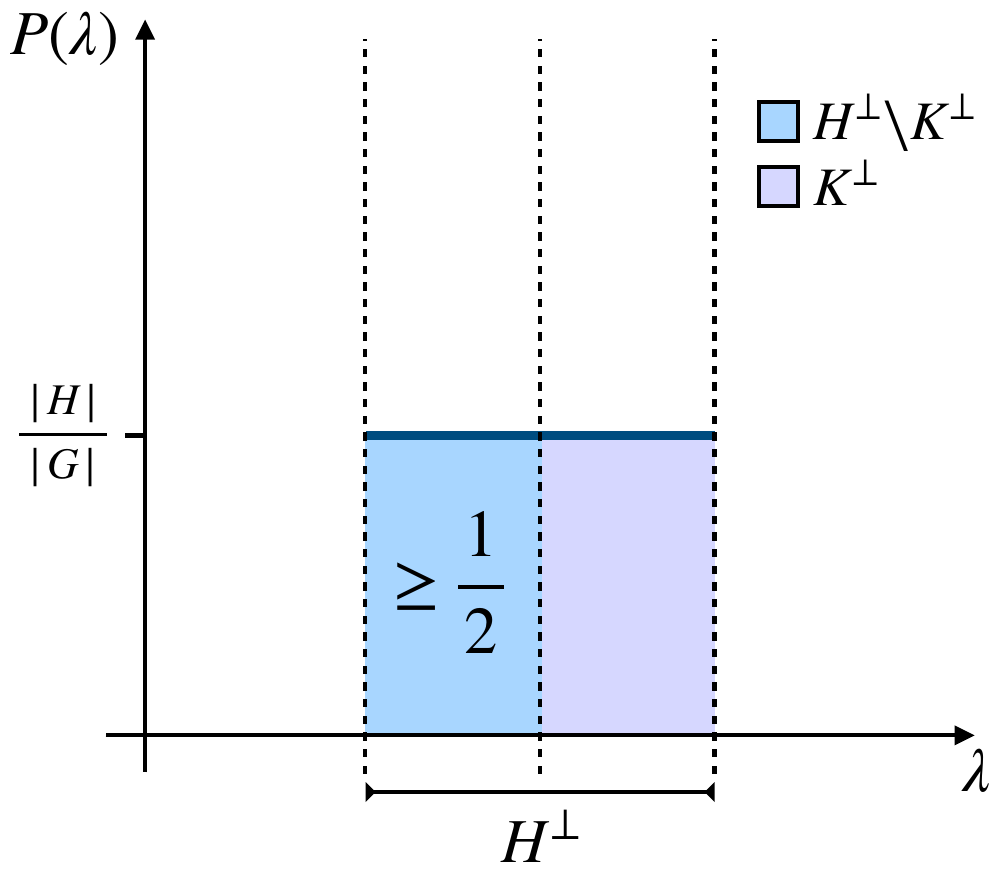}
        \caption{HSP's probability of sampling a span-increasing element.}
        \label{fig: hsp prob Lemma}
    \end{subfigure} 
        \qquad\qquad
    \centering
    \begin{subfigure}{0.35\linewidth}
        \includegraphics[width=\linewidth]{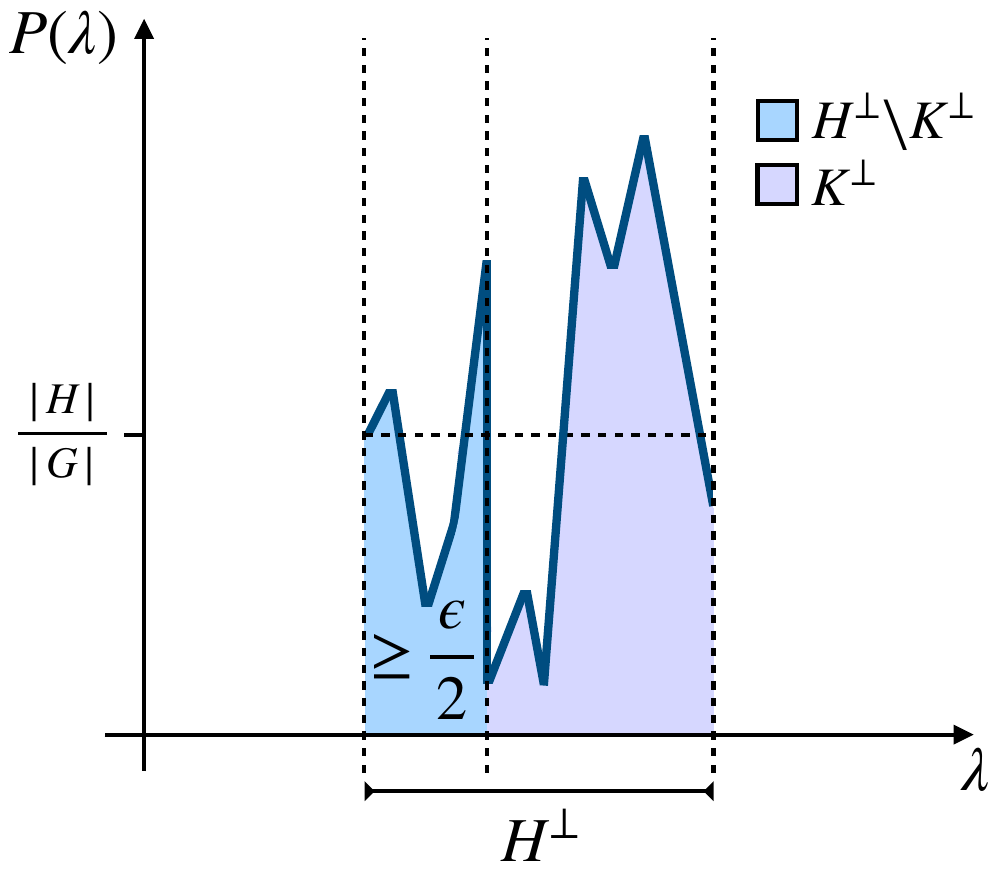}
        \caption{StateHSP's probability of sampling a span-increasing element.}
        \label{fig: prob Lemma}
    \end{subfigure} 
        \hfill
    \caption{This figure illustrates how the probability distribution $P(\lambda)$ over the dual subgroup's elements $\lambda$ \eqref{eqn:varphi-lambda} changes from abelian HSP $(\epsilon=1)$ to abelian State HSP $(\epsilon \in (0,1])$. (d) give a visual interpretation of StateHSP's anticoncentration \Cref{lemma: anticoncentration}.
    Here, $K^\perp < H^\perp$ could be any proper subgroup of $H^\perp$.
    The colored area under the curve represents the cumulative probability over the underlying set; \emph{i.e., } $P(H^\perp) = \sum_{\lambda \in H^\perp} P(\lambda) = 1$.}
    \label{fig: probability distributions}
\end{figure}

The idea of using Grover inside a Simon's-like strategy had already been used in 1997 by \citet{brassard1997exact}.
In that work, the authors give a polynomial-time exact algorithm for Simon's problem~\cite{Simonproblem94}, an emblematic instance of abelian HSP over $G=\mZ_2^n$.
In that case, each span-increasing sample is \emph{guaranteed} to appear with probability \emph{exactly} $1/2$, so with an appropriate subspace identifier for $\mZ_2^n$, one can use the plain Grover to make the algorithm succeed with certainty.
While their method makes Simon's algorithm exact, it does not yield any asymptotic improvements in query complexity, since in standard HSP the probability of obtaining a new span-increasing sample is constant.

Our strategy can be seen as an extension of their approach to StateHSP over arbitrary finite abelian groups.
There are two main differences: we target \emph{arbitrary} abelian groups, and we only have a \emph{lower bound} for the probability of obtaining a span-increasing sample.
For the first reason, we need to extend the subspace identifiers to efficiently flag span-increasing samples in more complicated algebraic structures.
For the second reason, we need to use a modified fixed-point amplitude amplification instead of plain Grover.
The lower bound's scaling with $\epsilon$, however, is what promotes the idea to an effective strategy for an asymptotic advantage.
Whenever the \emph{exact} probability is available in advance, our approach can likewise be made exact by combining an exact Fourier transform~\cite{mosca2004exact} with standard amplitude amplification~\cite{brassard2000quantum}.

\subsubsection{Lower bounds}
We complement our algorithms with two lower bounds that explain why neither access model permits a better dependence on $\epsilon$.
An adversary argument based on a fractional, padded version of Simon's problem establishes the query lower bound, while an information-theoretic construction, combined with Fano's inequality and the Holevo bound, proves the sample lower bound. 
Together with our upper bounds, they paint an essentially complete picture of the abelian StateHSP.

\paragraph{An adversary bound for the queries}
Our query lower bound shows that the amplification advantage is everything we can gain in this stronger access model.
To do so, we build a family of worst-case instances inspired by Simon's problem and prove the joint scaling with the promise gap $\epsilon$ through the adversary bound.

\begin{theorem}[Query lower bound; \cref{thm: lower bound}]
\label{thm: query lb informal}
    For every $\epsilon \in (0,1]$ and integers $n\geq m>1$, there is a family of abelian \textsc{StateHSP} instances with $G=\mZ_2^n$ and known $\log|G/H|=m-1$ such that any quantum algorithm identifying the hidden subgroup with probability at least $2/3$ requires $\Omega\!\left(\frac{\log|G/H|}{\sqrt{\epsilon}}\right)$ oracle queries in the worst case.
    This bound holds with access to the state-preparation unitary $U_\varphi$ and its inverse $U_\varphi^{-1}$, their complex conjugates $U_\varphi^*$ and $U_\varphi^T$, and the controlled version of all these oracles.
\end{theorem}

In Simon's problem, one is granted access to a function $f: \mZ_2^m \to \mZ_2^m$ that hides a secret bitstring $s\in \mZ_2^m$ through the promise that for any two inputs $x\neq y$, we have $f(x) = f(y)$ if and only if $x=y+s$. 
This function defines an HSP with hidden subgroup $H=\{0, s\}$, and it has long been known that the quantum query complexity of locating $s$ is $\Theta(m)$~\cite{Simonproblem94, koiran2005quantum}.
This bound holds even with inverse and controlled access; indeed, the quantum oracle $O_f$ is self-inverse, and one can show that controlled-$O_f$ can be built from a single query.

We build on this result, pad the problem to study the scaling with $\log |G/H|$, and build a fractional oracle to study the scaling in $\epsilon$.
Specifically, we construct a hard family of StateHSP states, the Fractional Padded Simon's states:
\begin{align}
\label{eq: hard query state}
    |\varphi_f^\epsilon\rangle=\sqrt{1-\epsilon}\ket{0}\frac{1}{\sqrt{2^n}}\sum_{x\in \mZ_2^n}\ket{x}\ket{0} + \sqrt{\epsilon}\ket{1} \frac{1}{\sqrt{2^n}}\sum_{x\in \mZ_2^n}\ket{x}\ket{f(x)}.
\end{align}
Here, $f$ only depends on the last $m\leq n$ bits of $x$ and is a Simon's function hiding a bitstring $s\in \mZ_2^m$.
To complete the instance, the abelian group $\mZ_2^n$ acts on the state through its regular representation as a modular adder $R(g)\ket{x} = \ket{x+g}$ on the middle register.
By definition, this is a StateHSP instance with promise gap $\epsilon$ and hidden group $H = \mZ_2^{n-m}\times \{0,s\}$. 
The choice of $n$ and $m$ regulates the scaling of $|H|$ and $|G|$, while $\epsilon$ hides the branch with information on the hidden string $s$. 
For $\epsilon = 1$, this state is the same as in Simon's problem.

To prove the bound for this family, we use the adversary argument and apply two properties of the adversary bound: (1) query complexity and adversary bound coincide up to constants, and (2) rescaling the difference of every two problems' oracles $O_f- O_g$ by a non-zero scalar $c$ suppresses the adversary bound by $1/c$~\cite{belovs2015variations}. 
First, we observe that $|\varphi_f^\epsilon\rangle$ can be prepared through a single call to an oracle $O_f^\epsilon$, reducing the problem of \emph{locating} $s$ \emph{through} $O_f^\epsilon$ to \emph{solving StateHSP through} $U_\varphi$.
Then, we show that for any two Simon's functions $f$ and $g$, it holds that $O_f^\epsilon -O_g^\epsilon = \sqrt{\epsilon} (O_f^1 - O_g^1)$.
Moreover, for any function $f$, the oracle $O_f^1$ can be implemented through a single query to Simon's oracle $O_f$, and vice versa.
By the two-way reduction from Simon's problem, the query complexity (and therefore the adversary bound) of \emph{locating} $s$ \emph{through} $O_f^1$ is $\Theta(m)$.
Through the rescaling property of the adversary bound, we determine that the query complexity of \emph{locating} $s$ \emph{through} $O_f^\epsilon$ is $\Theta(m/\sqrt{\epsilon})$.
The reduction to StateHSP lets us state the lower bound.

By the definition of the query model used in the adversary bound, the lower bound holds even with access to the inverse and with controlled queries.
Moreover, since our worst-case state-preparation unitaries $U_\varphi$ are real, the bound holds even with access to conjugate $U_\varphi^*=U_\varphi$ and $U_\varphi^T=U_\varphi^\dagger$ queries, and their controlled versions.

\paragraph{An information-theoretic bound for the copies}
Our sample lower bound, on the other hand, shows that the quadratic improvement in $\epsilon$ genuinely requires access to the preparation circuit: in the standard sample model, the $1/\epsilon$ scaling is information-theoretically necessary. This settles the sample complexity of the abelian StateHSP at $\Theta(\log(|G/H|)/\epsilon)$, a question left open by Refs.~\cite{bouland2024state,hinsche2025povm}.

\begin{theorem}[Sample lower bound; \cref{theorem: copy lower bound}]
\label{thm: copy lb informal}
    For any $\epsilon \in (0,1]$ and $r,m \in \N$, with $r \geq 0$ and $m \geq 1$, there is a family of abelian \textsc{StateHSP} instances on $G=\mZ_2^{r+2m}$, with known $\log |G/H|=m$, such that any algorithm that identifies the hidden subgroup with probability at least $2/3$ from copies of the input state vector $\ket{\varphi}$, even with collective measurements, requires $\Omega(\log(|G/H|)/\epsilon)$ copies.
\end{theorem}

We build a worst-case family and use information-theoretical tools: Fano's inequality and Holevo's bound.
To construct the hard family, we partition $\mZ_2^n$ using $n=r+2m$ bits, take a phase $\alpha \in [0,2\pi)$, and hide an $m$-dimensional subgroup $\widetilde{H} \cong \mZ_2^m$ of $\mZ_2^{2m}$ in this state:
\begin{align}
    \label{eq: hard sample state}
    \sqrt{1-\epsilon}\ket{0}\frac{1}{\sqrt{2^n}}\sum_{x\in \mZ_2^n}\ket{x} + e^{i\alpha}\sqrt{\epsilon}\ket{1} \frac{1}{\sqrt{2^{n-m}}}\sum_{k\in \mZ_2^r}\ket{k}\sum_{h\in \widetilde{H}}\ket{h}.
\end{align}
The abelian group $\mZ_2^n$ acts on this state through the regular representation, as in the previous example.
One can check that this defines a StateHSP with promise gap $\epsilon$ and hidden subgroup $H=\mZ_2^r \times \widetilde{H}$.
Similar to Padded Fractional Simon's family, the choice of $n$ and $m$ regulates the scaling of $|H|$ and $|G|$, while $\epsilon$ hides the branch with information on the true hidden subgroup $\widetilde{H}$.

To prove the bound, we consider a distribution of the above states that is uniform in $H$ and in $\alpha$. 
Consider an algorithm that applies a POVM on $t$ copies of the state and outputs a subgroup label $Z$, which equals the hidden subgroup label $H$ with probability at least $2/3$.
We sandwich the mutual information $I(H;Z)$, bounding it from both above and below. This is then used to establish the minimal number of copies $t$.

The large number of candidate hidden subgroups $N\geq 2^{m^2}$ allows us to use Fano's inequality and obtain $I(H;Z) \geq \frac{m^2}{6}$; \emph{i.e.,} the algorithm needs roughly $m^2$ bits of information about the subgroup in order to achieve a constant success probability.
Then, through Holevo's bound, we can upper bound the mutual information as a function of the number of copies, $I(H;Z) \leq 2t\epsilon m$; \emph{i.e.,} each copy contributes to roughly $\epsilon m$ bits of information\footnote{The hidden phases $\alpha$ play a crucial role in this upper bound, but we leave this technical discussion for \cref{sec: sample lower bound}.}.
Finally, solving $\frac{m^2}{6} \leq 2t \epsilon m$ for $t$, we obtain the $t \in \Omega(\frac{m}{\epsilon})$ lower bound.

\subsection{Document structure}
\Cref{sec: preliminaries} introduces the preliminaries required to understand our technical results.
\Cref{sec: algorithms} details the algorithmic results and proves the upper bounds in the two input models. 
It contains the adaptive budgeting strategy and the construction of the subspace identifiers that efficiently flag the new span-increasing Fourier samples.
\Cref{sec: query lower bound} presents the query model and the adversary lower bound.
\Cref{sec: sample lower bound} presents the sample model and the information-theoretic lower bound.
\Cref{sec: applications} summarizes how to reinterpret the applications of Refs.~\cite{bouland2024state, hinsche2025povm} through the state-preparation unitary results.
\Cref{sec: outlook} wraps up this work by discussing the results and further research directions. 
Finally, the appendices detail some useful group theoretical facts, dive into integer matrix forms and the computation of the Howell Normal Form, and present an alternative subspace identifier closer to the original one of \citet{brassard1997exact}.

\section{Preliminaries}
\label{sec: preliminaries}
In this section, we provide the preliminaries required to understand our work.
First, we introduce basic concepts in group and representation theory. 
Then, we move to canonical forms for matrices in integer rings and fields.
Finally, we describe two quantum arithmetic operations that we assume are available throughout.

\subsection{Tools in group and representation theory}
We begin by setting the notation used throughout this work. 
Given an integer $n\in \mathbb{Z}_+$, we denote $[n] = \{0, \dots, n-1\}$. 
All logarithms are base $2$ unless otherwise stated.
Throughout, $G$ denotes a finite abelian group, written additively. By the fundamental theorem of finite abelian groups, every such $G$ decomposes as a direct product of cyclic groups,
\begin{equation}\label{eq:abelian-decomposition}
    G\cong \mZ_{M_1}\times\mZ_{M_2}\times\cdots\times\mZ_{M_n},
\end{equation}
for some moduli $M_1,\dots,M_n\in\N$. We accordingly label a group element by a tuple $g=(g_1,\dots,g_n)$ with $g_l\in [M_l]$, and the group operation is componentwise addition, $(g+g')_l := g_l+g_l'\pmod{M_l}$. Two special cases recur in this work: the uniform-modulus group $G=\mZ_M^n$, whose elements are strings in $\{0,\dots,M-1\}^n$ with componentwise addition modulo $M$, and the Boolean group $G=\mZ_2^n$, whose elements are bitstrings in $\{0,1\}^n$ with addition given by bitwise XOR.
 
Every finite abelian group $G$ admits a dual group $\hat{G}$, consisting of its irreducible representations $\Irr(G)$, all of which are one-dimensional. The dual is isomorphic to the group itself, $\hat{G}\cong G$, so we likewise label its elements by tuples $\lambda=(\lambda_1,\dots,\lambda_n)$ with $\lambda_l\in[M_l]$. The irreducible representation $\rho_\lambda$ (equivalently, the character $\chi_\lambda$) indexed by $\lambda$ acts as
\begin{equation}
    \chi_\lambda(g)=\rho_\lambda(g)=\prod_{l=1}^n e^{\frac{2\pi i}{M_l}\, g_l \lambda_l}.
\end{equation}
For the uniform-modulus group $G=\mZ_M^n$ this reduces to $\chi_\lambda(g)=e^{\frac{2\pi i}{M}\, g\cdot\lambda}$ with the inner product $g\cdot\lambda:=\sum_{l=1}^n g_l\lambda_l \pmod M$, and for the Boolean group $G=\mZ_2^n$ it further reduces to $\chi_\lambda(g)=(-1)^{g\cdot\lambda}$.

We now recall Lagrange's theorem and use it to bound the size of any independent generating set of a finite group.
For a set of group elements $S=\{s^{(1)},\dots,s^{(m)}\}$, we write $\mspan(S):=\langle s^{(1)},\dots,s^{(m)}\rangle$ for the subgroup generated by $S$.

\begin{theorem}[Lagrange's theorem]
\label{theorem: lagrange}
    If $G$ is a finite group and $H \leq G$ is a subgroup, then $|H|$ divides $|G|$.
\end{theorem}

\begin{fact}[Generators size]
\label{fact: log-generators}
    Let $G$ be a finite group. Any independent generating set of $G$ has size at most $\floor{\log_2|G|}$, where a set is \emph{independent} if none of its elements lies in the subgroup generated by the others (\cref{def: redundant}).
    Equivalently, any set of more than $\floor{\log_2|G|}$ elements of $G$ is redundant: at least one of its elements lies in the subgroup generated by the others.
\end{fact}
\begin{proof}
    Let $g^{(1)},\ldots,g^{(m)}$ be independent elements of $G$.
    Then the chain
    \begin{equation}
        \{e\} < \langle g^{(1)}\rangle < \cdots < \langle g^{(1)},\ldots,g^{(m)}\rangle
    \end{equation}
    is strictly increasing.
    By Lagrange's theorem, each strict inclusion increases the size by a factor of at least two.
    Hence $2^m \leq |G|$, so $m \leq \floor{\log_2|G|}$, since $m$ is an integer.
\end{proof}

The dual of a subgroup, sometimes called its annihilator, plays an important role in the HSP literature.
For finite abelian groups, it is defined as follows.

\begin{definition}[Dual subgroup $H^\perp$]
    Let $G$ be an abelian group and $H\leq G$ be a subgroup. 
    The dual subgroup of $H$ consists of the characters of $G$ that are trivial on $H$.
    Namely, $H^\perp := \{\lambda \in \dual{G}: \chi_\lambda(h) = 1, \; \forall h \in H\}$.
\end{definition}

For a subgroup $H$ of a finite abelian group $G$, its dual $H^\perp$ is a subgroup of $\hat{G}$ and satisfies $|H^\perp|=|G|/|H| $. 
Taking the dual reverses subgroup inclusion: let $H,K$ be two subgroups of $G$, then $H\leq K$ if and only if $K^\perp\leq H^\perp$. 
Details are presented in~\Cref{app:group}.

We next introduce the group Fourier transform, which maps the basis indexed by group elements to one indexed by irreducible representations and their matrix entries.
We state the definition for arbitrary finite groups before specializing to the abelian case.

\begin{definition}[Group Fourier transform]
Let $G$ be a finite group, and let $\Irr(G)$ denote a complete set of inequivalent irreducible representations (irreps) $\rho_\lambda:G\to \mathrm{U}(\mathbb{C}^{d_\lambda})$. The group Fourier transform is defined by
\begin{equation}\label{eq:QFT}
    \mathcal{F}_G \ket{g}
    = \sum_{\lambda \in \Irr(G)} \sqrt{\frac{d_\lambda}{|G|}}
    \sum_{i,j=1}^{d_\lambda} \rho_{\lambda}(g)_{i,j} \ket{\lambda}\ket{i}\ket{j}.
\end{equation}
Its inverse can be explicitly spelled out and is given by $\mathcal{F}_G^{-1}\ket{\lambda}\ket{i}\ket{j}
    =\sqrt{\frac{d_\lambda}{|G|}}\sum_{g\in G} \rho_\lambda(g^{-1})_{j,i}\ket{g}$.
\end{definition}

For abelian groups, every irreducible representation is one-dimensional, so the registers $\ket{i}$ and $\ket{j}$ are trivial and can be omitted.
The Fourier transform therefore takes the simpler form
\begin{equation}
\mathcal{F}_G\ket{g}
=\frac{1}{\sqrt{|G|}}\sum_{\lambda\in\hat{G}}\chi_\lambda(g)\ket{\lambda},
\end{equation}
mapping each group basis state to a superposition of irrep labels.
We refer to the unitary $\mathcal{F}_G$ as the \emph{quantum Fourier transform} (QFT) over $G$.
For a given cyclic decomposition, we have $\mathcal{F}_G = \bigotimes_l \mathcal{F}_{\mZ_{M_l}}$.
QFTs of arbitrary cyclic order admit efficient approximate implementations~\cite{kitaev1995quantum}, and exact circuits are available when suitably chosen rotations are allowed~\cite{mosca2004exact}.
Throughout the remainder of the paper, we assume exact QFTs and bound the number of calls to them without choosing a specific circuit implementation.
The analysis could be extended to approximate QFTs and concrete resource estimates.

We will use the following two character orthogonality relations.
We prove both below; the second also appears as Fact~1 in Ref.~\cite{hinsche2025povm}.

\begin{lemma}[Character orthogonality]
\label{lemma: character orthogonality}
    Let $H \leq G$ be a subgroup of an abelian group $G$ and $H^\perp \leq \dual{G}$ the corresponding dual subgroup, then 
    \begin{align}
        \sum_{h \in H} \chi_\lambda(h) = \begin{cases}
            |H| & \lambda \in H^\perp, \\
            0 & \lambda \notin H^\perp,
        \end{cases} 
        \qquad \text{ and } \qquad
        \sum_{\lambda \in H^\perp} \chi_\lambda(g) = \begin{cases}
            |H^\perp| & g\in H, \\
            0 & g \notin H.
        \end{cases}
    \end{align}
\end{lemma}
\begin{proof} 
    Let us start with the first equation, and define $S := \sum_{h \in H} \chi_\lambda(h)$. 
    \begin{enumerate}
        \item For every $\lambda \in H^\perp$, we have $\chi_\lambda(h) = 1$ for all $h \in H$. 
        Hence, $\sum_{h \in H} \chi_\lambda(h) = |H|$.
        \item For every $\lambda \notin H^\perp$, there exists $h^{(0)} \in H$ such that $\chi_\lambda (h^{(0)}) \neq 1$.
        Then,
        \begin{align}
            S = \sum_{h \in H} \chi_\lambda(h) = \sum_{h \in H} \chi_{\lambda}(h^{(0)} + h) = \chi_{\lambda}(h^{(0)})\sum_{h \in H} \chi_{\lambda}(h) =\chi_{\lambda}(h^{(0)})S.
        \end{align}
        Equivalently, $(1- \chi_{\lambda}(h^{(0)}))S = 0$. 
        Since $\chi_\lambda (h^{(0)}) \neq 1$, it must follow that $S=0$.
    \end{enumerate}
    
    To prove the second equation, the argument is similar. When $g\in H$, it is clear that $\chi_\lambda(g)=1$ for all $\lambda\in H^\perp$ and therefore $S=\sum_{\lambda\in H^\perp}\chi_\lambda(g)=|H^\perp|$. When $g\not\in H$ there must exist at least one element $\lambda^{(0)}\in H^\perp$ such that $\chi_{\lambda^{(0)}}(g)\neq 1$, for otherwise $g$ would belong to $(H^\perp)^\perp=H$. 
    Note that since $\chi_{\lambda^{(0)} + \lambda}(g)=\chi_{\lambda^{(0)}}(g)\chi_\lambda(g)$, which holds from the group structure of $\hat{G}$, one can proceed by the same argument as above to have $(1- \chi_{\lambda^{(0)}}(g))S = 0$ and conclude that $S=0$ when $g\not\in H$.
\end{proof}

Finally, we define the direct sum of two subgroups. 
\begin{definition}[Direct sum of subgroups]
\label{def:direct-sum-decomposition}
    Let $G$ be an Abelian group and let $H,K\leq G$ be two subgroups. We say $G$ is the direct sum of $H$ and $K$ if (1) $H\cap K=\{0\}$ where $0$ is the group identity element; (2) $G$ is generated by $H$ and $K$, $G=H\oplus K:=\{h+k:h\in H, k\in K\}$.
\end{definition}

For example, fix $0\leq m \leq n$ and let $H=\mZ_2^{n-m} \times \{0\}^m$ and $K= \{0\}^{n-m}\times \mZ_2^{m}$.
These subgroups satisfy $H \cap K = \{0\}$, and every bitstring in $\mZ_2^n$ decomposes uniquely as a sum of an element of $H$ and an element of $K$.
Thus, $\mZ_2^n = H \oplus K$.

\subsection{Integer matrix forms and their kernels}
\label{sec:matrix-form-kernel}
We now turn to discussing integer matrices and their manipulation by row operations.
Given a matrix $A\in\mathbb{Z}_M^{i\times n}$, we seek a canonical form reachable by row operations that preserves its row span, $\mspan(A):=\langle A_1,A_2\cdots, A_i\rangle$, where $A_k$ denotes the $k$-th row. 
When $M$ is prime, $\mathbb{Z}_M$ is a field and $A$ admits a reduced row echelon form.
When $M$ is composite, $\mathbb{Z}_M$ is a ring and the appropriate substitute is the Howell normal form.
We defer all proofs in this subsection to \cref{apx: RREF} and \cref{apx: howell normal form}. 

The reduced row echelon form is defined as follows.

\begin{definition}[Pivot]
\label{def: pivot}
    The first non-zero entry of a matrix row is called a pivot. The column index of the pivot in row $k$ is denoted $j_k$, and the corresponding pivot is $A_{k,j_k}$. 
\end{definition}

\begin{definition}[Reduced row echelon form (RREF)]
\label{def: RREF}
    Let $M$ be a prime number. A matrix $A\in\mZ_M^{i\times n}$ is said to be in RREF if:
    \begin{enumerate}
        \item each pivot is the unique nonzero entry in its column, and the pivot is 1;
        \item pivot columns are strictly increasing left-to-right across rows.
    \end{enumerate}
\end{definition}

Every matrix over $\mZ_M$ with $M$ prime has a unique RREF.
After zero rows are omitted, this form is uniquely determined by the row span.
Thus, if $A$ and $B$ are both in RREF and have no zero rows, then $\mspan(A)=\mspan(B)$ if and only if $A=B$.

The following statements describe how to test membership in the row span of an RREF matrix and maintain RREF when a new row is added.

\begin{proposition}[Membership testing]
\label{prop:membership-testing-RREF}
Let $M$ be prime, $1\leq i<n$, and $A\in\mZ_M^{i\times n}$ in RREF with no all-zero rows and pivot columns $j_1<\dots<j_i$. 
Then every $g\in \mZ_M^n$ decomposes uniquely as $ g=\sum_{k=1}^{i} c_k A_k + q$, where $c_k\in\mZ_M$, $q\in\mZ_M^n$, and $q_{j_k}=0$ for all $k$. 
Moreover, $g\in \mspan(A)$ iff $q=0$. Such a decomposition costs $O((n-i)i~\polylog M)$ classical binary operations. 
\end{proposition}

\begin{theorem}[Incremental Gauss-Jordan elimination]
\label{theorem: gauss-jordan}
    Let $M$ be prime, $1\leq i<n$, and $A\in\mZ_M^{i\times n}$ in RREF with no all-zero rows. For a
    new vector $b\in\mZ_M^n$, the matrix obtained by appending $b$ to the rows of $A$ can be returned in RREF
    using $O((n-i) i~\polylog M)$ binary operations.
\end{theorem}

When $M$ is not prime, a matrix $A \in \mathbb{Z}_M^{i \times n}$ need not admit an RREF under row operations alone. 
In this case, the appropriate canonical form is the Howell normal form, which preserves the row span and can be obtained via row operations over $\mathbb{Z}_M$, possibly after creating new rows.

\begin{definition}[Howell normal form]
\label{def: howell normal form}
    Let $M\in\mathbb{Z}_+$. A matrix $A\in\mathbb{Z}_M^{i\times n}$ over $\mathbb{Z}_M$ is in the Howell normal form if 
    the following is true:
    \begin{enumerate}
        \item There is no zero row. 
        \item The pivot indices $j_1 < j_2 < \cdots < j_i$ are strictly increasing.
        \item For each $1 \leq k \leq i$, the pivot $A_{k,j_k}$ divides $M$.
        \item For each $1 \leq k \leq i$: the entries above each pivot satisfy
              $0 \leq A_{k',j_k} < A_{k,j_k}$ for all $k' < k$,
              and the entries below each pivot satisfy $A_{k',j_k} = 0$ for all $k' > k$.
        \item  (Extended rows property) Let $v$ be an element in the row span of $A$, i.e., $v\in\langle A_1,A_2,\cdots,A_i\rangle$ where $A_k$ denotes the $k$-th row of $A$. If the first $j_k$ components of $v$ are zero for some $1 \leq k \leq i$,
              then $v \in \langle A_{k+1}, \ldots, A_i \rangle$.
    \end{enumerate}
\end{definition}

Specifically, multiplying row $k$ by $M/A_{k,j_k}$ yields a row in 
$\langle A_{k+1}, \ldots, A_i \rangle$. Consequently, the subgroup generated by the rows of $A$ has cardinality $\prod_{k=1}^{i} (M/d_k)$, where $d_k := A_{k,j_k}$ is the $k$-th pivot. We note that given a matrix $A$ over $\mathbb{Z}_M$, its Howell normal form is unique. Furthermore, if $A$ and $B$ are both in Howell normal form, then $\mspan(A)=\mspan(B)$ leads to $A=B$.

\begin{proposition}[Membership testing]
\label{prop:membership-testing-HNF}
    Let $A \in \mathbb{Z}_M^{i \times n}$ in Howell normal form with pivot columns $j_1 < \cdots < j_{i}$. Then any $g \in \mathbb{Z}_M^n$ decomposes 
    uniquely as $g = \sum_{k=1}^{i} c_k A_k + q$, 
    where $c_k \in [M/A_{k,j_k}]$ for all $k$ and $q_{j_k}\in [A_{k,j_k}]$ for all $k$. 
    Moreover, $g\in \mspan(A)$ iff $q=0$. Such a decomposition costs $O(ni~\polylog M)$ binary operations. 
\end{proposition}

\begin{theorem}[Howell normal form algorithm~{\cite{storjohann1998fast}}]
\label{theorem: howell algorithm}
    Let $A\in\mathbb{Z}_M^{i\times n}$ be a matrix over $\mathbb{Z}_M$.
    There exists an algorithm that brings $A$ into Howell normal form using $O(n^2\max(n,i) ~\polylog M)$ binary operations.
\end{theorem}

\begin{table}[t]
    \centering
    \begin{tabular}{|c|c|c|c|}
    \hline
                  &  Update the normal form & Generate the kernel & Membership testing\\
                  \hline
        $\mZ_M^n$, $M$ prime & $O((n-i)i ~\polylog M)$ & $O((n-i)i ~\polylog M)$ & $O((n-i)i ~\polylog M)$ \\ \hline
        $\mZ_M^n$, $M$ composite & $O(n^2\max(n,i) ~\polylog M)$ & $O(n i^2 ~\polylog M)$ & $O(ni~\polylog M)$ \\ \hline
    \end{tabular}
    \caption{Classical cost of integer matrix algorithms with sparse representation of vectors.  }
    \label{tab:classical}
\end{table}

We comment that when bringing a matrix $A\in\mathbb{Z}_M^{i\times n}$ to RREF, the number of rows would not increase; while when bringing a matrix $A\in\mathbb{Z}_M^{i\times n}$ to Howell normal form, the number of rows may increase up to $n$. 

\paragraph{Generating the kernel.} 
Given a matrix in the RREF or Howell normal form, one can generate its kernel efficiently. 
The proofs are in \cref{apx: RREF} and \cref{apx: howell normal form}.

\begin{definition}[Kernel]
    Let $A\in\mathbb{Z}_M^{i\times n}$ be a matrix over $\mathbb{Z}_M$. We define $\ker(A)$ as $\ker(A)=\{x\in \mathbb{Z}_M^{n\times 1}| Ax=0\}$. 
\end{definition}

\begin{theorem}[Generating the kernel, $\mZ_M^n$ prime]
\label{theorem: kernel generation}
    Let $M \in \mathbb{Z}_+$ be prime and $A \in \mZ_{M}^{i\times n}$, for $i \leq n$, in RREF with no all-zero rows.
    Its kernel $\ker(A)$ can be generated by $n-i$ vectors with at most $i+1$ non-zero entries each, and there exists an algorithm that outputs these vectors in $O((n-i)n~\polylog M)$ classical binary operations, or in $O((n-i)i~\polylog M)$ classical binary operations if using sparse representation of the vectors.
\end{theorem}

\begin{theorem}[Generating the kernel, $\mZ_M^n$]
\label{theorem:kernel-M}
    Let $M\in\mathbb{Z}_+$ and $A \in \mZ_{M}^{i\times n}$, for $i \leq n$, in Howell normal form with no all-zero rows.
    Its kernel $\ker(A)$ can be generated by $n$ vectors with at most $i+1$ non-zero entries each, and there exists an algorithm that outputs these vectors in $O((ni^2+n^2) ~\polylog M)$ classical binary operations, or in $O(n i^2 ~\polylog M)$ classical binary operations if using sparse representation of the vectors.
\end{theorem}

\subsection{Quantum arithmetic primitives}
\label{sec: quantum arithmetic primitives}
Finally, we introduce two quantum arithmetic primitives and specify the resource bounds we assume for their construction and implementation.

\begin{definition}[Quantum NOR]
\label{def: NOR gate}
    Let $M,n \in \mZ_+$ and let $a_1, \dots, a_n \in [M]$ be encoded in $n$ $M$-level quantum systems $\ket{a_1}\dots\ket{a_n}$. A quantum NOR gate is a circuit $U_{\NOR}$ that implements 
    \begin{align}
        U_{\NOR}:\ket{a_1}\dots\ket{a_n} \ket{0} \mapsto 
        \begin{cases}
            \ket{a_1}\dots\ket{a_n} \ket{1} & \text{if } a_1=\dots = a_n=0,\\
            \ket{a_1}\dots\ket{a_n} \ket{0} & \text{otherwise}.
        \end{cases}
    \end{align}
\end{definition}

\begin{definition}[Quantum modular multiply-adder]
\label{def: modular multiplier}
    Let $M \in \mZ_+$ and let $s \in [M]$ be a classically known value. 
    A quantum modular multiply-adder is a circuit $U_{\mul(s)}$ acting on two $M$-level 
    quantum systems as
    \begin{align}
        U_{\mul(s)}:\ket{a} \ket{b} \mapsto \ket{a} \ket{b + s \cdot a \bmod M},
        \qquad a, b \in [M].
    \end{align}
\end{definition}
For every fixed $s$, the map $(a,b) \mapsto (a,\, b + s \cdot a \bmod M)$ is a bijection on $[M]^{2}$, with inverse $(a,b) \mapsto (a,\, b - s \cdot a \bmod M)$. 
Thus, $U_{\mul(s)}$ permutes the computational basis and is unitary, even when $s$ and $M$ are not coprime.

We assume that descriptions of the NOR and modular multiply-adder circuits can be generated classically using $O(n~\polylog M)$ and $O(\polylog M)$ classical binary operations, respectively.
Using reversible binary arithmetic, we also assume that these circuits can be implemented with $O(n~\polylog M)$ and $O(\polylog M)$ elementary one- and two-qubit gates, respectively.

\section{Algorithms}
\label{sec: algorithms}
Following the approach of weak Fourier sampling, our algorithms for StateHSP recover the hidden subgroup $H\leq G$ by first finding a generating set $S^\perp$ for its dual subgroup $H^\perp\leq\hat{G}$.
Once $\langle S^\perp\rangle=H^\perp$, we recover $H$ by classical postprocessing using the identity $(H^\perp)^\perp=H$.

We introduce a sampling framework that automatically stops once the sampled elements generate $H^\perp$, without requiring prior knowledge of $|H|$.
The framework uses three primitives: $\textsc{HSP-Sampler}$, $\textsc{Update}$, and $\textsc{Dual-Solver}$.
The classical routines $\textsc{Update}$ and $\textsc{Dual-Solver}$ follow readily from the integer-matrix machinery of \cref{sec:matrix-form-kernel}, whereas $\textsc{HSP-Sampler}$ represents the quantum core of the algorithm and has different implementations depending on the access model. 

After presenting the framework, we develop these implementations in stages.
We start from the Fourier sampling circuit and show how to use it when the algorithm has access to copies of the input state. 
We then consider access to a state-preparation unitary and inverse, which allows us to combine the Fourier sampling circuit with amplitude amplification.
We explain how to implement the required flagging procedure efficiently.
The main text presents a flagging method based on solving a system of linear equations, while an alternative extending the approach of \citet{brassard1997exact} appears in~\cref{app:second-approach}.

We begin by formalizing the three abstract primitives used by the framework.

\begin{definition}[$\textsc{HSP-Sampler}$]
\label{def: HSP-sampler}
    Consider a StateHSP instance with hidden subgroup $H$ and dual subgroup $H^\perp$.
    Let $S^\perp$ be a data structure encoding a generating set for a subgroup $K^\perp \leq H^\perp$, so that $\langle S^\perp\rangle=K^\perp$, and let $\tilde{\delta} \in (0,1)$.
    An \emph{HSP-sampler} is a randomized classical-quantum routine $\textsc{HSP-Sampler}(S^\perp,\tilde{\delta})$ with the following behavior:
    \begin{enumerate}
        \item If $K^\perp = H^\perp$, it outputs $\mathsf{END}$ with certainty.
        \item If $K^\perp < H^\perp$, it outputs an element of
        $H^\perp \setminus K^\perp$ with probability $\geq 1-\tilde{\delta}$, and outputs $\mathsf{END}$ otherwise.
    \end{enumerate}
    Thus, the routine fails only when it outputs $\mathsf{END}$ before the encoded generators span $H^\perp$.
\end{definition}

The data structure encoding the generating set is maintained by a classical routine that performs the update.

\begin{definition}[$\textsc{Update}$]
\label{def: update}
    The classical deterministic routine $\textsc{Update}(S^\perp, s^\perp)$ takes a data structure encoding a generating set $S^\perp$ and updates it to encode a generating set for $\langle S^\perp\cup \{s^\perp\}\rangle$.
\end{definition}

Finally, a classical routine converts generators of a subgroup of $\widehat{G}$ into generators of its dual in $G$.

\begin{definition}[$\textsc{Dual-Solver}$]
\label{def: dual solver}
    Let $S^\perp$ be a data structure encoding a generating set for a subgroup $K^\perp \leq \widehat{G}$. 
    A \emph{Dual-solver} is a deterministic classical routine $\textsc{Dual-Solver}(S^\perp)$ that outputs a data structure encoding a generating set $S$ for the corresponding group $K \defeq (K^\perp)^\perp \leq G$.
\end{definition}

Using these three primitives, we solve abelian StateHSP by repeatedly sampling elements of $H^\perp$ that enlarge the subgroup generated so far, until we generate the whole subgroup.
The specifics of $\textsc{HSP-Sampler}$ allow us to stop after $O(\log |G/H|)$ samples, with the special $\mathsf{END}$ symbol indicating when the subgroup $H^\perp$ is fully generated.
The only possible error is premature termination: the sampler may return $\mathsf{END}$ before $S^\perp$ spans all of $H^\perp$.

To control the total failure probability, we use an adaptive stopping strategy.
The algorithm starts with a fixed iteration budget and increases it whenever the sampler returns a new element.
We show that this strategy achieves the desired running time and success probability without prior knowledge of $|H|$ or an additional $\polylog|G|$ overhead.

\begin{theorem}[HSP-Sampling framework]
\label{theorem: HSP-sampling framework}
    Consider a StateHSP instance over a finite abelian group $G$ with hidden subgroup $H\leq G$, and let $\delta\in(0,1]$.
    Then, Algorithm~\ref{alg: HSP-Sampling framework} outputs a generating set for $H$ with probability at least $1-\delta$.
    Moreover, the algorithm makes at most $3(\floor{\log_2|G/H|}+\ceil{\log_2\tfrac{1}{\delta}})$ calls to $\textsc{HSP-Sampler}$ with failure parameter $\tilde{\delta}=1/2$, at most $\floor{\log_2 |G/H|}$ calls to $\textsc{Update}$, and one call to $\textsc{Dual-Solver}$.
\end{theorem}
\begin{proof}
    We begin by studying the algorithm's running time.
    Let $i_t$ be the value of $i$ after $t$ iterations, and define the remaining iteration budget $B_t\coloneqq 3(i_t+b)-t$. 
    The loop continues exactly while $B_t> 0$.
    The initial budget is $B_0 = 3b$, while in general we have
    \begin{align}
        B_t = B_{t-1} + 
        \begin{cases}
            +2 & \text{if } \textsc{HSP-Sampler} \text{ outputs a new generator,}\\
            -1 & \text{if } \textsc{HSP-Sampler} \text{ outputs $\mathsf{END}$}.
        \end{cases}
    \end{align}

    Every non-$\mathsf{END}$ output of the $\textsc{HSP-Sampler}$ lies in $H^\perp \setminus \langle S^\perp\rangle$, and it strictly increases the subgroup generated by $S^\perp$, the counter $i$, and, consequently, the remaining budget.
    After $S^\perp$ generates the entire subgroup, the sampler only outputs $\mathsf{END}$, decreasing the budget and bringing the algorithm to termination.
    Since $|H^\perp|=|G/H|$, every sequence $\langle s_1^\perp\rangle
    < \langle s_1^\perp,s_2^\perp\rangle
    < \cdots
    < H^\perp$ of span-increasing generators of $H^\perp$ has length at most $\floor{\log_2|G/H|}$ (\cref{fact: log-generators}). 
    This upper bounds the value of $i$ by $\floor{\log_2|G/H|}$ and the total number of iterations by $3(\floor{\log_2|G/H|}+\ceil{\log_2\tfrac{1}{\delta}})$, proving our running time claims.

    We now turn to the algorithm's correctness.
    One call to the sampler can only fail by yielding an $\mathsf{END}$ even if there are new generators to output.
    To bound our overall failure probability, it remains to bound the probability that the loop terminates before $S^\perp$ spans the whole subgroup $H^\perp$.
    
    For $N,k\geq 0$, let $p_N(k)$ denote the supremum of the conditional probability: conditioned on the current budget being $k$, the budget reaches 0 within the next $N$ iterations and $\langle S^\perp\rangle<H^\perp$ at termination. The supremum is taken over all possible sampler outputs throughout the algorithm.  
    When the budget is over the algorithm terminates, so $p_N(0)=1$.
    We claim that, for every $N\geq 0$ and for all $k \geq 1$,
    \begin{align}
    \label{eq: algo probability bound}
        p_N(k) < 2^{-k/3}, \quad \text{independently of $N$.}
    \end{align}
    We prove this by induction on $N$.

    \begin{enumerate}
        \item \emph{Base case.} When $N=0$ and $k\geq 1$, we have $p_0(k)=0 < 2^{-k/3}$.
        \item \emph{Induction step.} Now suppose the claim holds for $N-1$.
        By \cref{def: HSP-sampler}, let $a \geq \frac{1}{2}$ be the probability that the next sampler call returns a new generator, which is valid since $\langle S^\perp\rangle<H^\perp$.
        If a new generator is returned, then the budget becomes $k+2$; if $\mathsf{END}$ is returned, then the budget becomes $k-1$. 
        Then,
        \begin{align}
            p_N(k) &\leq a p_{N-1}(k+2) + (1-a)p_{N-1}(k-1)\\
            \nonumber
            &\leq \frac{1}{2} p_{N-1}(k+2) + \frac{1}{2}p_{N-1}(k-1) \qquad (\text{since }p_{N-1}(k-1) \geq p_{N-1}(k+2))\\
              \nonumber
            &\leq  \frac{1}{2}2^{-\frac{k+2}{3}} + \frac{1}{2}2^{-\frac{k-1}{3}}\qquad\qquad\qquad\quad (\text{by induction, and note that $p_{N-1}(0)=1$})\\
              \nonumber
            & = \frac{3}{2^{\frac{5}{3}}} 2^{-\frac{k}{3}} < 2^{-\frac{k}{3}}.
              \nonumber
        \end{align}
    \end{enumerate}
    The algorithm starts with budget $B_0=3b$. 
    Therefore, for every finite horizon $N$, the probability of terminating the loop before $S^\perp$ entirely generates $H^\perp$ is at most $p_N(3b) < 2^{-3b/3} = 2^{-b} \leq \delta.$
\end{proof}

\begin{algorithm}[t]
  \caption{HSP-Sampling framework.}
  \label{alg: HSP-Sampling framework}
  \DontPrintSemicolon          
  \SetAlgoLined                
  \SetKwInOut{Input}{Input}
  \SetKwInOut{Output}{Output}
  \BlankLine          

  \Input{An instance of the StateHSP problem over a finite abelian group $G$, hiding the subgroup $H \leq G$; a failure probability $\delta \in (0,1]$; an HSP-sampler (\cref{def: HSP-sampler}); an Update (\cref{def: update}); and a Dual-solver (\cref{def: dual solver}).}
  \Output{A generating set for the hidden subgroup $H \leq G$.}

  Set $S^\perp \gets \emptyset$;\;
  Set $t \gets 0$; \tcp*[l]{Iteration counter}
  Set $b \gets \ceil{\log\tfrac{1}{\delta}}$; \tcp*[l]{Initial budget}
  Set $i \gets 0$; \tcp*[l]{span-increasing generators counter}

  \While{$t < 3(i+b)$}{
      Set $s^\perp \gets \textsc{HSP-Sampler}(S^\perp,\tfrac{1}{2})$;\;
      \If{$s^\perp \neq \mathsf{END}$}{
        Set $S^\perp \gets \textsc{Update}(S^\perp, s^\perp)$;\;
        Set $i \gets i+1$;\;
      }
      Set $t \gets t+1$;\;
  }
  \Return{$\textsc{Dual-solver}(S^\perp)$.}
\end{algorithm}

A simpler strategy would use a known upper bound $B \geq \floor{\log |G/H|}$, call the sampler with failure probability $\tilde{\delta}=\frac{\delta}{B}$, and stop at the first $\mathsf{END}$. 
For example, one could take $B=\floor{\log|G|}$.
This would introduce a multiplicative $\log(B/\delta)$ overhead in place of the additive $\log(1/\delta)$ term achieved by our stopping strategy, preventing us from matching the query lower bounds up to constant factors.

The following sections will focus on the implementation of the three primitives used in the sampling framework: $\textsc{HSP-Sampler}$, $\textsc{Update}$, and $\textsc{Dual-Solver}$. 

\subsection{Fourier sampling circuit}
\label{sec: fourier sampling}
The $\textsc{HSP-Sampler}$ is the only primitive that requires a quantum construction. 
Its foundation is the standard \emph{weak Fourier sampling} circuit, whose output distribution is supported on $H^\perp$.
As long as the sampled elements generate a proper subgroup of $H^\perp$, each circuit evaluation returns an element outside that subgroup with probability $\Omega(\epsilon)$.
This dependence on $\epsilon$ accounts for the $O(\log(|G|)/\epsilon)$ copy complexity of standard StateHSP algorithms.
With access to $U_\varphi$ and $U_\varphi^{-1}$, we can use amplitude amplification to boost the probability of obtaining a new generator.
This requires a flagging procedure that distinguishes elements inside the current span from those outside it.
We begin by reviewing the Fourier sampling circuit and its output distribution~\cite{bouland2024state,hinsche2025povm}. 

For the abelian StateHSP, the Fourier sampling circuit is
\begin{align}
\label{eq: fourier sampling circuit}
    U_{F} \defeq \widetilde{U}_{F}(I\otimes U_\varphi) \defeq (\mathcal{F}_G \otimes I) U_R (U_G \otimes U_\varphi),
\end{align}
where $U_\varphi$ prepares the input state $\ket{\varphi}$, and $U_G$ prepares the uniform superposition over $G$: 
\begin{align}
\label{eq:U_G}
U_G \ket{0}^{\otimes n}=\frac{1}{\sqrt{|G|}}\sum_{g \in G} \ket{g}.
\end{align}
Here, $U_R=\sum_{g \in G} \ketbra{g}{g} \otimes R(g)$ implements the controlled representation, and $\mathcal{F}_G$ is the group Fourier transform.
The circuit $\widetilde{U}_F$ acts on an already-prepared copy of $\ket{\varphi}$, so it can be used in the copy-access model.

Applying $U_F$ to the all-zero state gives
\begin{equation}
\label{eq: state weak fourier sampling}
\begin{aligned}
    \ket{\varphi_F} &= (\mathcal{F}_G \otimes I) U_R (U_G \otimes U_\varphi)|0\rangle^{\otimes n}|0\rangle\\
    & = \frac{1}{|G|} \sum_{g \in G} \sum_{\lambda \in \Irr(G)} \chi_\lambda(g) \ket{\lambda} \otimes R(g) \ket{\varphi}.
\end{aligned}
\end{equation}

Access to the group Fourier transform also allows us to implement $U_G$: we may take $U_G=\mathcal{F}_G^{-1}$, since the inverse Fourier transform maps the trivial character label to the uniform superposition over $G$.
This choice suffices for our asymptotic results, but we keep $U_G$ as a separate primitive to allow more precise resource estimates and potentially simpler implementations.

Indeed, $U_G$ only needs to prepare the uniform superposition from a fixed input state, whereas $\mathcal{F}_G$ must implement the Fourier transform on every basis state.
For some groups, this state-preparation task admits a smaller circuit.
For example, for $G=\mathbb{Z}_M$ encoded in qubits, even when $M$ is not a power of two, the uniform superposition $\frac{1}{\sqrt{M}} \sum_{i=0}^{M-1}\ket{i}$ can be prepared using $O(\log M)$ elementary gates~\cite{shukla2024efficient,bellante2026compiling}, while the full quantum Fourier transform may require a larger circuit.

The first register of $\ket{\varphi_F}$ encodes elements of the dual group $\widehat{G}=\Irr(G)$.
We review the distribution obtained by measuring this register~\cite{bouland2024state,hinsche2025povm}, starting with an explicit expression for the probability of each outcome.

\begin{fact}
    Measuring the first register of $\ket{\varphi_F}$ in~\eqref{eq: state weak fourier sampling} yields $\lambda\in\widehat{G}$ with probability
    \begin{align}
    \label{eq: lambda probability}
        P(\lambda) = \frac{1}{|G|} \sum_{g \in G} \chi_\lambda(g) \bra{\varphi}R(g)\ket{\varphi}.
    \end{align}
\end{fact}
\begin{proof}
Consider the measurement operators $\{M_\lambda = \ketbra{\lambda}{\lambda} \otimes I\}_{\lambda \in \Irr(G)}$. 
The probability of measuring $\lambda$ is given by 
\begin{align}
    P(\lambda) &= \bra{\varphi_F}M_\lambda^\dagger M_\lambda \ket{\varphi_F} \\
    \nonumber
    &= \frac{1}{|G|^2} \sum_{g^{(1)} \in G}\sum_{g^{(2)} \in G} \overline{\chi}_\lambda(g^{(1)}) \chi_\lambda(g^{(2)}) \bra{\varphi}R^\dagger(g^{(1)})R(g^{(2)})\ket{\varphi}\\
    \nonumber
    &= \frac{1}{|G|^2} \sum_{g^{(1)} \in G}\sum_{g^{(2)} \in G} \chi_\lambda((g^{(1)})^{-1}g^{(2)}) \bra{\varphi}R((g^{(1)})^{-1}g^{(2)})\ket{\varphi}.
\end{align}
Using the double sum on the group, we obtain \eqref{eq: lambda probability}. 
Indeed, each $g^{(3)} \in G$ is the result of $|G|$ different multiplications $(g^{(1)})^{-1}g^{(2)}$ (Hint: picture the $|G| \times |G|$ matrix whose entries are group multiplications).
\end{proof}

This formula shows that Fourier sampling always returns an element of $H^\perp$. 

\begin{fact}[Support of $P$]
\label{fact:supportP}
    The probability distribution $P$ is only supported on $H^\perp$. In other words, $P(\lambda) = 0$ for all $\lambda \notin H^\perp$, and $P(H^\perp) = \sum_{\lambda \in H^\perp} P(\lambda)  = 1$.
\end{fact}
\begin{proof}
    We can rewrite \eqref{eq: lambda probability} as
    \begin{align}
        P(\lambda) &= \frac{|H|}{|G|} \sum_{c \in G/H} \frac{1}{|H|}\sum_{h \in H} \chi_\lambda(ch) \bra{\varphi}R(ch)\ket{\varphi}\\
        \nonumber
        &= \frac{|H|}{|G|} \sum_{c \in G/H} \chi_\lambda(c) \bra{\varphi}R(c)\ket{\varphi} \frac{1}{|H|}\sum_{h \in H} \chi_\lambda(h).
         \nonumber
    \end{align}
    Then, $\sum_{h \in H} \chi_\lambda(h) = 0$ whenever $\lambda \notin H^\perp$ (\Cref{lemma: character orthogonality}).

    Since $P$ is a probability distribution that does not have support outside $H^\perp$, 
    it follows that $P(H^\perp) = \sum_{\lambda \in H^\perp} P(\lambda) = 1$.
\end{proof}

By~\cref{fact:supportP}, one can write 
\begin{align}
\label{eqn:varphi-lambda}
    \ket{\varphi_F} &= \frac{1}{|G|} \sum_{g \in G} \sum_{\lambda \in H^\perp} \chi_\lambda(g) \ket{\lambda} \otimes R(g) \ket{\varphi}\nonumber\\
    &=\frac{1}{|G|}\sum_{\lambda\in H^\perp} |\lambda\rangle\otimes (\sum_{g\in G}\chi_\lambda(g) R(g)|\varphi\rangle)\nonumber\\
    &=\sum_{\lambda\in H^\perp}  \sqrt{P(\lambda)}|\lambda\rangle\otimes |\varphi_\lambda\rangle,
\end{align}
where, for $P(\lambda)>0$, the normalized state $\ket{\varphi_\lambda}$ is $\frac{1}{\sqrt{P(\lambda)}}\frac{1}{|G|}\sum_{g\in G}\chi_\lambda(g) R(g)|\varphi\rangle$.

For a subset $A\subseteq H^\perp$, define
$P(A)
    \defeq
    \left\|
        \left(
            \sum_{\lambda\in A}\ketbra{\lambda}{\lambda}
            \otimes I
        \right)
        \ket{\varphi_F}
    \right\|^2.$
Equivalently, $P(A)$ is the probability that measuring the first register returns an element of $A$. 
We next show the anti-concentration of $P$, which is essential for amplitude amplification. 

\begin{lemma}[Anti-concentration of $P$~{\cite[Lemma 3]{hinsche2025povm}}] 
\label{lemma: anticoncentration}
For every  proper subgroup $K^\perp < H^\perp$, the probability of sampling an element from $H^\perp$ that is outside $K^\perp$ is at least $\epsilon/2$.
In formula, $P(H^\perp \setminus K^\perp) \geq \epsilon/2$.
\end{lemma}
\begin{proof}
We have $P(H^\perp) = P(K^\perp) + P(H^\perp \setminus K^\perp) = 1$.
Then,
\begin{align}
    P(H^\perp \setminus K^\perp) &= 1 - \sum_{\lambda \in K^\perp} P(\lambda)\\
    \nonumber
    &= 1 - \frac{1}{|G|} \sum_{g \in G} \left(\sum_{\lambda \in K^\perp} \chi_\lambda(g) \right) \bra{\varphi}R(g)\ket{\varphi}\\
     &= 1 - \frac{1}{|K|} \sum_{k \in K} \bra{\varphi}R(k)\ket{\varphi}.
      \nonumber
\end{align}
Here, we have used character orthogonality (\Cref{lemma: character orthogonality}) and $|K^\perp| |K| = |G|$ (\Cref{proposition: cardinality of perp}).
Since $K^\perp < H^\perp$, we have $K > H$ (\Cref{proposition: duality inequalities}). 
Then,
\begin{align}
    P(H^\perp \setminus K^\perp) &=  1 - \frac{1}{|K|} \Big( \sum_{h \in H} \bra{\varphi}R(h)\ket{\varphi} + \sum_{k \in K \setminus H} \bra{\varphi}R(k)\ket{\varphi} \Big)\\
     \nonumber
    &\geq  1 - \frac{1}{|K|} \Big( |H| + \sum_{k \in K \setminus H} |\bra{\varphi}R(k)\ket{\varphi}| \Big)\\
     \nonumber
    &\geq  1 - \frac{1}{|K|} \Big( |H| + (1-\epsilon)(|K| - |H|) \Big)\\
     \nonumber
    &=  1 - (1 - \epsilon) - \frac{1}{|K|} \Big( |H| - (1-\epsilon)|H| \Big)\\
     \nonumber
    &= \epsilon - \frac{|H|}{|K|}\epsilon \geq \frac{\epsilon}{2}.
\end{align}
Here we used $\bra{\varphi}R(h)\ket{\varphi}=1$, $|\bra{\varphi}R(k)\ket{\varphi}| \leq 1-\epsilon$, and $\frac{|H|}{|K|} \leq \frac{1}{2}$, which follows from Lagrange's theorem.
\end{proof}

The lemma guarantees that, whenever $K^\perp<H^\perp$, Fourier sampling returns an element outside the current span with probability at least $\epsilon/2$.
This lower bound is the key ingredient for the two $\textsc{HSP-Sampler}$ implementations developed below.

\subsection{Copy-based HSP-Samplers}

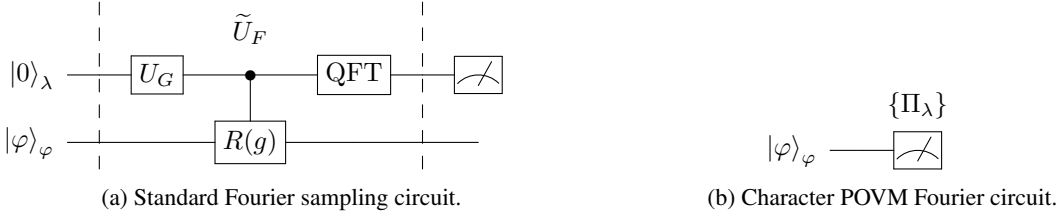
\begin{figure}[t]
    \centering

    \begin{subfigure}[b]{0.49\textwidth}
        \centering
        \scalebox{1}{ 
            \Qcircuit @C=1.2em @R=1.0em {
                \barrier[-0cm]{2}
                & 
                & 
                & \mbox{$\widetilde{U}_{F}$}
                & \barrier[-0cm]{2}
                & \\
              \lstick{\ket{0}_\lambda} 
                & \qw 
                & \gate{U_G} 
                & \ctrl{1}              
                & \gate{\mathrm{QFT}} 
                & \qw
                & \meter \\
              \lstick{\ket{\varphi}_\varphi} 
                & \qw
                & \qw
                & \gate{R(g)}
                & \qw 
                & \qw
                & \qw\\
            }
        }
        \caption{Standard Fourier sampling circuit.}
        \label{fig:fourier-sampling-circuit}
    \end{subfigure}
    \hfill
    \begin{subfigure}[b]{0.49\textwidth}
        \centering
        \scalebox{1}{ 
            \Qcircuit @C=1.2em @R=1.0em {
                \\ & & \{\Pi_\lambda\}\\
                \lstick{\ket{\varphi}_{\varphi}}
                & \qw & \meter
            }
        }
        \caption{Character POVM Fourier circuit.}
        \label{fig:character POVM}
    \end{subfigure}

    \caption{
        Two implementations of weak Fourier sampling.
        The circuit $\widetilde{U}_{F}$ in
        (\subref{fig:fourier-sampling-circuit}) uses a single copy of the input state to prepare $\ket{\varphi_F}
            =
            \sum_{\lambda\in H^\perp}
            \sqrt{P(\lambda)}
            \ket{\lambda}\otimes\ket{\varphi_\lambda}$.
        Alternatively, when an efficient implementation is available, the character POVM $\Pi_\lambda = \frac{1}{|G|}\sum_{g \in G} \chi_\lambda(g) R(g)$ can be applied directly, as described by \citet{hinsche2025povm} and illustrated in (\subref{fig:character POVM}). 
        This avoids the explicit controlled representation action and Fourier transform. 
        Both approaches sample from $P$.
        By \cref{lemma: anticoncentration}, for every proper subgroup $K^\perp<H^\perp$, the probability of obtaining $\lambda\in H^\perp\setminus K^\perp$ is at least $\epsilon/2$.
}
    \label{fig:copy-based-fourier-sampler}
\end{figure}
In this subsection, we construct an HSP-sampler using copies of the input state, without access to a state-preparation unitary, as in StateHSP's \cref{def: StateHSP copies}. 
The construction combines weak Fourier sampling, the anticoncentration bound from the previous subsection, and a classical procedure to test for membership in the subgroup generated so far.
Together with the adaptive stopping strategy of Algorithm~\ref{alg: HSP-Sampling framework}, this sampler improves the sample complexity of previous algorithms from $O\left(\tfrac{\log|G|+\log\frac{1}{\delta}}{\epsilon}\right)$~\cite{bouland2024state,hinsche2025povm} to $O\left(\tfrac{\log|G/H|+\log\frac{1}{\delta}}{\epsilon}\right)$, without requiring prior knowledge of $|H|$.
\cref{sec: sample lower bound} proves that this sample complexity is optimal, and together these results formalize the intuition that smaller hidden subgroups are harder to locate than larger hidden subgroups.

Algorithm~\ref{alg: copy-based HSP-Sampler} presents the copy-based $\textsc{HSP-Sampler}$, and the following theorem shows that the routine satisfies \cref{def: HSP-sampler}.

\begin{algorithm}[t]
  \caption{Copy-based HSP-Sampler}
  \label{alg: copy-based HSP-Sampler}
  \DontPrintSemicolon          
  \SetAlgoLined                
  \SetKwInOut{Input}{Input}
  \SetKwInOut{Output}{Output}
  \BlankLine          

 \Input{An instance of StateHSP over a finite abelian group $G$, with gap $\epsilon$, and with access to copies of the input state vector $\ket{\varphi}$; a generating set $S^\perp$ for a subgroup $K^\perp \leq H^\perp$;
  a failure probability $\tilde{\delta} \in (0,1]$.}
  \Output{If $K^\perp=H^\perp$, it outputs $\mathsf{END}$.
    If $K^\perp < H^\perp$, it outputs an element of $H^\perp\setminus K^\perp$ with probability at least $1-\tilde{\delta}$, and otherwise outputs $\mathsf{END}$.}

  Set $i \gets 0$; 

  \While{$i < \lceil\frac{2}{\epsilon}\ln(\frac{1}{\tilde{\delta}})\rceil$}{
      Set $s^\perp \gets$ weak Fourier sampling from $\ket{\varphi}_{\varphi}$;\;
      \If{$s^\perp\notin \mspan(S^\perp)$}
         {\Return{$s^\perp$.}}
      
      Set $i \gets i+1$;\;
  }
  
  \Return{$\mathsf{END}$.}
\end{algorithm}

\begin{theorem}[Copy-based $\textsc{HSP-Sampler}$]
\label{theorem: copy hsp-sampler}
    Consider an abelian StateHSP instance with gap parameter $\epsilon\in(0,1]$, access to copies of $|\varphi\rangle$, and hidden subgroup $H$. 
    Let $S^\perp$ be a data structure encoding a generating set for a subgroup $K^\perp \leq H^\perp$, and let $\tilde{\delta} \in (0,1]$. 

    Then, Algorithm~\ref{alg: copy-based HSP-Sampler} implements an $\textsc{HSP-Sampler}(S^\perp, \tilde{\delta})$ using at most $m=\lceil\frac{2}{\epsilon}\ln \frac{1}{\tilde{\delta} }\rceil$ calls to a Fourier sampling routine and at most $m$ classical membership tests.
    Each test checks whether the sampled element $s^\perp$ belongs to $\langle S^\perp\rangle$, or equivalently, whether $\langle S^\perp\cup\{s^\perp\}\rangle=\langle S^\perp\rangle$.
\end{theorem}

\begin{proof}
    From~\cref{fact:supportP}, every weak Fourier sampling output $s^{\perp(i)}$ belongs to the dual subgroup $H^\perp$. 
    If $K^\perp=H^\perp$, every output $s^{\perp(i)}$ gets us $\langle S^\perp \cup \{s^\perp\} \rangle = \langle S^\perp \rangle$, so the sampler outputs $\mend$ with certainty.
   
   If $K^\perp<H^\perp$, from~\cref{lemma: anticoncentration}, the probability that $s^{\perp(i)}\in K^\perp$ is at most $1-\epsilon/2$. The probability that $s^{\perp(i)}\in K^\perp$ for all $i\in[m]$ is then 
    \begin{equation}
        P\leq (1-\frac{\epsilon}{2})^m\leq e^{-\frac{\epsilon}{2}m}.
    \end{equation}
    Therefore, choosing $m=\lceil\frac{2}{\epsilon}\ln(\frac{1}{\tilde{\delta} })\rceil$ leads to $P\leq \tilde{\delta} $. 
    That is, the probability of outputting a $s^\perp \in H^\perp\setminus K^\perp$ is at least $1-\tilde{\delta} $. 
\end{proof}

Each call to the Fourier sampling circuit in \cref{fig:fourier-sampling-circuit} consumes a single copy of $\ket{\varphi}$, while the membership test is entirely classical and requires no additional copies.
Taking $\tilde{\delta}=1/2$ and combining this $\textsc{HSP-Sampler}$ with the HSP-Sampling framework (Algorithm~\ref{alg: HSP-Sampling framework}) therefore gives the claimed sample complexity.
We discuss the cost of the classical membership tests alongside the flagging procedure for the amplification-based sampler, and include it in the overall cost analysis in~\cref{subsec: overall cost}.

Our construction also allows alternative implementations of the sampling step.
The approaches of \citet{bouland2024state} and \citet{hinsche2025povm} offer two possibilities:
\begin{enumerate}
    \item The algorithm of \citet[Algorithm 2 in the arXiv version]{bouland2024state} can be viewed as an $\textsc{HSP-Sampler}$ for $\mZ_2^n$.
    It applies a single weak Fourier sampling circuit collectively to $O(1/\epsilon)$ copies, followed by one span test.
    In contrast, our construction performs $O(1/\epsilon)$ single-copy sampling steps, each followed by a membership test.
    The collective approach reduces the number of classical tests at the cost of a wider quantum circuit.
    Our sequential approach requires less quantum space and may therefore be preferable when quantum memory is limited.
    
    The algorithm of \citet{bouland2024state} also adaptively updates the unitary $U_G$ used to prepare the group register. 
    The Fourier sampling analysis reviewed in \cref{sec: fourier sampling}, due to \citet{hinsche2025povm}, shows that this update is not necessary for our construction.
    \item \citet{hinsche2025povm} obtain the weak Fourier sampling distribution through a direct measurement called the \emph{Character POVM} (\cref{fig:character POVM}).
    This measurement avoids the explicit irrep register, controlled representation action, and Fourier transform.
    Whenever it admits an efficient implementation, it can replace the Fourier sampling circuit in \cref{fig:fourier-sampling-circuit}.
\end{enumerate}
Both prior works use some anticoncentration bounds to choose a fixed copy budget before the algorithm begins.
Our HSP-sampling framework instead adapts the budget as new generators are found.
Combined with classical membership testing, this stopping strategy replaces the dependence on $\log|G|$ with $\log|G/H|$, without prior knowledge of $|H|$ or additional $\polylog |G|$ overhead.

\subsection{Amplification-based HSP-Samplers}
In the StateHSP formulation of \cref{def: StateHSP with circuit}, we can apply the state-preparation unitary $U_\varphi$ and its inverse.
This allows us to prepare and unprepare the Fourier sampling state, and therefore use amplitude amplification to obtain new generators more efficiently.
The goal of the following subsections is to construct an amplification-based HSP-sampler satisfying \cref{def: HSP-sampler}. 
Specifically, given a generating set $S^\perp$, the sampler should either return an element in $H^\perp\setminus \langle S^\perp\rangle$, or output $\mathsf{END}$, with the required success guarantee.

The construction combines the Fourier sampling circuit $U_F$ from Eq.~(\ref{eq: fourier sampling circuit}), 
fixed-point amplitude amplification, and an efficient subspace identification procedure.
The role of the latter is to coherently mark Fourier labels outside the subgroup generated so far.  

We use a fixed-point amplitude amplification construction inspired by the block-encoding formulation of \citet{gilyen2019quantum}; see also the earlier construction of \citet{yoder2014fixed}. 
In contrast to standard amplitude amplification, the fixed-point version does not require an exact estimate of the success probability: it just requires a lower bound on the success amplitude and it avoids overshooting.
By \cref{lemma: anticoncentration}, whenever the current span is a proper subgroup of $H^\perp$, the success probability is at least $\epsilon/2$.
We can therefore use the amplitude lower bound $\sqrt{\epsilon/2}$ to obtain an $\textsc{HSP-Sampler}(S^\perp,\tilde{\delta})$ with cost $O\left(\frac{\log\mathopen{(} 1/\tilde{\delta} \mathclose{)}}{\sqrt{\epsilon}}\right)$.

We note that our technical construction of amplification differs from that of \citet{gilyen2019quantum} in two respects. 
First, we identify the good subspace through a flag qubit, without requiring a description of the actual good state.
Indeed, although $H^\perp$ is unknown, the Fourier sampling state is supported on it, so testing whether a label lies outside the known subgroup $\langle S^\perp\rangle$ suffices to identify the good subspace.
Second, our construction guarantees that, conditioned on measuring the flag qubit to be $0$, the data register is \emph{exactly} in the good state, rather than just close enough.
The flag lets us understand if the amplification succeeded, ensuring that every accepted sample belongs to $H^\perp\setminus\langle S^\perp\rangle$.
This modification is very important, as allowing samples outside $H^\perp$ would introduce errors into the constraints used by the classical postprocessing, leading to a noisy reconstruction problem related, in the Boolean case, to Learning Parity with Noise.

\begin{theorem}[Fixed-point amplitude amplification]
\label{theorem: fixed point amp amp}
    Let $U_0\ket{0}=\ket{\psi_0}$ and suppose that
    \begin{align}
        U_\perp \ket{\psi_0}\ket{0}_f
        =
        c\ket{\psi_G}\ket{0}_f
        +
        \sqrt{1-c^2}\ket{\psi_B}\ket{1}_f,
    \end{align}
    for normalized states $\ket{\psi_G},\ket{\psi_B}$ and $c\in \mathbb{R}$.
    For any $\gamma,\tilde{\delta}\in(0,1)$, there is a unitary circuit
    $\widetilde U_{\gamma,\tilde{\delta}}$ such that, whenever $|c|\ge \gamma$, measuring the
    flag qubit (the second register) of $\widetilde U_{\gamma,\tilde{\delta}}\ket{\psi_0}\ket{0}$ gives outcome $0$
    with probability at least $1-\tilde{\delta}$, and conditioned on this outcome the remaining
    state is $\ket{\psi_G}$. If $c=0$, the flag is measured to be $1$
    with probability one.
 
    The circuit  $\widetilde U_{\gamma,\tilde{\delta}}$ uses two auxiliary qubits and
    $O(\log\mathopen{(} 1/\tilde{\delta} \mathclose{)}/\gamma)$ calls to $U_0$, $U_0^\dagger$, $U_\perp$, and
    $U_\perp^\dagger$,
    plus $O(n \log\mathopen{(} 1/\tilde{\delta} \mathclose{)}/\gamma)$ elementary gates where $n$ is the number of qudits in $|\psi_0\rangle$. Computing the circuit requires $O\left(\frac{\log\mathopen{(} 1/\tilde{\delta} \mathclose{)}}{\gamma}~\polylog \frac{\log\mathopen{(} 1/\tilde{\delta} \mathclose{)}}{\gamma \tilde{\delta}}\right)$ classical binary operations. 
\end{theorem}

The construction and proof appear in \cref{apx: fixed point amp amp}, specifically in \cref{prop:amp-psi0} and its proof. 
The classical computation determines the phases used for the quantum singular value transformation (QSVT) function.
For fixed $\gamma$ and $\tilde{\delta}$, these phases need to be computed only once; see~\cref{prop:explicit-phase}.

In our application, the initial state $\ket{\psi_0}$ is the Fourier sampling state, together with any auxiliary registers required by the flagging procedure.
More precisely, if the procedure uses $a$ auxiliary qubits, we set $\ket{\psi_0} = \ket{0}_aU_F\ket{0}_{\lambda,\varphi} = \ket{0}_a\ket{\varphi_F}_{\lambda,\varphi}$ (\ref{eqn:varphi-lambda}).
Preparing this state costs exactly one forward query to the state-preparation unitary $U_\varphi$ and the ability to implement $U_R$ and the group QFT.  

It remains to construct the unitary $U_\perp$ that flags the good subspace, such that $|\psi_G\rangle\in H^\perp\setminus \langle S^\perp\rangle$ and $|\psi_B\rangle\in \langle S^\perp\rangle$.
We package this requirement into the following abstract primitive.

\begin{definition}[Subspace identifier]
\label{def: subspace identifier}
    Consider a StateHSP instance with hidden subgroup $H$ and dual subgroup $H^\perp\leq\widehat{G}$. 
    Let $S^\perp$ be a data structure encoding a generating set for a subgroup $K^\perp \leq H^\perp$.
    A \emph{subspace identifier} is a classical routine that, given $S^\perp$, outputs a positive integer $a$ and the circuit description of a unitary $U_\perp$ acting on the Fourier sampling state $|\varphi_F\rangle_{\lambda,\varphi}$, $a$ auxiliary qubits, and a flag qubit, as
    \begin{align}
    \label{eqn:subspace-identifier-def}
        U_\perp \ket{\varphi_F}_{\lambda,\varphi}\ket{0}_a\ket{0}_f
        =
        \sqrt{P(H^\perp\setminus K^\perp)}
        \ket{\psi_G}_{\lambda,\varphi,a}\ket{0}_f
        +
        \sqrt{P(K^\perp)}
        \ket{\psi_B}_{\lambda,\varphi,a}\ket{1}_f .
    \end{align}
    Here $\ket{\psi_G}$ and $\ket{\psi_B}$ are normalized states, and measuring the group irrep register of $\ket{\psi_G}$ returns an element of $H^\perp\setminus K^\perp$ with probability one.
\end{definition}

In practice, a subspace identifier isolates the elements in $H^\perp \setminus K^\perp$ by entangling the flag register in state $0$.
Using this primitive, we can implement an amplification-based HSP sampler.

\begin{algorithm}[t]
  \caption{Amplification-based HSP-sampler.}
  \label{alg: amplification-based-hsp-sampler}
  \DontPrintSemicolon
  \SetAlgoLined
  \SetKwInOut{Input}{Input}
  \SetKwInOut{Output}{Output}
  \BlankLine

  \Input{
  A StateHSP instance over a finite abelian group $G$, with gap $\epsilon$ and access to the state-preparation unitary and inverse; a generating set $S^\perp$ for a subgroup $K^\perp \leq H^\perp$;
  a failure probability $\tilde{\delta} \in (0,1]$; and a subspace identifier as in \cref{def: subspace identifier}.
  }
  \Output{
    If $K^\perp=H^\perp$, it outputs $\mathsf{END}$.
    If $K^\perp < H^\perp$, it outputs an element of $H^\perp\setminus K^\perp$ with probability at least $1-\tilde{\delta}$, and otherwise outputs $\mathsf{END}$.
  }
  
  Call the subspace identifier on input $S^\perp$ to obtain an integer $a$ and a unitary $U_\perp$;\;
  Construct the fixed-point amplitude amplification circuit $\widetilde{U}$ for the pair $U_\perp$ and $U_0 = (U_F \otimes I_a )$, with amplitude lower bound $\gamma=\sqrt{\epsilon/2}$ and precision $\tilde{\delta}$;\;
  Prepare $\ket{0}\ket{0}_f$ and apply $\widetilde{U}$;\;
  Measure the flag qubit;\;
  \If{the outcome is $1$}{
    \Return{$\mathsf{END}$}\;
  }
  Measure the group irrep register\;
  \Return{the measured element $s^\perp \in H^\perp \setminus K^\perp$;}\;
\end{algorithm}

\begin{theorem}[Amplification-based HSP-sampler]
\label{theorem: amplification hsp-sampler}
    Consider an abelian StateHSP instance with gap parameter $\epsilon \in (0, 1]$, hiding a subgroup $H$ with dual $H^\perp$. 
    Let $S^\perp$ be a data structure encoding a generating set for a subgroup $K^\perp \leq H^\perp$, and let $\tilde{\delta} \in (0,1)$.
    Assume access to a subspace identifier as in \cref{def: subspace identifier}.
    
    Then, Algorithm~\ref{alg: amplification-based-hsp-sampler} implements an $\textsc{HSP-Sampler}(S^\perp,\tilde{\delta})$.
    It makes one classical call to the subspace identifier to obtain a circuit for $U_\perp$.
    It then runs a quantum circuit using $O\left(\frac{\log\mathopen{(} 1/\tilde{\delta} \mathclose{)}}{\sqrt{\epsilon}}\right)$ queries to $U_\varphi^{\pm1}$, $U_R^{\pm1}$, $U_\perp^{\pm1}$, and the group $QFT^{\pm1}$, as well as other $O\left(\frac{(\log |G|+ n) \log\mathopen{(} 1/\tilde{\delta} \mathclose{)}}{\sqrt{\epsilon}}\right)$ elementary gates, where $n$ is the number of qudits of $\ket{\varphi}$. 
\end{theorem}

\begin{proof}
    The subspace identifier returns a unitary $U_\perp$ such that~\cref{eqn:subspace-identifier-def} holds, 
    where measuring the group register of $\ket{\psi_G}$ returns an element of $H^\perp\setminus K^\perp$ with probability one.
    The algorithm applies fixed-point amplitude amplification to this decomposition, using the state-preparation unitary $U_0=U_F\otimes I_a$, amplitude lower bound $\gamma=\sqrt{\epsilon/2}$, and failure parameter $\tilde{\delta}$.

    To see that this implements an $\textsc{HSP-Sampler}(S^\perp,\tilde{\delta})$, observe that:
    \begin{enumerate}
        \item If $K^\perp=H^\perp$, then $P(H^\perp\setminus K^\perp)=0$. By the zero-overlap guarantee of fixed-point amplitude amplification (\cref{theorem: fixed point amp amp}, $c=0$ case), the flag qubit is measured as $1$ with probability one, so the algorithm outputs $\mathsf{END}$.
        \item If $K^\perp < H^\perp$, then by the anticoncentration guarantee of \cref{lemma: anticoncentration}, $P(H^\perp\setminus K^\perp)\geq \epsilon/2$. Thus the good amplitude is at least $\sqrt{\epsilon/2}$, and we identify $\gamma=\sqrt{\epsilon/2}$. Fixed-point amplitude amplification with precision $\tilde{\delta}$ therefore makes the flag outcome $0$ occur with probability at least $1-\tilde{\delta}$. Conditioned on this outcome, the group register is supported on $H^\perp\setminus K^\perp$, so measuring it returns an element of $H^\perp\setminus K^\perp$.
    \end{enumerate}
    
    The resource bound follows directly from \cref{theorem: fixed point amp amp} and the decomposition of the Fourier sampling circuit $U_F$ into more elementary resources~(\ref{eq: fourier sampling circuit}). Implementing $U_\perp$ requires elementary gates, and we will specify the gate complexity in the following with the concrete construction. 
\end{proof}

\begin{figure}[t]
    \centering
    \scalebox{1}{ 
            \Qcircuit @C=1.2em @R=1.0em {
                \barrier[-0cm]{4}
                & 
                & 
                & \mbox{$U_{F}$}
                & \barrier[-0cm]{4}
                & 
                &
                & \\
              \lstick{\ket{0}_f} 
                & \qw 
                & \qw 
                & \qw 
                & \qw
                & \qw
                & \multigate{2}{U_\perp}
                & \qw
                & \meter \\
              \lstick{\ket{0}_a} 
                & \qw 
                & \qw 
                & \qw 
                & \qw
                & \qw
                & \ghost{U_\perp} 
                & \qw
                & \\
              \lstick{\ket{0}_\lambda} 
                & \qw 
                & \gate{U_G} 
                & \ctrl{1}              
                & \gate{\mathrm{QFT}} 
                & \qw
                & \ghost{U_\perp} 
                & \qw 
                & \meter \\
              \lstick{\ket{0}_\varphi} 
                & \qw
                & \gate{U_{\varphi}}
                & \gate{R(g)}
                & \qw 
                & \qw
                & \qw
                & \qw
            }
    }
    \caption{Circuit schematic for the amplification-based HSP-sampler.
    The block $U_F$ prepares the Fourier sampling state $\ket{\varphi_F} =\sum_{\lambda\in H^\perp}  \sqrt{P(\lambda)}|\lambda\rangle\otimes |\varphi_\lambda\rangle$ from the all-zero input.
    Using the auxiliary register, $U_\perp$ flags labels outside the current subgroup $K^\perp=\langle S^\perp\rangle$ with flag value $0$.
    When $K^\perp<H^\perp$, fixed-point amplitude amplification increases the probability of measuring this flag value.
    Conditioned on outcome $0$, measuring the irrep register returns an element of $H^\perp\setminus K^\perp$ with certainty.}
    \label{fig: amplification circuit base}
\end{figure}
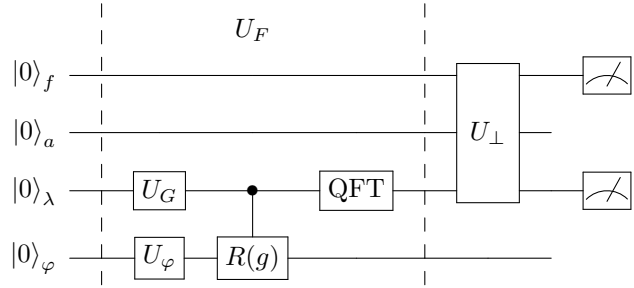

\Cref{fig: amplification circuit base} summarizes the circuit underlying our amplification-based HSP-sampler.

We present two subspace identification strategies. 
The first uses a system of linear equations to implement a membership test for $\langle S^\perp \rangle$ and flag labels in  $H^\perp\setminus K^{\perp}$; we develop this method in the next subsection. 
The second exploits a direct-sum decomposition $H^\perp=K^{\perp}\oplus Q$ to restrict the irrep register to labels in $Q$, with the $0$ label treated separately. This approach is closer to the classical membership tests discussed in \cref{sec:matrix-form-kernel}, and we describe it in \cref{app:second-approach}.  

\subsection{Subspace identifier}
A small generating set $S^\perp$ can span a large subgroup $K^\perp$, so our membership test must work directly with the generators, leveraging their information efficiently without enumerating the subgroup explicitly.
We construct a subspace identifier $U_\perp$ that performs this test by coherently checking a system of equations.
This is analogous to testing membership in a code specified by the generators $S^\perp$.
Specifically, we implement
\begin{align}
\label{eq: cpinot}
    U_\perp \ket{\lambda} \ket{0}_a\ket{0}_f = \begin{cases}
        \ket{\lambda} \ket{b}_a\ket{1}_f & \text{if }\lambda \in K^{\perp},\\
        \ket{\lambda} \ket{c}_a \ket{0}_f & \text{otherwise}.\\
    \end{cases}
\end{align}
Here, $\ket{0}_a$ is an ancillary register, and $\ket{b}_a$ and $\ket{c}_a$ are two normalized quantum states.
The last qubit is the flag that records the outcome of the membership test.

The main idea is to test membership in $K^\perp$ using generators of its annihilator $K=(K^\perp)^\perp$. 
Recall the definition of the dual subgroup $K^{\perp} = \{\lambda \in \hat{G}: \chi_\lambda(h) = 1,\forall h \in K\}$ and the group homomorphism property: $\forall g^{(1)},g^{(2)}\in G:\;\chi_\lambda(g^{(1)}+g^{(2)}) = \chi_\lambda(g^{(1)})\chi_\lambda(g^{(2)}).$ 
Then, for the membership test it suffices to test the character equations on a generating set $S=\{s^{(1)}, \dots, s^{(d)}\}$ for $K$:
\begin{align}
\label{eq: system check characters}
    \lambda \in  K^{\perp} \iff 
    \begin{cases}
        \chi_\lambda(s^{(1)}) = 1,\\
        \chi_\lambda(s^{(2)}) = 1,\\
        \vdots \\
        \chi_\lambda(s^{(d)}) = 1.
    \end{cases}
\end{align}
The quantum circuit $U_\perp$ checks these $d$ equations coherently and sets the flag to $1$ if all are satisfied.

We next explain how to compute a generating set $S$ for $K$ and implement these checks.
Along the way, we describe the data structure encoding $S^\perp$ and the associated $\textsc{Update}$ and $\textsc{Dual-Solver}$ routines.
We begin with the Boolean group $G=\mZ_2^n$, then extend the construction to general finite abelian groups.

\subsubsection{Bitstrings and parity checks}
Consider the Boolean group $G =\mZ_2^n$, with bitwise XOR as the group operation. 
We encode $S^\perp$ as a boolean matrix whose rows generate $K^\perp$. 
The data structure starts as an empty matrix.
Each call to $\textsc{Update}(S^\perp,s^\perp)$ incorporates the new generator and restores RREF using incremental Gauss-Jordan elimination (\cref{theorem: gauss-jordan}).

For the $\textsc{Dual-Solver}$, recall that $\chi_\lambda(g) = (-1)^{\lambda \cdot g}$.
Thus, $g$ belongs to $K=(K^\perp)^\perp$ exactly when it is orthogonal to every row of the encoding matrix.
In other words, $K$ is the kernel of this matrix over $\mZ_2$.
Accordingly, the routine $\textsc{Dual-Solver}(S^\perp)$ coincides with the kernel-finding procedure described in \cref{theorem: kernel generation}, and it outputs a generating set $S$ for $K$.

Given $S$, the conditions in Eq.~(\ref{eq: system check characters}) become parity checks
\begin{align}
\label{eq: system check binary}
    \lambda \in  K^{\perp} \iff 
    \sum_{l =1}^n\lambda_{l}s_l =0 \pmod 2 ,\quad \forall s\in S.
\end{align}
Armed with our data structure for $S^\perp$ and the $\textsc{Dual-Solver}$, we can use (\ref{eq: system check binary}) to implement $U_\perp$. 

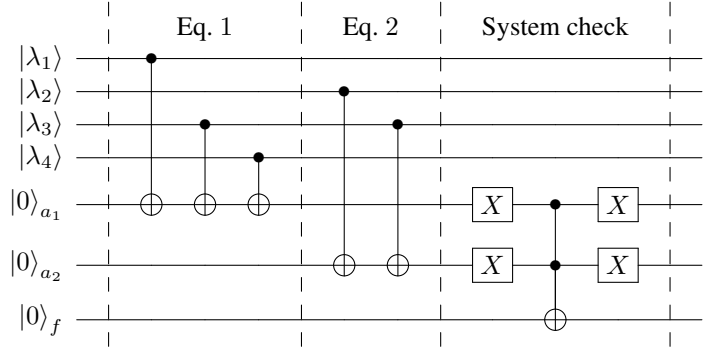
\begin{figure}[t]
    \centering
    \scalebox{1}{ 
            \Qcircuit @C=1.2em @R=1.0em {
                   \barrier[-0cm]{7}
                &  
                &  
                &  \mbox{Eq.~$1$}
                &  \barrier[-0cm]{7}
                &  
                &  \mbox{\qquad Eq.~$2$}
                &  \barrier[-0cm]{7}
                &
                &  
                &  \mbox{System check} 
                &  \barrier[-0cm]{7}
                &
                &
                \\
              \lstick{\ket{\lambda_{1}}} 
                & \qw 
                & \ctrl{4}              
                & \qw 
                & \qw
                & \qw 
                & \qw
                & \qw
                & \qw
                & \qw
                & \qw
                & \qw
                & \qw
                & \qw \\
              \lstick{\ket{\lambda_{2}}} 
                & \qw 
                & \qw
                & \qw 
                & \qw
                & \qw
                & \ctrl{4}              
                & \qw
                & \qw
                & \qw
                & \qw
                & \qw
                & \qw
                & \qw \\
              \lstick{\ket{\lambda_{3}}} 
                & \qw 
                & \qw
                & \ctrl{2}              
                & \qw              
                & \qw
                & \qw
                & \ctrl{3}              
                & \qw
                & \qw
                & \qw
                & \qw
                & \qw 
                & \qw \\
              \lstick{\ket{\lambda_{4}}} 
                & \qw 
                & \qw
                & \qw
                & \ctrl{1}              
                & \qw
                & \qw
                & \qw
                & \qw
                & \qw
                & \qw
                & \qw
                & \qw 
                & \qw \\
              \lstick{\ket{0}_{a_1}} 
                & \qw 
                & \targ              
                & \targ              
                & \targ              
                & \qw 
                & \qw 
                & \qw
                & \qw
                & \gate{X}
                & \ctrl{2}
                & \gate{X}
                & \qw 
                & \qw \\
              \lstick{\ket{0}_{a_2}} 
                & \qw 
                & \qw
                & \qw
                & \qw
                & \qw 
                & \targ               
                & \targ               
                & \qw
                & \gate{X}
                & \ctrl{1}
                & \gate{X}
                & \qw 
                & \qw \\
              \lstick{\ket{0}_f} 
                & \qw 
                & \qw
                & \qw
                & \qw
                & \qw 
                & \qw 
                & \qw 
                & \qw
                & \qw
                & \targ
                & \qw
                & \qw
                & \qw 
            }
    }
    \caption{Parity-check circuit for the binary subspace identifier of \cref{theorem: cpinot binary}. Here, $G=\mZ_2^4$ and $K$ is generated by $s^{(1)}=1011$ and $s^{(2)}=0110$. The first two blocks compute $\lambda\cdot s^{(1)}$ and $\lambda\cdot s^{(2)}$ modulo $2$ into auxiliary qubits $a_1$ and $a_2$. 
    The final block sets the initially zero flag qubit $f$ to $1$ if both parities are zero, identifying $\lambda\in K^\perp$ as in~\eqref{eq: system check binary}. The auxiliary qubits can be reset by undoing their $X$ gates and reversing the parity-check blocks.}
    \label{fig: quantum circuit system checker Z2}
\end{figure}

\begin{theorem}[Parity subspace identifier]
\label{theorem: cpinot binary}
    Let $S^\perp$ be an RREF matrix with $i$ rows, encoding the generators of $K^\perp$ in its rows.
    We can implement the $U_\perp$ unitary of Eq.~(\ref{eq: cpinot}) using $O((n-i)i)$ one- and two-qubit gates, and $O(n-i)$ ancillae.
    A classical computer can output a description of the circuit in $O((n-i)i)$ binary operations.
\end{theorem}
\begin{proof}
    Using the $\textsc{Dual-Solver}$ routine of \cref{theorem: kernel generation}, we can obtain a description of $S$ in sparse representation in $O((n-i)i)$ classical binary operations. 
    The description consists of $n-i$ group elements with $O(i)$ non-zero entries each.

    We can then implement $U_\perp$ through the $n-i$ equations from Eq.~(\ref{eq: system check binary}). We allocate one ancilla per equation and use it to evaluate the result of each individual parity check: for each generator $s$, if the $l$-th component is equal to $1$, meaning $s_l = 1$, we insert a $\mathrm{CNOT}$ controlled on $|\lambda_{l}\rangle$ and targeting the dedicated ancilla.
    Finally, we make sure the parity is satisfied for all the equations simultaneously by applying an $X$ gate to all the auxiliary systems and using an $(n-i)$-multicontrolled-$\mathrm{NOT}$ targeting the flag register.
    This way, the flag is equal to $1$ if and only if $\lambda \in K^\perp$. 
    \Cref{fig: quantum circuit system checker Z2} shows an example of such circuit.
    
    This circuit requires $O((n-i)i)$ $\mathrm{CNOT}$s, $(n-i)$ $X$ gates, and a multicontrolled-$\mathrm{NOT}$ that can be decomposed with $O(n-i)$ one and two qubit gates using one additional clean ancilla~\cite{khattar2025rise}.
    Optionally, one can uncompute all the auxiliary systems at the same asymptotic costs.
    A description of this circuit can be generated in linear $O((n-i)i)$ classical time. 
    Thus, the $\textsc{Dual-Solver}$ dominates the classical cost.
\end{proof}

\subsubsection{Uniform moduli and equations}
We extend the construction to abelian groups $G = \mZ_M^n$, which are products of $n$ cyclic groups of the same order $M$.
The implementations of $\textsc{Update}$ and $\textsc{Dual-Solver}$ depend on whether $M$ is prime or composite.
The subspace identifier $U_\perp$, however, uses the same construction in both cases.

As in the binary case, we encode $S^\perp$ in a matrix $A \in \mZ_{M}^{i \times n}$ whose rows generate $K^\perp$.
The characters are
\begin{align}
    \chi_\lambda(g) = e^{i2\pi (\sum_{l =1}^n \frac{\lambda_{l}g_{l}}{M})}.
\end{align}

The routine $\textsc{Update}(S^\perp,s^\perp)$ incorporates a new generator and restores the matrix to canonical form while preserving the enlarged row span.
When $M$ is prime $\mZ_M$ is a field, and we maintain RREF using incremental Gauss-Jordan elimination \cref{theorem: gauss-jordan}, as in the binary case. When $M$ is composite, $\mZ_M$ is only a ring, so we maintain the Howell normal form (\cref{def: howell normal form}) using the algorithm of \cref{theorem: howell algorithm}.
Unlike the RREF, we do not have an incremental procedure for the Howell normal form. 
This introduces a slight overhead in $n$ compared to the prime case. 

The $\textsc{Dual-Solver}(S^\perp)$ routine computes a generating set for $K = (K^\perp)^\perp$, which is the kernel of $A$ over $\mZ_M$.
For prime $M$, it uses the RREF kernel-finding procedure of \cref{theorem: kernel generation}.
For composite $M$, it uses the Howell-form procedure of \cref{theorem: kernel generation}.

Let $S = \{s^{(1)}, \dots, s^{(d)}\} \subset \mZ_M^n$ be the resulting generating set for $K=(K^\perp)^\perp$, as outputted by \textsc{Dual-Solver}($S^\perp$). For prime $M$, assuming $A$ has no zero rows, we have $d=n-i$.
For composite $M$, the routine returns a list of $n$ generators.
The membership test $\lambda \in K^{\perp}$ of Eq.~(\ref{eq: system check characters}) reduces to the system of equations
\begin{align}
\label{eq: system check uniform}
    \lambda \in  K^{\perp} \iff 
    \sum_{l =1}^n\lambda_l s_l^{(k)} \equiv 0 \pmod M,\quad \forall k\in\{1,2,\cdots,d\}.
\end{align}
These equations generalize the binary parity checks in~\eqref{eq: system check binary}.
We implement them coherently using modular arithmetic.

\begin{theorem}[Subspace identifier, $\mZ_M^n$]
\label{theorem: general system checker}
Let $A\in \mathbbm{Z}_M^{i\times n}$ be a matrix in RREF if $M$ is prime, or in Howell normal form if $M$ is composite, encoding the generators $S^\perp$ of $K^\perp$ in its rows. Let $S = \{s^{(1)}, \dots, s^{(d)}\} \subset \mZ_M^n$ be a set generators for $K \leq \mZ_M^n$, output by \textsc{Dual-Solver}($S^\perp$), which guarantees that each generator has at most $i+1$ nonzero elements. We can implement the subspace identifier unitary $U_\perp$ of Eq.~(\ref{eq: cpinot}) for the dual subgroup $K^{\perp}$ using $O(di)$ quantum arithmetic operations, and $O(d)$ $M$-level auxiliary systems.
A classical computer can output a description of the circuit in $O(dn~\polylog M)$ classical binary operations, or in $O(di~\polylog M)$ using sparse access to $S$.    
\end{theorem}

\begin{proof}
    Given the output $S$ from \textsc{Dual-Solver}($S^\perp$), 
    we need to implement the system of equations from Eq.~(\ref{eq: system check uniform}). 
    The construction is a generalization of the one in \cref{theorem: cpinot binary}. 
    We allocate $d$ $M$-level ancillary registers $\ket{0}_{a_1}\dots\ket{0}_{a_d}$ to 
    store the results of the $d$ equations in the system, and then use a 
    $U_{\mathrm{NOR}}$ circuit (\cref{def: NOR gate}) controlling on the equations and 
    targeting the flag qubit to ensure that all $d$ equations are satisfied.
 
    Let $\ket{\lambda} = \ket{\lambda_1}\ket{\lambda_2}\dots \ket{\lambda_n}$. 
    Consider the equation involving $s^{(k)}$. Since each component $s_l^{(k)}$ is a 
    classically known constant, the register $\ket{0}_{a_k}$ can accumulate the sum 
    $\sum_{l=1}^n \lambda_l s_l^{(k)} \bmod M$ directly: at the $l$-th term, 
    we apply the modular multiply-adder $U_{\mul(s_l^{(k)})}$ of 
    \cref{def: modular multiplier} to the pair of registers 
    $(\ket{\lambda_l}, \ket{\cdot}_{a_k})$, implementing
    \begin{align}
    \ket{\lambda_l}\,\big|\textstyle\sum_{l' < l}\lambda_{l'} s_{l'}^{(k)}\big\rangle_{a_k}
    \to \ket{\lambda_l}\,\big|\textstyle\sum_{l' \leq l}\lambda_{l'} s_{l'}^{(k)}\big\rangle_{a_k},
    \end{align}
    with all sums taken modulo $M$. Terms with $s_l^{(k)} = 0$ can be skipped. After the $U_\NOR$ circuit, the auxiliary systems can be uncomputed by running the multiply-adders in reverse with $-s_l^{(k)}$. 
    Since we have $O(d)$ equations and each has at most $i+1$ non-zero terms, we need $O(di)$ 
    quantum arithmetic operations.
    The classical complexity comes from reading the generators and outputting one 
    arithmetic operation per non-zero entry, each costing $O(\polylog M)$: reading 
    the generators as dense vectors costs $O(dn)$ entries, while with sparse access 
    only the $O(di)$ non-zero entries are read, giving the two stated bounds. Therefore, $\textsc{Dual-Solver}$ dominates the classical cost.
\end{proof}

\subsubsection{Mixed moduli}
Finally, we consider a general abelian group, $G=\mZ_{M_1}\times\cdots\times\mZ_{M_n}$, where $M_1,\dots,M_n\in\mZ_+$. 
We reduce the computations to the uniform composite moduli case.
The character of $G$ are 
\begin{align}
\label{eq: character mixed moduli}
    \chi_\lambda(g) = e^{i2\pi (\sum_{l =1}^n \frac{\lambda_{l}g_{l}}{M_l})}.
\end{align}
Let $M=\lcm(M_1, \dots, M_n)$ be the least common multiple of the cyclic orders and take $c_l = \frac{M}{M_l}$.  
Given a generating set $S = \{s^{(1)}, \dots, s^{(d)}\}$ for $K$, the membership conditions become
\begin{align}
\label{eq: system check mixed}
    \lambda \in  K^{\perp} \iff 
    \sum_{l =1}^nc_l\lambda_l s_l^{(k)} \equiv 0 \pmod M,\quad \forall k\in\{1,2,\cdots,d\},  
\end{align}
Each coordinate satisfies $\lambda_l\in[M_l]$, as $|\lambda\rangle=|\lambda_1\rangle|\lambda_2\rangle\cdots|\lambda_n\rangle$ is stored in $n$ $M_l$-level registers.
We evaluate these equations using $M$-level auxiliary accumulators.
The required multiply-adder (\cref{def: modular multiplier}) acts as
\begin{align}
\label{eqn:m-adder-mixed}
    \ket{\lambda_l}\,\big|\textstyle\sum_{l' < l} c_{l'}\lambda_{l'} s_{l'}^{(k)}\big\rangle_{a_k}
    \to \ket{\lambda_l}\,\big|\textstyle\sum_{l' \leq l} c_{l'} \lambda_{l'} s_{l'}^{(k)}\big\rangle_{a_k},
    \end{align}
with all sums taken modulo $M$.

As previous subsections, we first describe the $\textsc{Update}$ and $\textsc{Dual-Solver}$ that produce $S$, and then the subspace identifier $U_\perp$, which is essentially unchanged.
The $\textsc{Update}(S^\perp,s^\perp)$ and $\textsc{Dual-Solver}(S^\perp)$ both reduce to the composite uniform-moduli case by lifting the encoding matrix to $\mZ_M^n$: the $\textsc{Update}$ maintains the lifted matrix in Howell normal form (\cref{theorem: howell algorithm}), and the $\textsc{Dual-Solver}$ extracts a generating set for $K$ in $\mZ_M^n$ and then unlifts it. We formalize this below.
 
\begin{theorem}[\textsc{Dual-Solver}, mixed moduli $\mZ_{M_1} \times \dots \times \mZ_{M_n}$]
\label{theorem: dual solver mixed}
    Let $M_1, \dots, M_n \in \mZ_+$ and $G=\mZ_{M_1} \times \dots \times \mZ_{M_n}$. Given a generating set $S^\perp=\{s^{\perp(1)},\cdots s^{\perp(i)}\}$ of the subgroup $K^\perp$, one can obtain a set $S = \{s^{(1)}, \dots, s^{(d)}\} \subset \mZ_{M_1} \times \dots \times \mZ_{M_n}$ that generates $K$ with $O(n^3~\polylog M)$ classical binary operations, where $M=\lcm(M_1, \dots, M_n)$.
\end{theorem}
\begin{proof}
    We leverage Eq.~(\ref{eq: character mixed moduli}) and (\ref{eq: system check mixed}) to obtain the set of generators $S$.
    Using the set $S^\perp$ and the coefficients $c_l = \frac{M}{M_l}$ for $l \in \{1, \dots, n\}$, we define the encoding matrix
    \begin{align}
    \label{eqn:lift-A}
        A = \begin{bmatrix}
            c_1s_1^{\perp(1)} & c_2s_2^{\perp(1)} & \dots & c_ns_n^{{\perp(1)}}\\
            c_1s_1^{\perp(2)} & c_2s_2^{\perp(2)} & \dots & c_ns_n^{{\perp(2)}}\\
            \vdots & \vdots & \cdots & \vdots\\
            c_1s_1^{\perp(i)} & c_2s_2^{\perp(i)} & \dots & c_ns_n^{{\perp(i)}}
        \end{bmatrix} \in \mZ_M^{i \times n}.
    \end{align}
    The target subgroup $K$ is exactly the kernel of $A$ over $G$, $K = \{g \in G : Ag \equiv 0 \bmod M\}$. 
    However, we cannot directly invoke a kernel routine for matrix $A$: doing so would instead return the larger set $Q = \{x \in \mZ_M^n : Ax \equiv 0 \bmod M\} \supseteq K$, which includes solutions whose $l$-th component is not a valid element of $\mZ_{M_l}$.

    To resolve this, we lift the problem to $\mZ_M^n$ through the map $\phi: G \to \mZ_M^n$,
    \begin{align}
        \phi((g_1, \dots, g_n)) = (c_1g_1, \dots, c_ng_n).
    \end{align}
    This is an injective group homomorphism: for each coordinate, $c_l g_l \equiv c_l g'_l \bmod M \iff g_l \equiv g'_l \bmod M_l$ (since $c_l = M/M_l$), so $\phi$ is injective. Its image is the subgroup
    $
        \mathrm{im}(\phi) = \langle c_1 \rangle \times \dots \times \langle c_n \rangle, \abs{\langle c_l \rangle} = \frac{M}{\gcd(M, c_l)} = M_l $, 
    so that $\abs{\mathrm{im}(\phi)} = \prod_l M_l = \abs{G}$. On this image, $\phi$ is an isomorphism onto $G$ with inverse $\phi^{-1}((y_1, \dots, y_n)) = (y_1/c_1, \dots, y_n/c_n)$, where the division is exact since each coordinate is a multiple of $c_l$.

    Writing $y = \phi(g)$, the kernel condition $Ag \equiv 0$ becomes $\sum_l c_l s_l^{\perp(k)} g_l = \sum_l s_l^{\perp(k)} y_l \equiv 0 \bmod M$ for every $k$, which removes the factors $c_l$ from $A$. We therefore define
    \begin{align}
        A_0 = \begin{bmatrix}
            s_1^{\perp(1)} & s_2^{\perp(1)} & \dots & s_n^{{\perp(1)}}\\
            s_1^{\perp(2)} & s_2^{\perp(2)} & \dots & s_n^{{\perp(2)}}\\
            \vdots & \vdots & \cdots & \vdots\\
            s_1^{\perp(i)} & s_2^{\perp(i)} & \dots & s_n^{{\perp(i)}}
        \end{bmatrix} \in \mZ_M^{i \times n},
        \qquad
        \tilde{A} =
        \begin{bmatrix}
            A_0\\
            \mathrm{diag}(M_1, \dots, M_n)
        \end{bmatrix} \in \mZ_{M}^{(i+n) \times n}.
    \end{align}
    The appended block $\mathrm{diag}(M_1, \dots, M_n)$ enforces membership in $\mathrm{im}(\phi)$: the row for coordinate $l$ imposes $M_l y_l \equiv 0 \bmod M$, which is equivalent to $c_l \mid y_l$, i.e. $y_l \in \langle c_l \rangle$. Denote
    \begin{align}
        \tilde{K} = \{y \in \mZ_M^n : \tilde{A}y \equiv 0 \bmod M\} = \{y \in \mathrm{im}(\phi) : A_0 y \equiv 0 \bmod M\} = \phi(K).
    \end{align}
    We compute a generating set for $\tilde{K}$ in $O(n^3~\polylog M)$ classical binary operations by bringing $\tilde{A}$ into Howell normal form (\cref{theorem: howell algorithm}) and solving for its kernel (\cref{theorem:kernel-M}). Since $\tilde{K} \subseteq \mathrm{im}(\phi)$ and $\phi^{-1}$ is a homomorphism on $\mathrm{im}(\phi)$, applying $\phi^{-1}$ to these generators yields a generating set $S$ for $\phi^{-1}(\tilde{K}) = K$.
\end{proof}
 
Finally, given the generating set $S$, the subspace identifier $U_\perp$ is implemented as in \cref{theorem: general system checker}, with the only modification that each term $\lambda_l s_l^{(k)}$ in the equations is weighted by $c_l$ (following \cref{eq: system check mixed} and~\cref{eqn:m-adder-mixed}). This uses $O(n^2)$ quantum arithmetic operations and $O(n)$ $M$-level ancillae, and a classical computer can output a description of the circuit in $O(n^2~\polylog M)$ binary operations.

Similarly, membership testing and \textsc{Update} reduce to the uniform-modulus $\mZ_M^n$ case through the embedding $\phi$.
That is, using the encoding matrix of ~\cref{eqn:lift-A} we have $\mspan(A)=\phi(K^\perp)$, and we can lift any new element $s^\perp$ to $\phi(s^\perp)$ to perform a membership test or an update. 
This is because $\phi$ is bijective in the domain of interest (that is, the image of $\phi$), and the operations (addition and scalar multiplication) used in membership testing and \textsc{Update} preserve $\mathrm{im}(\phi)$.
 
\begin{table}[t]
    \centering
    \resizebox{\linewidth}{!}{%
    \begin{tabular}{|c|c|c|c|c|c|}
    \hline
                  & Arithmetic ops & Ancillae & Ancillae levels & Classical ops & Classical ops (Sparse) \\
                  \hline
        $\mZ_2^n$ & $O((n-i)i)$ & $O(n-i)$ & $2 $ & $O((n-i)n)$ & $O((n-i)i)$\\ \hline
        $\mZ_M^n$, $M$ prime & $O((n-i)i)$ & $O(n-i)$ & $M$ & $O((n-i)n~\polylog M)$ & $O((n-i)i~\polylog M)$\\ \hline
        $\mZ_M^n$, $M$ composite & $O(ni)$ & $O(n)$ & $M$ & $O((ni^2+n^2)~\polylog M)$ & $O(ni^2~\polylog M)$\\ \hline
        $\mZ_{M_1} \times \dots \times \mZ_{M_n}$
        & $O(n^2)$ & $O(n)$ & $M=\lcm(M_1, \dots, M_n)$ & $O(n^3~\polylog M)$ & $O(n^3~\polylog M)$ \\ \hline
    \end{tabular}
    }
    \caption{Cost of implementing $U_\perp$ given a matrix $A\in\mZ_M^{i\times n}$ in RREF or Howell normal form that encodes the generators $S^\perp$ of $K^\perp$.}
    \label{tab: subspace id costs}
\end{table}

\subsection{Overall costs}
\label{subsec: overall cost}
Combining all the ingredients, we can estimate the number of resources required to solve a \textsc{StateHSP} instance. 

\begin{theorem}[StateHSP with access to the state 
preparation unitaries]
\label{theorem: stateHSP unified}
    Consider a \textsc{StateHSP} instance as in \cref{def: StateHSP with circuit}, with $G = \mZ_{M_1} \times \dots \times \mZ_{M_n}$ and a hidden subgroup $H\le G$ of unknown size. Let $\delta\in(0,1]$. 
    Then, there exists an adaptive polynomial-time quantum algorithm that identifies $H$ with probability at least $1-\delta$, 
    using $O(\log|G/H| + \log \frac{1}{\delta})$
    circuits that each make $O(1/\sqrt{\epsilon})$
    queries to $U_\varphi^{\pm1}$, $U_R^{\pm1}$, and the group QFT. Therefore, the overall query complexity is
    \begin{equation}
        O\left(\frac{\log |G/H|  + \log \frac{1}{\delta}}{\sqrt{\epsilon}}\right). 
    \end{equation}
    
    Let $M=\lcm(M_1,\dots,M_n)$. Each circuit uses $O(n)$ $M$-level qudit ancillae and requires $O\left(n^2/\sqrt{\epsilon}\right)$ elementary quantum arithmetic operations, each realizable with $O(\polylog M)$ elementary one- and two-qubit
    gates. 
\end{theorem}
\begin{proof}
    We combine the \textsc{HSP-Sampling Framework} (\cref{theorem: HSP-sampling framework}) with the
    amplification-based \textsc{HSP-Sampler} (\cref{theorem: amplification hsp-sampler}, with $\tilde{\delta}=1/2$ in $\textsc{HSP-Sampler}$). The
    \textsc{Update} and \textsc{Dual-Solver} are instantiated according to the moduli: the RREF
    routines of \cref{theorem: gauss-jordan} and \cref{theorem: kernel generation} when $M$ is prime,
    and the Howell-normal-form routines of \cref{theorem: howell algorithm} and \cref{theorem:kernel-M}
    otherwise (the mixed case being lifted to $\mZ_M^n$ via the reduction of the mixed-moduli
    \textsc{Dual-Solver}). The subspace identifier $U_\perp$ is that of
    \cref{theorem: general system checker}, with the $c_l$-weighting of \cref{eq: system check mixed}
    in the mixed case. The classical cost is read from~\cref{tab:classical}, which is dominated by maintaining normal form in \textsc{Update}. 
\end{proof}

The classical cost for computing each circuit and postprocessing its output depends on the cyclic decomposition of $G$, and it is: (1) $O(n^2)$ classical arithmetic operations for uniform prime $M$ moduli; (2) $O(n^2\max(\log|G/H|,n))$ classical arithmetic operations for uniform composite $M$ moduli; or (3) $O(n^3)$ classical arithmetic operations otherwise, with $M=\lcm(M_1, \dots, M_n)$. 
Computing the amplification phases also requires a one-time cost of $O(\frac{1}{\sqrt{\epsilon}}~\polylog \frac{1}{\sqrt{\epsilon}})$ classical binary operations.   

The $O(n^2)$ bound on the arithmetic cost of the subspace identifier can be refined for particular groups.
\Cref{tab: subspace id costs} gives the more precise bounds.

\begin{theorem}[StateHSP with copy access]
\label{theorem: stateHSP copies z2}
    Consider a \textsc{StateHSP} instance as in \cref{def: StateHSP copies}, with $G=\mZ_{M_1} \times \dots \times \mZ_{M_n}$ and a hidden subgroup $H\le G$. Let $\delta\in(0,1]$ of unknown size. Then, there exists an adaptive polynomial-time quantum algorithm that identifies $H$ with probability at least $1-\delta$ using 
    \begin{equation}
        O\left(\frac{\log |G/H| + \log \frac{1}{\delta}}{\epsilon}\right)
    \end{equation}
    circuits, each requiring one copy of the input state vector $\ket{\varphi}$, one
    query to $U_R$, and at most two group QFTs. Each circuit uses $n$ qudit ancillae for the group irrep register, each one with $M_1, \dots, M_n$ levels respectively. 
\end{theorem}
\begin{proof}
    We just need to combine the $\textsc{HSP-Sampling Framework}$ of \cref{theorem: HSP-sampling framework} with the copy-based $\textsc{HSP-Sampler}$ of \cref{theorem: copy hsp-sampler} (with $\tilde{\delta}=1/2$ in $\textsc{HSP-Sampler}$ and the weak Fourier sampling circuit from \cref{fig:fourier-sampling-circuit}), choosing adequate $\textsc{Update}$ and $\textsc{Dual-Solver}$ routines, and membership testing to determine whether $s^\perp\in\mspan(S^\perp)$. When $M$ is prime, we take the RREF routines of \cref{theorem: gauss-jordan} and \cref{theorem: kernel generation} and membership testing in \cref{prop:membership-testing-RREF}. When $M$ is composite, we take the Howell Normal Form routines of \cref{theorem: howell algorithm} and \cref{theorem:kernel-M} and membership testing in \cref{prop:membership-testing-HNF}. 
\end{proof}

For copy access, the classical postprocessing has the same arithmetic cost bounds: (1) $O(n^2)$ extra classical arithmetic operations for uniform prime $M$ moduli; (2) $O(n^2\max(\log|G/H|,n))$ extra classical arithmetic operations for uniform composite $M$ moduli; or (3) $O(n^3)$ extra classical arithmetic operations on $\mZ_M$ otherwise, with $M=\lcm(M_1, \dots, M_n)$.  

\section{Query lower bound}
\label{sec: query lower bound}
In this section, we prove that the $O\left(\frac{\log |G/H|}{\sqrt{\epsilon}}\right)$ query complexity of \cref{thm:stateHSP-u} is optimal. 
Specifically, any quantum algorithm that solves the abelian StateHSP with success probability at least $2/3$ requires $\Omega\left(\frac{\log|G/H|}{\sqrt{\epsilon}}\right)$ queries to $U_\varphi$ and $U_\varphi^{-1}$ in the worst case.
Our construction shows that this lower bound holds even with access to their complex conjugates, $U_\varphi^*$ and $U_\varphi^T$, and to the controlled version of all these oracles.
This result settles the query complexity of abelian StateHSP and contributes to the broader study of the power of controlled-unitaries and conjugate queries~\cite{araujo2014quantum, tang2025controlled, RandomPurificationTang}.

Our construction proceeds in three steps. 
We start from the $\Omega(\log|G|)$ query lower bound for Simon's problem, which gives the lower bound at $\epsilon=1$. 
We then slightly modify Simon's oracle to obtain a \emph{padded} version of the problem that captures the dependence on $|H|$, yielding the finer bound $\Omega(\log|G/H|)$.
Finally, we embed this padded Simon's problem into a family of StateHSP instances parametrized by $\epsilon$ and use the adversary method to track how the query complexity grows as $\epsilon\to 0$. 
The crux is that, in the relative-$\gamma_2$-norm formulation of \citet{belovs2015variations}, the fractional version of the instance rescales the relevant oracle differences by $\sqrt{\epsilon}$.
This increases the lower bound by a factor of $1/\sqrt{\epsilon}$.

\subsection{Query complexity and the adversary method}
In this subsection, we review the quantum adversary method, which ties bounded-error query complexity to the solution of an optimization problem, the \emph{relative $\gamma_2$ norm}. 
We use this connection to establish the $\epsilon$ dependency. 
We first define the query model and bounded-error query complexity, then present the adversary method and the relative $\gamma_2$ norm.

\begin{definition}[Query model (\citet{belovs2015variations})]
\label{def: query model}
    Let $\mathcal{H}_S$ be a finite-dimensional query register. An algorithm with oracle
    access to a unitary $O$ on $\mathcal{H}_S$ acts on
    $\mathcal{H} \defeq  \mathcal{S}\oplus(\mathcal{H}_S \otimes \mathcal{W})$, where
    $\mathcal{W}$ is a finite-dimensional workspace and $\mathcal{S}$ is a finite-dimensional space. A forward query is
    $\widetilde{O} \defeq I_{\mathcal{S}} \oplus (O \otimes I_{\mathcal{W}})$ and an
    inverse query is
    $\widetilde{O}^{-1} = I_{\mathcal{S}}\oplus ( O^{\dagger} \otimes I_{\mathcal{W}})$.
    A \emph{$T$-query algorithm} is a unitary of the form
    \begin{align}
        \mathcal{A}_O \defeq
        U_T\,\widetilde{O}^{\pm1}\,U_{T-1}\cdots U_1\,\widetilde{O}^{\pm1}\,U_0 ,
    \end{align}
    where $U_0, \dots, U_T$ are unitaries on $\mathcal{H}$ independent of $O$, and the
    signs are fixed in advance.
\end{definition}

The direct sum structure in $\mathcal{H}$ implicitly allows the implementation of controlled-$O$, which is a stronger promise~\cite{araujo2014quantum}. For example, implementing $I\oplus O$ in the space of $\mathcal{K}\oplus\mathcal{K}$ is equivalent to implementing controlled-$O$ on $\mathbb{C}^2\otimes \mathcal{K}$. 
We prove our query lower bound in this stronger model.

\begin{remark}[Adaptivity]
Although the inter-query unitaries $U_0,\dots,U_T$ are fixed in advance and independent of the oracle, this model captures adaptive quantum algorithms whose later operations and queries depend on information obtained earlier.
By the deferred-measurement principle, intermediate measurements can be replaced by coherently recording their outcomes in workspace registers within $\mathcal{W}$. 
Subsequent classically controlled operations can be replaced by unitaries controlled on these registers. 
\end{remark}

We can now define bounded-error query complexity.

\begin{definition}[Bounded-error query complexity]
\label{def: bounded error query complexity}
    Let $\mathcal{X}$ be a finite set of instances, $\{O_x\}_{x \in \mathcal{X}}$ a family of
    unitaries on $\mathcal{H}_S$, and $s\colon \mathcal{X} \to E$ an answer function taking values in a finite set $E$. 
    A $T$-query algorithm $\mathcal{A}$ \emph{evaluates $s$ with bounded error} if there
    is a fixed input-independent initial state $\ket{\psi_0} \in \mathcal{H}$ such that, for every
    $x \in \mathcal{X}$, measuring a designated output register of
    $\mathcal{A}_{O_x}\ket{\psi_0}$ in the computational basis returns $s(x)$ with
    probability at least $2/3$. The bounded-error query complexity
    $Q\big(s, \{O_x\}\big)$ is the smallest such $T$.
\end{definition}

In our application, $\mathcal{X}$ is a set of Simon's functions, the oracle is the state preparation unitary $O=U_\varphi$, and $\mathcal{H}_S$ is the register on which $U_\varphi^{\pm 1}$ act.
Therefore, $T$ counts the calls to $U_\varphi$, $U_\varphi^{-1}$, and their controlled version.
We address conjugate queries, $U^*_\varphi$ and $U^T_\varphi$, later using a separate observation.

We now introduce the adversary method. 
Introduced by \citet{ambainis2002quantum}, this method lower-bounds quantum query complexity by tracking how much a single oracle call can help distinguish inputs with different answers. 
We use the reformulation of \citet{belovs2015variations}, which differs from the original in two helpful ways: it expresses the bound as a \emph{relative $\gamma_2$ norm}, which generalizes both the $\gamma_2$ norm~\cite{ambainis2002quantum} and the filtered $\gamma_2$ norm~\cite{lee2011quantum}; and it applies not only to Boolean oracles, but to arbitrary unitary oracles queried in both directions, which matches our model. 

For a matrix $A$ and a family of matrices $\Delta = (\Delta_{xy})_{x,y \in \mathcal{X}}$,
let $\gamma_2(A \mid \Delta)$ denote the relative $\gamma_2$ norm of $A$ with respect to $\Delta$.
This is the value of a semidefinite program that finds a minimum-cost factorization of each entry $A_{x,y}$ as an inner product through the matrix $\Delta_{xy}$. 
We use this norm as a black box, relying only on the following two statements.

The first states that the adversary bound characterizes bounded-error query complexity up to constant factors.
We write $\bo_{s(x) \neq s(y)}$ for
the $\mathcal{X} \times \mathcal{X}$ matrix whose $(x, y)$ entry $1$ if $s(x) \neq s(y)$ and $0$
otherwise. 

\begin{theorem}[Adversary characterization~{\cite[Theorem~37]{belovs2015variations}}]
\label{thm: belovs 37}
    There are universal constants $c_2>c_1>0 $ such that, for every
    $s\colon \mathcal{X} \to E$ and every family of unitaries
    $\{O_x\}_{x \in \mathcal{X}}$,
    \begin{align}
        c_1\,\gamma_2\!\left(\bo_{s(x) \neq s(y)} \;\middle|\; (O_x - O_y)_{x,y}\right)
        \;\leq\; Q\big(s, \{O_x\}\big) 
        \;\leq\;
        c_2\,\gamma_2\!\left(\bo_{s(x) \neq s(y)} \;\middle|\; (O_x - O_y)_{x,y}\right).
    \end{align}
\end{theorem}

Intuitively, inputs with different answers ($\bo_{s(x)\neq s(y)}=1$) must be distinguished \emph{through} the differences $O_x-O_y$ that a single query can see.
When these differences are small, many queries are needed. 
For this theorem, the query complexity and the relative $\gamma_2$ norm are essentially the same quantity: this lets us transport Simon's bound across our family. 

The second important statement is an elementary property of the norm.

\begin{fact}[Rescaling the oracle, {\cite[Proposition~6(d)]{belovs2015variations}}]
\label{fact: gamma2 scaling}
    For every nonzero scalar $c$, $\;\gamma_2(A\mid c\,\Delta)=\tfrac{1}{\abs{c}}\,\gamma_2(A\mid\Delta)$.
\end{fact}

Reducing the oracle differences by a factor $c$ makes the inputs proportionally harder to distinguish and scales the adversary bound by $1/\abs{c}$.
With $c=\sqrt{\epsilon}$, this is the tool that makes our bound $\epsilon$-dependent.

\subsection{Padding Simon's problem and the StateHSP reduction}
\label{sec: padding and embedding simons}

In this subsection, we introduce a padded version of Simon's problem and then reduce it to a hard family of abelian StateHSP instances. 
We call this family the \emph{Fractional Padded Simon's} StateHSP.

\subsubsection{Simon's problem}
We begin with Simon's problem. 
Our reduction will transfer its query lower bound to StateHSP.

Let $V=\mZ_2^m$, and consider a function $f\colon V\to V$ promised to hide a period $s=s(f)$, meaning that
\begin{align}
    f(x)=f(y)\iff x + y\in\{0,s\}.
\end{align}
The function is accessed through the standard Boolean oracle $O_f\ket{x}\ket{b}=\ket{x}\ket{b + f(x)}$, and the goal is to recover $s$. 
In the notation of the adversary method, $\mathcal{X}$ is the set of all valid Simon's functions $f$, $s(f)$ outputs the hidden period of $f$, and $E=V$ is the set of possible answers.

D. R. Simon introduced this problem and gave a time-efficient quantum algorithm that solves it using $O(m)$ queries to $O_f$~\cite{Simonproblem94}.
Koiran, Nesme, and Portier later proved that this query complexity is optimal~\cite{koiran2005quantum}.
Their proof establishes a lower bound for the decision version of the problem: determine whether $s(f)=0$.

\begin{proposition}[Simon's decision lower bound~{\cite[Theorem~1]{koiran2005quantum}}]
\label{prop: simon decision lb}
    Any quantum algorithm that uses queries to $O_f$ to decide whether $f\colon \mZ_2^m \to \mZ_2^m$ is a bijection or hides a nonzero period, with success probability at least $2/3$, requires $\Omega(m)$ queries.
\end{proposition}

We use this decision lower bound to derive a lower bound for a search problem. 
In this search version, $f$ is promised to hide a nonzero period $s(f)$, and the algorithm must
output it with probability at least $2/3$. 
The algorithm's behavior on inputs that violate this promise is undefined: on a bijective $f$, it may return any 
element of $\mZ_2^m$ or a symbol indicating that the promise fails.

\begin{proposition}[Simon's search lower bound]
\label{prop: simon lb}
    Let $f: \mZ_2^m \to \mZ_2^m$ hide a nonzero period $s(f)$. 
    Any quantum algorithm that queries $O_f$ to determine $s(f)$, with success probability at least $2/3$, requires $\Omega(m)$ queries.
\end{proposition}

\begin{proof}
    We reduce the decision problem to the search problem.  
    Let $\mathcal{A}$ solve the \emph{search} problem using $q$ queries. 
    We construct a decision algorithm $\widetilde{\mathcal{A}}$ that uses at most $q+2$ queries:
    \begin{enumerate}
        \item Run $\mathcal{A}$ on $f$ and denote its output by $\hat{s}$;
        \item If $\hat{s} \notin \mZ_2^m \setminus \{0\}$ (it is $0$ or a special symbol), output \emph{bijective};
        \item Otherwise, query $f$ at $0$ and $\hat{s}$. 
        Output \emph{non-bijective} if $f(0) = f(\hat{s})$, and \emph{bijective} otherwise.
    \end{enumerate}
    If $f$ is bijective, then $f(0)\neq f(\hat{s})$ for every nonzero $\hat{s}$. 
    Thus, $\widetilde{\mathcal{A}}$ outputs \emph{bijective} with certainty, regardless of whether $\mathcal{A}$ detects the promise violation.
    Instead, if $f$ hides a nonzero period $s(f)$, then $\mathcal{A}$ returns $\hat{s}=s(f)$ with probability at least $2/3$.
    In this case, $f(0) = f(\hat{s})$, so $\widetilde{\mathcal{A}}$ outputs \emph{non-bijective}.
    Therefore, $\widetilde{\mathcal{A}}$ is correct with probability at least $2/3$ on every valid decision instance. 
    By \cref{prop: simon decision lb}, $q + 2 \in \Omega(m)$, and hence
    $q \in \Omega(m)$.
\end{proof}

These query lower bounds hold even with access to inverse and controlled oracles.
Indeed, $O_f$ is self-inverse, so $O_f=O_f^{-1}=O_f^\dagger$, and its controlled version can be implemented using a single query.
To see this, introduce an auxiliary register in the $\ket{+}^m$ state.
Conditioned on the control qubit being $\ket{0}$, swap this register with the oracle's answer register.
Then apply $O_f$ and repeat the conditional swap.
Since a XOR with any string leaves $\ket{+}^m$ unchanged, this procedure acts as the identity when the control is $\ket{0}$ and applies $O_f$ when the control is $\ket{1}$.

This search version of Simon's problem is an instance of HSP with group $G=\mZ_2^m$ and hidden subgroup $H=\{0,s\}$.
Since $s\neq 0$, every instances has $|H|=2$. 
Thus, $\log_2|G/H|=m-1$ and $\log_2|G|=m$ are both $\Theta(m)$.
Using these instances directly would give an $\Omega(m)=\Omega(\log|G|)$ lower bound, but would not capture the dependence on $|H|$.
To obtain the finer bound $\Omega(\log|G/H|)$ with varying subgroup sizes, we modify the hard instances further.

\subsubsection{Padded Simon's problem}
To capture the dependence on $|H|$, we introduce a family of \emph{Padded Simon's} problems by adding a fixed, known subgroup to Simon's problem.

\begin{definition}[Padded Simon's problems]
\label{def: padded simons problem}
    Let $K=\mZ_2^r$ and $V=\mZ_2^m$, where $r\ge0$ and $m\geq 1$ are known integers. 
    Set $G=K\oplus V \cong \mZ_2^{n}$, with $n=r+m$. Let $f:V \to V$ be a Simon's function with nonzero hidden period $s(f)$, and define $\widetilde{f}:G \to V$ by $\widetilde{f}(k,v)= f(v)$, for $k \in K$, $v \in V$.
    Let this function be accessible via a quantum oracle
    \begin{align}
        O_{\widetilde{f}} \ket{k}_r\ket{v}_m\ket{b}_m \to \ket{k}_r\ket{v}_m\ket{b+ f(v)}_m.
    \end{align}
    The task is to recover the hidden subgroup of $\widetilde{f}$ using queries to $O_{\widetilde{f}}$.
\end{definition}

Here, $\oplus$ denotes the direct sum of subgroups, as introduced in~\cref{def:direct-sum-decomposition}. Since $\widetilde{f}(k,v) = \widetilde{f}(k',v')$ if and only if $v + v' \in \{0, s(f)\}$, the hidden subgroup is $\widetilde{H} \defeq K \oplus \langle s(f) \rangle$.
Thus, $|G| = 2^{n}$, $|\widetilde{H}| = 2^{r+1}$, and $|G/\widetilde{H}| = 2^{m-1}$.
This family lets us vary the subgroup size and refine the lower bound.

\begin{lemma}[Padded Simon's lower bound]
\label{lemma: padded simon}
    Any quantum algorithm that recovers the hidden subgroup $\widetilde{H}$ in a Padded Simon's problem with success probability at least $2/3$ requires $\Omega(m)$ queries to $O_{\widetilde{f}}$ in the worst case, even when $|\widetilde{H}|$ is known. 
\end{lemma}

\begin{proof}
    We reduce Simon's search problem to Padded Simon's problem. 
    Let $\mathcal{B}$ solve Padded Simon's problem using $q$ queries to $O_{\widetilde{f}}$, with
    success probability at least $2/3$. 
    We use $\mathcal{B}$ to solve Simon's search problem on $\mZ_2^m$ using $q$ queries to $O_f$. 

    Under the decomposition $G = K \oplus V$, we have $O_{\widetilde{f}} = I_K \otimes O_f$.
    Thus, each query to $O_{\widetilde{f}}$ can be simulated by a single query to $O_f$, leaving the $\ket{k}_r$ register untouched. 
    Running $\mathcal{B}$ therefore recovers $\widetilde{H}$ with probability at least $2/3$ using $q$ queries to $O_f$. 
    Since $K$ and $V$ are known and $\widetilde{H} \cap V = \langle s(f) \rangle$, classical postprocessing recovers $s(f)$ as the unique nonzero element of $\widetilde{H} \cap V$.
    This solves the Simon's search problem with $q$ queries and success probability at least $2/3$. 
    By \cref{prop: simon lb}, 
    $q \in \Omega(m)$.
    
    Finally, $|\widetilde{H}| = 2^{r+1}$ for every instance of the family, so
    knowing $|H|$ reveals nothing about $f$.
\end{proof}

Any algorithm for the abelian HSP must solve Padded Simon's problem and therefore
requires $\Omega(m) = \Omega(\log |G/H|)$ queries in the worst case, even when $|H|$ is known.
Note that $\log|G| = n$, whereas $\log|G/H| = m-1$: it is the padding that separates the two quantities.
For simplicity, we use $f$ to denote both the original and padded Simon's functions in the remainder, with the intended meaning clear from context.

\subsubsection{StateHSP reduction}
We now reduce Padded Simon's problem to a StateHSP instance with access to a state-preparation unitary. Consider the state
\begin{align}
\label{eq: simons statahsp initial}
    \ket{\varphi_f}_{n+m}=\frac{1}{\sqrt{2^n}}\sum_{k\in\mZ_2^r}\sum_{v\in\mZ_2^m}\ket{k}_r\ket{v}_m\ket{f(v)}_m,
\end{align}
which can be prepared using a single query to either Simon's oracle $O_{f}$ or the padded oracle $O_{\widetilde{f}}$. 
Let $G=\mZ_2^n$ act on the group registers through the left-regular shift representation, extended by the identity on all other registers,
\begin{align}
\label{eq: simons statahsp repre}
    R(g)\ket{x}_n=\ket{x+g}_n.
\end{align}
Equivalently, $R(k_1,v_1)\ket{k}_r\ket{v}_m= \ket{k + k_1}_r\ket{v + v_1}_m$. 

The representation $R$ is fixed and independent of the oracle, so only the state-preparation unitary requires an oracle query.
Indeed, $U_R = \sum_{g \in \mZ_2^n} \ketbra{g}{g} \otimes R(g)$ simply performs an out-of-place addition on $\mZ_2^n$ and requires no oracle queries. 
Taking $H:=\widetilde{H}=K\oplus\langle s(f)\rangle$, a direct computation gives $R(h)\ket{\varphi_f}=\ket{\varphi_f}$ for every $h\in H$ and $\bra{\varphi_f}R(g)\ket{\varphi_f}=0$ for every $g \in G \setminus H$. 
Thus, this is a StateHSP instance with $\epsilon=1$, and recovering $H$ determines the hidden period $s(f)$. 

Embedding an ordinary Simon's instance in the same way would give an $\Omega(\log|G|)$ lower bound by \cref{prop: simon decision lb}.
This already improves on the $\Omega\left(\frac{\log|G|}{\log\log|G|}\right)$ bound shown by \citet{bouland2024state} using the work of \citet{jones2025testing}. 
In comparison, embedding Padded Simon's problem and using \cref{lemma: padded simon} refines the bound further, capturing the dependence on $|H|$: we obtain an $\Omega(m)=\Omega(\log|G/H|)$ bound, even when $|H|$ is known.
We next establish the dependence on $\epsilon$.

\subsection{The adversary argument}
The padded construction gives the desired lower bound at $\epsilon=1$.
For general $\epsilon\in(0,1]$, we modify the input state to obtain the \emph{Fractional Padded Simon's} StateHSP family, then use the adversary method to derive the additional factor $1/\sqrt{\epsilon}$.
First, we show that solving these instances is equivalent to recovering the hidden period using a fractional oracle $O_f^\epsilon$.
Next, the rescaling property of the adversary bound relates the query complexity for $O_f^\epsilon$ to that for $O_f^1$.
Finally, we show that $O_f^1$ and the padded oracle $O_f$ are equivalent at unit query cost.
Following these connections backwards transfers the lower bound for Padded Simon's problem to StateHSP and incorporates the dependency on $\epsilon$.

We begin by defining the \emph{Fractional Padded Simon's} StateHSP family. 
All instances have group $G=\mZ_2^n= \mZ_2^r \oplus \mZ_2^m$ and the same representation $R$.
The input state depends on a Simon's function $f$ on the subgroup $\mZ_2^m$ with nonzero hidden period $s(f)$. 
The corresponding StateHSP instance hides the subgroup $H=\mZ_2^r\oplus \langle s(f)\rangle$. 

For $\epsilon\in(0,1]$, define the input state
\begin{align}
    \ket{\varphi_f^\epsilon}=\sqrt{1-\epsilon}\,\ket{0}\ket{u_G}+\sqrt{\epsilon}\,\ket{1}\ket{\varphi_f},
\end{align}
where $\ket{u_G}=\frac{1}{\sqrt{2^n}}\sum_{k\in\mZ_2^r}\sum_{v\in\mZ_2^m}\ket{k}\ket{v}\ket{0}$ is the uniform superposition over the group with the answer register set to zero, and $\ket{\varphi_f}$ is defined in Eq.~(\ref{eq: simons statahsp initial}). 
The representation $R(g)$ acts on the group register as in Eq.~(\ref{eq: simons statahsp repre}) and as the identity on all other registers. 
Since $\ket{u_G}$ is shift-invariant and the two branches have orthogonal flags, we have
\begin{enumerate}
    \item $R(h)|\varphi_f^\epsilon\rangle =|\varphi_f^\epsilon\rangle$ for every $h\in H$;
    \item $\langle\varphi_f^\epsilon|R(g)|\varphi_f^\epsilon\rangle =(1-\epsilon)\,\bra{u_G}R(g)\ket{u_G}+\epsilon\,\bra{\varphi_f}R(g)\ket{\varphi_f}=1-\epsilon$ for all the other $g \in G\setminus H$.
\end{enumerate}
Thus, $\{|\varphi_f^\epsilon\rangle\}_f$ is a family of StateHSP instances with tunable parameter $\epsilon$. 

We next show that solving these instances is equivalent to recovering $s(f)$ using a fractional oracle $O_f^\epsilon$.
The input state is prepared as $|\varphi_f^\epsilon\rangle = U_\varphi\ket{0} =O_f^\epsilon\,W\ket{0}$, where $W$ applies Hadamard gates to the group register, producing $\ket{u_G}$ with the flag and answer registers set to zero.
The fractional oracle is
\begin{align}
    O_f^\epsilon=\begin{bmatrix}\sqrt{1-\epsilon}\,I & -\sqrt{\epsilon}\,O_f\\ \sqrt{\epsilon}\,O_f & \sqrt{1-\epsilon}\,I\end{bmatrix},
\end{align}
written in block form with respect to the flag qubit.
Here, $O_f$ denotes the Padded Simon's oracle from \cref{def: padded simons problem}.
To \emph{define} a valid oracle, it suffices to verify that $O_f^\epsilon$ is unitary, and no actual implementation is needed.
Its unitarity follows from the fact that $O_f$ is a self-inverse unitary.

Because $W$ is fixed and independent of $f$, each query to $U_\varphi=O_f^\epsilon W$ or its inverse can be simulated using one query to $O_f^\epsilon$ or its inverse.
Moreover, recovering $H=\mZ_2^r\oplus\langle s(f)\rangle$ is equivalent to recovering $s(f)$.
It therefore suffices to lower-bound the number of queries to $O_f^\epsilon$ and its inverse needed to recover $s(f)$.

To relate the fractional oracle to Padded Simon's problem, we compare $O_f^\epsilon$ with its value at $\epsilon=1$.
At $\epsilon=1$, the oracle reduces to
\begin{align}
    O_f^1=\begin{bmatrix}0 & -O_f\\ O_f & 0\end{bmatrix}= \begin{bmatrix}
        0 & -1\\
        1 & 0
    \end{bmatrix}\otimes O_f.
\end{align}
Thus, $O_f^1$ differs from the padded Simon's oracle $O_f$ only by a fixed one-qubit gate, so a query to either can be simulated using one query to the other.

For any two Simon functions $f,g$ the diagonal blocks cancel in the oracle difference, giving
\begin{align}
\label{eq: oracle difference}
    O_f^\epsilon-O_g^\epsilon=\sqrt{\epsilon}\begin{bmatrix}0 & -(O_f-O_g)\\ O_f-O_g & 0\end{bmatrix}=\sqrt{\epsilon}\,\big(O_f^1-O_g^1\big).
\end{align}
We can now combine these ingredients to prove the query lower bound for abelian StateHSP.

\begin{theorem}[Query lower bound]
\label{thm: lower bound}
    For every $\epsilon \in (0,1]$ and integers $n\geq m>1$, there is a family of abelian StateHSP instances with $G=\mZ_2^n$ and known $\log|G/H|=m-1$ such that any quantum algorithm identifying the hidden subgroup with probability at least $2/3$ requires $\Omega\!\left(\frac{\log|G/H|}{\sqrt{\epsilon}}\right)$ oracle queries in the worst case.
    This bound holds with access to the state-preparation unitary $U_\varphi$ and its inverse $U_\varphi^{-1}$, their complex conjugates $U_\varphi^*$ and $U_\varphi^T$, and controlled version of all these oracles.
\end{theorem}
\begin{proof}
    We use the Fractional Padded Simon's family with $r= n-m$ and $G=\mZ_2^r \oplus \mZ_2^{m}$.
    Let $\mathcal{X}$ be the set of Simon's functions on $\mZ_2^{m}$ with nonzero hidden period, and let $s\colon\mathcal{X}\to\mZ_2^m\setminus{0}$ return that period.
    Every instance has hidden subgroup $H=\mZ_2^r\oplus\langle s(f)\rangle$, so $|H|=2^{r+1}$ and $\log|G/H|=m-1$. 

    We first establish the bounded-error query complexity at $\epsilon=1$ in the model of \cref{def: query model}.
    Both padded Simon's oracle $O_f$ and $O_f^1$ have inverses and controlled versions that can be simulated with a single forward query. 
    This allows us to determine their bounded-error query complexity in the right model.
    By \cref{lemma: padded simon} and the single-query equivalence of $O_f^1$ and $O_f$, identifying $s(f)$ from the oracle family $\{O_f^1\}$ has bounded-error query complexity $\Theta(m) = \Theta(\log|G/H|)$.
    
    By the characterization of \cref{thm: belovs 37}, the corresponding relative $\gamma_2$ norm satisfies
    \begin{align}
        \gamma_2\!\left(\bo_{s(f)\neq s(g)}\;\middle|\; (O_f^1-O_g^1)_{f, g}\right) \in \Theta(\log|G/H|).
    \end{align}
    Applying \cref{fact: gamma2 scaling} with $c=\sqrt{\epsilon}$ to the identity~\eqref{eq: oracle difference} gives
    \begin{align}
        \gamma_2\!\left(\bo_{s(f)\neq s(g)}\;\middle|\; (O_f^\epsilon-O_g^\epsilon)_{f, g}\right)
        =\frac{1}{\sqrt{\epsilon}}\,\gamma_2\!\left(\bo_{s(f)\neq s(g)}\;\middle|\;(O_f^1-O_g^1)_{f, g}\right)
        \in \Theta\!\left(\frac{\log|G/H|}{\sqrt{\epsilon}}\right).
    \end{align}
    Using \cref{thm: belovs 37} once more, we show that recovering $s(f)$ from $O_f^\epsilon$ has bounded-error query complexity $\Theta(\log(|G/H|)/\sqrt{\epsilon})$.
    
    Each query to $U_\varphi=O_f^\epsilon W$ or its inverse can be simulated using one query to $O_f^\epsilon$ or its inverse and fixed gates.
    The same simulation works for controlled queries by using controlled versions of the fixed gates $W$ and $W^\dagger$, together with one controlled query to $O_f^\epsilon$ or its inverse.
    Since solving the corresponding StateHSP instance allows us to recover $s(f)$ from $O_f^\epsilon$ at no additional query costs, the $\Theta$ bound implies an analogous $\Omega$ bound for the StateHSP problem. 
    This bound already allows controlled queries, as these are included in the model of \cref{def: query model}.

    Finally, $U_\varphi$ is real for every instance in this family, so $U_\varphi^*=U_\varphi$ and $U_\varphi^T=U_\varphi^\dagger$. 
    Conjugate queries, including their controlled versions, therefore provide no additional power on this worst case family.
\end{proof}

This lower bound matches the upper bound in \cref{theorem: stateHSP unified} up to constant factors, proving that our algorithm is query-optimal for any fixed success probability above $2/3$.
Remarkably, our algorithm requires neither conjugate nor controlled queries, while the lower bound holds even with these additional forms of access.

\section{Sample lower bound}
\label{sec: sample lower bound}
We conclude the bounds by showing that the $\sqrt{\epsilon}$ advantage genuinely requires access to the preparation unitary. 
In the \emph{sample} (or copies) model, the algorithm is given $t$ copies of $\ket{\varphi}$ (\cref{def: StateHSP copies}), rather than oracle access to a unitary preparing it (\cref{def: StateHSP with circuit}).  
In this model, the dependence on the gap parameter $\epsilon$ cannot be improved to $\sqrt{\epsilon}$: one must pay a full $1/\epsilon$ factor.

\begin{theorem}[Sample lower bound]
\label{theorem: copy lower bound}
    For any $\epsilon \in (0,1]$ and $r,m \in \N$, with $r \geq 0$ and $m \geq 1$, there is a family of abelian \textsc{StateHSP} instances on $G=\mZ_2^{r+2m}$, with known $\log|G/H|=m$, such that any algorithm that identifies the hidden subgroup with probability at least $2/3$ from copies of the input state vector $\ket{\varphi}$, even with collective measurements, requires $\Omega(\log(|G/H|)/\epsilon)$ copies.
\end{theorem}

The proof is an information-theoretic lower bound. 
Since any algorithm that succeeds with probability at least $2/3$ on every instance also succeeds with probability at least $2/3$ on average over any distribution on instances, it suffices to construct one hard distribution and prove an average-case lower bound. 
We build our hard distribution as follows.
Let $G=\mZ_2^{n}$, with $n=r+2m$, and consider the set of all subgroups $H\le G$ of size $|H|=2^{r+m}$, built from the union of the padding group $\mZ_2^r$ and a hidden $m$-dimensional subspace of $\mZ_2^{2m}$. 
The number of such subgroups is $N\geq 2^{m^2}$. 
Imagine an adversary picks $H$ uniformly at random from this set and, independently, a phase $\alpha$ uniformly from $[0,2\pi)$, and gives us $t$ copies of a StateHSP input state vector $\ket{\varphi^\epsilon_{H,\alpha}}$. 

Now, let $Z$ denote the classical output of the POVM that the entire algorithm implements on the $t$ copies. 
Under this setup both $H$ and $Z$ are random variables, and we analyze their joint distribution. 
Our task is to infer an instance of $H$ from an instance of $Z$. 
To bound the bounded-error sample complexity, we sandwich the mutual information $I(H;Z)$. 

First, Fano's inequality lower-bounds the information any successful algorithm must extract: if the algorithm identifies $H$ among $N$ equally likely candidates with constant success probability, then $I(H;Z)\in\Omega(\log N)=\Omega(m^2)$. Second, the Holevo bound upper-bounds the information contained in $t$ copies of the state; this is the main technical step, where we show $I(H;Z) \in O( t\epsilon m)$. 
The idea is as follows. 
The input states are constructed so that each copy $\ket{\varphi^\epsilon_{H,\alpha}}$ divides the Hilbert space into one $H$-informative branch, which can be accessed only with probability $\epsilon$, and an uninformative branch that carries most of the probability. 
After averaging over the hidden phase $\alpha$, the $t$-copy state vector $\ket{\varphi^\epsilon_{H,\alpha}}^{\otimes t}$ decomposes into orthogonal sectors indexed by the number $k$ of informative branches. The $k$-th sector occurs with binomial weight $p_k^\epsilon=\binom{t}{k}(1-\epsilon)^{t-k}\epsilon^k$ and carries at most $2km$ bits of information about $H$. 
Averaging over $k$ yields $I(H;Z)\leq \sum_{k=0}^t p_k^\epsilon\, 2 k m\leq 2t\epsilon m$, giving the desired upper bound. 
Finally, combining the lower and upper bounds forces $t=\Omega\!\left(\frac{m}{\epsilon}\right)$ for the algorithm to succeed with constant probability. 
Since $m=\log|G/H|$, any algorithm solving StateHSP requires $\Omega\!\left(\frac{\log|G/H|}{\epsilon}\right)$ copies of the input state.

The rest of the section makes this argument formal.
\Cref{sec: sample complexity and infotheory} introduces the copy-access model and recalls the information-theoretic tools. 
\Cref{sec: hard family of instances} defines the hard family of StateHSP instances. 
Finally, \cref{sec: infotheory argument} combines the Fano and Holevo estimates to prove \cref{theorem: copy lower bound}.

\subsection{Sample complexity and information-theoretic tools}
\label{sec: sample complexity and infotheory}
We operate in the following model. 

\begin{definition}[Bounded-error sample complexity]
\label{def: sample complexity}
    Let $\mathcal{X}$ be a finite set of instances, $\{|\varphi_x\rangle\}_{x\in\mathcal{X}}$ a family of normalized states in $\mathcal{H}_S$, and $s:\mathcal{X}\ra E$ an answer function with $E$ a finite set. Consider the full Hilbert space $\mathcal{H}:=\mathcal{H}_S^{\otimes t}\otimes \mathcal{W}$ where the qubits in workspace $\mathcal{W}$ are initialized to $|0\rangle$. 
    A $t$-sample algorithm \emph{evaluates $s$ with bounded error} if there is a fixed input-independent POVM on $\mathcal{H}$ that may act collectively on all $t$ copies at once ($|\varphi_x\rangle^{\otimes t}\otimes |0\rangle_{\mathcal{W}}$), such that the measurement of POVM returns $s(x)$ with probability at least 2/3 for any $x\in\mathcal{X}$. The \emph{bounded-error sample complexity} $\mathrm{S}(s)$ is the least $t$ for which such an
    algorithm exists.
\end{definition}

In the above definition, POVM can absorb arbitrary input-independent quantum operations into the final measurement. As in the query model of the previous section, this sample model already captures every \emph{adaptive} algorithm, one that measures part of its state, classically inspects the outcome, and proceeds accordingly. 
This follows from the deferred-measurement principle~\cite[Section 4.4]{NielsenChuang}, together with the fact that any classical control and input-independent quantum processing can be incorporated into the final POVM.
A lower bound proved against arbitrary collective measurements therefore binds every algorithm in this model.

We instantiate the model for the StateHSP of \cref{def: StateHSP copies}: the input state is $\ket{\varphi}$, the answer function returns the hidden subgroup $H$, and the algorithm additionally holds the representation $R$.
Because $R$ is fixed and independent of $H$, any use of it costs no copies and can be absorbed into the allowed channel and final POVM.

In the later section, our lower bound will be proved for a random instance drawn from a carefully chosen distribution, whereas the sample complexity concerns the algorithm's performance on the worst instance. 
The two are linked by the following elementary principle: if an algorithm does well on \emph{every} instance, it does well on a random one. Consequently, showing that some distribution defeats every efficient algorithm shows that some \emph{individual} instance does too.

\begin{fact}[Averaging principle]
\label{fact: averaging principle}
    Let $\mathcal{X}$ be a family of problem instances, and let $\mu$ be any probability distribution over $\mathcal{X}$. 
    If an algorithm fails with probability at most $\delta$ on every input $x\in\mathcal X$, then it fails with probability at most $\delta$ on an input drawn from $\mu$.
\end{fact}
\begin{proof}
    Let $X \sim \mu$ and let $F$ be the event that the algorithm fails on input $X$, where the probability is over both the draw of $X$ and the algorithm's internal randomness and measurements.
    By the law of total probability,
    \begin{align}
        \Pr[F]
        = \int_{x\in\mathcal{X}}
        \Pr[F\mid X=x]\; d\mu(x).
    \end{align}
    Since, by assumption, $\Pr[F\mid X=x]\le \delta$ for every $x\in\mathcal{X}$,
    it follows that $\Pr[F]
    \le
    \int_{x\in \mathcal X}\delta\; d\mu(x)
    =
    \delta.$
\end{proof}

To prove the bounds for sample complexity, we now introduce some information measures and corresponding inequalities. 
We adopt the convention that all logarithms are base $2$ and all entropies are in bits. 

For a probability distribution $p$ on a finite set, the \emph{Shannon entropy} is
\begin{align}
    \mathrm{H}(p)\defeq-\sum_{x}p(x)\log p(x).
\end{align}
For a random variable $X$ on a finite set with distribution $p_X$ we write
$\mathrm{H}(X)\defeq\mathrm{H}(p_X)$. 
For jointly distributed random variables $X,Y$ on finite sets, with joint distribution
$p_{XY}$ and conditionals $p_{X\mid Y}$, the \emph{conditional entropy} and the
\emph{mutual information} are
\begin{align}
\label{eq: mutual information}
    \mathrm{H}(X\mid Y)\defeq-\sum_{x,y}p_{XY}(x,y)\log (p_{X\mid Y}(x\mid y)),
    \qquad
    I(X;Y)\defeq\mathrm{H}(X)-\mathrm{H}(X\mid Y).
\end{align}
For a density operator $\rho$ on a finite-dimensional Hilbert space, with eigenvalues
$\{\lambda_i\}_{i=1}^d$ counted with multiplicity, the \emph{von Neumann entropy} is
the Shannon entropy of the spectrum,
\begin{align}
\label{eq:vN}
    S(\rho)\defeq-\Tr\!\big(\rho\log\rho\big)=-\sum_{i=1}^{d}\lambda_i\log\lambda_i.
\end{align}

Since the eigenvalues of $\rho$ form a probability distribution supported on $\operatorname{rank}(\rho)$ nonzero entries, their entropy is at most the entropy of the uniform distribution on that support.
Therefore, 
\begin{align}
\label{eq:vN-rank}
    S(\rho)\le \log \operatorname{rank}(\rho).
\end{align}
We also use the following standard bound on the entropy of a mixture.
\begin{theorem}[Entropy of a mixture {\cite[Theorem 11.10]{NielsenChuang}}]
\label{theorem: entropy of a mixture}
    Suppose $\rho =\sum_{x=1}^n p_x \rho_x$, where $\{p_x\}_{x=1}^n$ is a probability distribution, and the $\rho_x$ are density operators.
    Then,
    \begin{align}
        S(\rho) \leq \sum_{x=1}^n  p_x S(\rho_x) +\mathrm{H}(p)
    \end{align}
    with equality if and only if the states $\rho_x$ have pairwise orthogonal supports.
\end{theorem}

After defining these measures, we can state our main tools for the mutual information's lower and upper bounds: Fano's inequality and Holevo's bound.

\begin{fact}[Fano's inequality~{\cite[Theorem~2.10.1]{cover2006elements}}]
\label{fact:fano}
    Let $X$ be a random variable on a finite set $\mathcal{X}$, let $Y$ be any jointly distributed random variable, and let $\hat X=\hat{X}(Y)$ be any estimator of $X$ taking values in $\mathcal{X}$, with error probability $P_e\defeq\Pr [\hat{X}\neq X]$. Then
    \begin{align}
        \mathrm{H}(X\mid Y)\ \le\ \mathrm{H}(P_e)+P_e\log\big(|\mathcal X|-1\big)\ \le\ 1+P_e\log|\mathcal{X}| ,
    \end{align}
    where $\mathrm{H}(P_e)\defeq-P_e\log P_e-(1-P_e)\log(1-P_e)$ is the binary entropy, which is at most $1$.
\end{fact}

\begin{fact}[Holevo bound~{\cite{Holevo}; see also \cite[Theorem~12.1]{NielsenChuang}}]
\label{fact:holevo}
    Let $\mathcal{X}$ be a finite set, let $\{(p_x,\rho_x)\}_{x\in\mathcal X}$ be an ensemble of density operators and let $X$ be a random variable on $\mathcal{X}$ with $\Pr[X=x]=p_x$. Let $\{M_z\}_{z\in\mathcal{Z}}$ be any POVM, and let $Z$ be the outcome of
    applying it to $\rho_X$, i.e.\ $\Pr[X=x,\,Z=z]=p_x\Tr(M_z\rho_x)$ for all $x$. Writing $\bar\rho\defeq\sum_x p_x\rho_x$ for the average state, the mutual information between $X$ and outcome $Z$ obeys
    \begin{align}
        I(X;Z)\ \le\ \chi\defeq S(\bar\rho)-\sum_{x}p_x\,S(\rho_x).
    \end{align}
\end{fact}

The symbol $\chi$ is called the Holevo quantity. 

\subsection{A hard family of instances}
\label{sec: hard family of instances}

In this section, 
we introduce a hard family of instances. 
We separate the construction into two independent choices: first, the StateHSP instances themselves, and second, a distribution of instances that we tune for the proof.

\paragraph{The StateHSP instances.}
Let $G$ be a finite abelian group and $H\le G$ be a subgroup.
Work on $\bC^{|G|}\otimes\bC^2$, the tensor product of a \emph{group register} with orthonormal basis $\{\ket{x}:x\in G\}$ and a \emph{flag} qubit, and let $R$ be the left-regular representation, acting as a shift on the group register and trivially everywhere else,
\begin{align}
\label{eq: copy repr}
    R(g)\ket{x}=\ket{x+g}\qquad \forall x,g\in G.
\end{align}
For a phase $\alpha\in[0,2\pi)$ define the input state vector
\begin{align}
\label{eq: copy hard state}
    \ket{\varphi_{H,\alpha}^\epsilon}=\sqrt{1-\epsilon}\,\ket{u_G}\ket{0}+e^{i\alpha}\sqrt{\epsilon}\,\ket{H}\ket{1},
\end{align}
where $\ket{H}=\frac{1}{\sqrt{|H|}}\sum_{h\in H}\ket{h}$ is a uniform superposition over the subgroup and $\ket{u_G}=\frac{1}{\sqrt{|G|}}\sum_{g\in G}\ket{g}$ over the group.
We call $\ket{H}\ket{1}$ the \emph{informative} branch, the only part that depends on $H$, and $\ket{u_G}\ket{0}$ the \emph{uninformative} branch.
Following a similar argument to the one shown in \cref{sec: padding and embedding simons}, one can easily check that this family corresponds to valid StateHSP instances with hidden subgroup $H$ and gap parameter $\epsilon$.

\paragraph{The group and the distribution.}
We now create a padded version of the problem above and choose a distribution of instances.
Let $r, m \in \N$ with $r \geq 0$ and $m\geq 1$.
We take $G=\mZ_2^{n}=\mZ_2^r \oplus \mZ_2^{2m}$, so that $\log|G|=n = r+2m$.
Here, $\mZ_2^r$ will play the padding role.
Let $\mathcal{G}_m$ be the set of all $m$-dimensional subspaces of $\mZ_2^{2m}$, each of size $2^m$. 
We draw an instance from the distribution
\begin{align}
    \widetilde{H}\sim\mathrm{Unif}(\mathcal{G}_m),\qquad \alpha\sim\mathrm{Unif}[0,2\pi),
\end{align}
independently.
We then take the hidden subgroup $H \leq G$ to be of the form
\begin{align}
    H = \{(h_1, h_2): h_1 \in \mZ_2^r, h_2 \in \widetilde{H}\}.
\end{align}
Through the direct sum decomposition, we can write $\ket{u_G}=\ket{u_K}\ket{u_V}$ and $\ket{H}=\ket{u_K}|\widetilde{H}\rangle$, where $\ket{u_K}=\frac{1}{\sqrt{|\mZ_2^r|}}\sum_{k\in \mZ_2^r}\ket{k}$, $\ket{u_V}=\frac{1}{\sqrt{|\mZ_2^{2m}|}}\sum_{v\in \mZ_2^{2m}}\ket{v}$, and $|\widetilde{H}\rangle = \frac{1}{\sqrt{|\widetilde{H}|}}\sum_{h\in \widetilde{H}}\ket{h}$. 
The representation acts on the joint first two registers $\ket{g} = \ket{k}\ket{v}$, exactly as in the Padded Simon's case and as in Eq.~(\ref{eq: copy repr}).
For ease of notation, we equivalently say that $\widetilde{H}$ or $H$ are drawn from $\mathcal{G}_m$, where since the padding is independent from the subgroup choice and its presence can be deduced by the context.

The purpose of the random subgroup $\widetilde{H}$ is to make the candidate subgroup set large, while the random phase $\alpha$ is a technical device that will remove coherences between different copy-number sectors in the Holevo upper bound.

\begin{proposition}[Many candidate subgroups]
\label{prop: subspace count}
    There are $N \geq 2^{m^2}$ distinct $m$-dimensional subspaces of $\mZ_2^{2m}$.
\end{proposition}
\begin{proof}
    Every $m$-dimensional subspace is the row span of a unique $m \times 2m$ matrix in RREF, with $m$ pivot columns.
    Now consider only those RREF matrices of the form $(I_m~A)$, with $A \in \mZ_2^{m \times m}$. 
    Different choices of $A$ give different row spans, since the RREF matrix representing a subspace is unique.
    There are $2^{m^2}$ choices of $A$, and hence at least $2^{m^2}$ distinct $m$-dimensional subspaces.
\end{proof}
The exact count is $\Theta(2^{m^2})$, but the lower bound above is all we need.
Note that, for any candidate subgroup, $\log|G| = r+2m$, $\log|H|=r+m$, and $\log|G/H|= m$.

\subsection{The information-theoretic argument}
\label{sec: infotheory argument}

Throughout this subsection, $H$ and $\alpha$ are independent and drawn from the distribution fixed above.
For a fixed $t$-copy algorithm, let $Z$ denote its classical output under this random choice of the input. 
The algorithm is asked to identify $H$; the phase $\alpha$ is a nuisance parameter and is not reported.
We bound the same mutual information $I(H;Z)$ from below and from above using Fano's inequality and Holevo's bound, respectively.

\subsubsection{Lower bound}
\begin{proposition}[Information lower bound]
\label{prop:fano-lb}
    For every $m \in \N$ and $m \geq 1$, any $t$-copy algorithm that identifies the hidden subgroup with probability at least $2/3$ on every input state of the family satisfies
    \begin{align}
        I(H;Z)\ \ge\ \frac{m^2}{6}.
    \end{align}
\end{proposition}
\begin{proof}
    For the averaging principle of \cref{fact: averaging principle}, any algorithm succeeding with probability at least $2/3$ on every instance also succeeds with probability at least $2/3$ under the uniform distribution fixed above.
    Hence we can draw a hidden subgroup $H$ from the distribution and require the algorithm to succeed.
    
    We apply Fano's inequality (\cref{fact:fano}) with $X=H$, $Y=Z$, and $\hat{X}(Z) = Z$.
    Here, $H$ is uniform on the $N=|\mathcal{G}_m|$ candidate subgroups, $Z$ is the subgroup output by the algorithm, and the error event is therefore $Z \neq H$, whose probability is $P_e \leq 1/3$.
    Fano's inequality gives $\mathrm{H}(H\mid Z)
    \leq
    \mathrm H(P_e)+P_e\log(N-1)$.
    Since $H$ is uniform, $\mathrm H(H)=\log N$.
    Thus, by \eqref{eq: mutual information},
    \begin{align}
    \label{eq:fano-special}
        I(H;Z)
        \geq
        \log N-\mathrm H(P_e)-P_e\log(N-1).
    \end{align}
In order to proceed, 
    we first consider $m=1$. In this case, $\mathcal G_1$ consists of the three one-dimensional subspaces of $\mZ_2^2$, and hence $N=3$.
    Because $P_e\leq 1/3$, and both $\mathrm H(p)$ and $p\log(N-1)$ are increasing for $p\in[0,1/3]$, \eqref{eq:fano-special} implies
    \begin{align}
        I(H;Z)
        &\geq
        \log 3-\mathrm H(1/3)-\frac13\log 2
        =\frac13
        \geq\frac{m^2}{6}.
    \end{align}
    Suppose now that $m\geq 2$.
    Using $\mathrm H(P_e)\leq 1$ and $\log(N-1)\leq\log N$ in \eqref{eq:fano-special}, we obtain $I(H;Z)
    \geq
    (1-P_e)\log N-1$.
    By \cref{prop: subspace count}, $\log N\geq m^2$, while $1-P_e\geq 2/3$. Consequently,
    $I(H;Z)
        \geq
        \frac23m^2-1
        \geq
        \frac{m^2}{6},
    $ where the last inequality holds for every $m\geq2$.
    \end{proof}

Intuitively, the proposition says that identifying the hidden subgroup is information-theoretically expensive: since there are exponentially many candidate subgroups, the classical output of any successful bounded-error algorithm must share $\Omega(m^2)$ bits of information with the true subgroup label.

\subsubsection{Upper bound}
Next, we use the Holevo upper bound to show that $t$ copies of the state can provide at most $2t\epsilon m$ such bits, which is the heart of the argument.
We first decompose the $t$-copy state $(\ketbra{\varphi_{H,\alpha}^\epsilon}{\varphi_{H,\alpha}^\epsilon})^{\otimes t}$ along the number of informative branches it contains, then average over the phase, and finally apply the Holevo bound.

For $k\in\{0,1,\dots,t\}$ define the unit vector
\begin{align}
\label{eq: sector state}
    \ket{\phi_{H,k}^t}= \frac{1}{\sqrt{\binom{t}{k}}}\sum_{\substack{S\subseteq[t]\\|S|=k}}\ \bigotimes_{i\in[t]}
    \begin{cases}
        \ket{H}\ket{1}, & i\in S,\\
        \ket{u_G}\ket{0}, & i\notin S.
    \end{cases}
\end{align}
This state encodes the equal superposition over the $\binom{t}{k}$ ways to place $k$ informative branches among the $t$ copies; for instance $\ket{\phi_{H,0}^t}=(\ket{u_G}\ket{0})^{\otimes t}$ and $\ket{\phi_{H,t}^t}=(\ket{H}\ket{1})^{\otimes t}$.
For a fixed $H$, these vectors are orthonormal since distinct $k$ produce distinct flag Hamming weights. They decompose the $t$-copy state as
\begin{align}
\label{eq: copy expansion}
    \ket{\varphi_{H,\alpha}^\epsilon}^{\otimes t}=\sum_{k=0}^{t}\sqrt{p_k^\epsilon}\,e^{ik\alpha}\,\ket{\phi_{H,k}^t},
\end{align}
where $p_k^\epsilon=\binom{t}{k}(1-\epsilon)^{t-k}\epsilon^{k}$ is the $\mathrm{Binomial}(t,\epsilon)$ probability of $k$ informative branches.

Because $\alpha$ is drawn uniformly and is never reported, the state relevant for guessing $H$ is the phase average.
To see this more formally, let the algorithm induce the POVM $\{M_z\}$ on the $t$-copy input state. 
For fixed $H$ and $\alpha$, the probability that the algorithm outputs $z$ is $\Pr[Z=z\mid H,\alpha] = \Tr[M_z(\ketbra{\varphi_{H,\alpha}^\epsilon}{\varphi_{H,\alpha}^\epsilon})^{\otimes t}]$.
With $\alpha$ uniform on $[0,2\pi)$, the law of total probability and linearity of the trace give 
\begin{align}
\label{eq: prob int}
    \Pr[Z=z\mid H]
    =\frac{1}{2\pi}\int_0^{2\pi}\Pr[Z=z\mid H,\alpha]\,d\alpha
    =\Tr\!\big(M_z\,\sigma_H\big),    
\end{align}
where we have defined $\sigma_H\defeq\frac{1}{2\pi}\int_0^{2\pi}\big(\ketbra{\varphi_{H,\alpha}^\epsilon}{\varphi_{H,\alpha}^\epsilon}\big)^{\otimes t}\,d\alpha$.
Averaging the phase annihilates the cross terms of \eqref{eq: copy expansion}
and leaves the sector mixture
\begin{align}
\label{eq: sigma sectors}
    \sigma_H=\sum_{k=0}^{t}p_k^\epsilon\,\ketbra{\phi_{H,k}^t}{\phi_{H,k}^t} .
\end{align}
Thus $Z$ is the outcome of a fixed measurement of the ensemble $\{(1/N,\sigma_H)\}_{H\in\mathcal{G}_m}$, and $I(H;Z)$ is exactly the quantity bounded in \cref{prop:fano-lb}.

The dependence on $H$ inside sector $k$ is carried by the $k$ copies of $\ket{H}$ alone, which confines the sector average to a low-dimensional space.
The dimension of the space provides a bound for the rank, which in turn bounds the entropy of the sector average. 

\begin{lemma}[Sector entropy]
\label{lemma: sector entropy}
    Let $\bar\rho_k\defeq\frac1N\sum_{H\in\mathcal{G}_m}|\phi_{H,k}^t\rangle\langle\phi_{H,k}^t|$. Then $S(\bar\rho_k)\le 2km$.
\end{lemma}
\begin{proof}
    Let $W_k\colon(\bC^{2^{2m}})^{\otimes k}\to\Hh^{\otimes t}$ be the linear map that places its $k$ input registers, in order, into the informative slots, attaches the flag $\ket{1}$, fills the remaining slots with $\ket{u_G}\ket{0}$, and averages over the placements:
    \begin{align}
        W_k\big(|v^{(1)}\rangle\otimes\cdots\otimes|v^{(k)}\rangle\big)
        =\frac{1}{\sqrt{\binom{t}{k}}}\sum_{\substack{S\subseteq[t]\\|S|=k}}\ \bigotimes_{i\in[t]}
        \begin{cases}
            |u_K\rangle|v^{(j)}\rangle\ket{1}, & i\text{ is the }j\text{-th smallest element of }S,\\
            \ket{u_G}\ket{0}, & i\notin S.
        \end{cases}
    \end{align}
    The map is defined for any $k$ states $\ket{v} \in \bC^{2^{2m}}$ and is independent of $H$.
    Moreover, the different placements have orthogonal flag patterns, so $W_k$ is an isometry.
    By construction, $W_k|\widetilde{H}\rangle^{\otimes k}=|\phi_{H,k}^t\rangle$. 
    Therefore every vector $|\phi_{H,k}^t\rangle$ lies in the image of $W_k$, a subspace of $\Hh^{\otimes t}$ of dimension at most $\dim(\bC^{2^{2m}})^{\otimes k}=2^{2mk}$. The whole averaged operator $\bar{\rho}_k$ is also supported on that subspace.
    Hence $\operatorname{rank}\bar\rho_k\leq 2^{2mk}$, and the rank bound \eqref{eq:vN-rank} gives $S(\bar\rho_k)\le\log(2^{2mk})=2km$.
\end{proof}

We are now ready to use Holevo's bound to upper bound the mutual information.
\begin{proposition}[Information upper bound]
\label{prop:holevo-ub}
    Any algorithm using $t$ copies of the input state, under the distribution fixed above and with POVM $\{M_z\}$, satisfies $I(H;Z)\le 2t\epsilon m$.
\end{proposition}
\begin{proof}
    As shown above, after averaging over the hidden phase, the conditional law of the algorithm's output given $H$ is the same as the law obtained by measuring the state $\sigma_H$.
    Thus the random variables $H$ and $Z$ arise from the ensemble $\{(1/N,\sigma_H)\}_{H\in\mathcal{G}_m}$ and the POVM $\{M_z\}$. 
    By the Holevo bound, (\cref{fact:holevo}),
    \begin{align}
    \label{eq: chi decomposition}
        I(H;Z)\ \le\ \chi=S\!\Big(\tfrac1N\sum_{H\in\mathcal{G}_m}\sigma_H\Big)-\tfrac1N\sum_{H\in\mathcal{G}_m}S(\sigma_H).
    \end{align}
    For each $H$ the sectors are orthonormal, so $\sigma_H$ in \eqref{eq: sigma sectors} has eigenvalues $\{p_k^\epsilon\}_k$ and $S(\sigma_H)=\mathrm{H}(p^\epsilon)$, independent of $H$; hence the second term of \eqref{eq: chi decomposition} equals $\mathrm{H}(p^\epsilon)$. The average state is $\tfrac1N\sum_H\sigma_H=\sum_{k=0}^{t}p_k^\epsilon\,\bar\rho_k$, a mixture of the $\bar\rho_k$, which have pairwise orthogonal supports because the sectors are orthogonal. \Cref{theorem: entropy of a mixture}, with equality, then gives
    $S\!\Big(\tfrac1N\sum_{H\in\mathcal{G}_m}\sigma_H\Big)=\mathrm{H}(p^\epsilon)+\sum_{k=0}^{t}p_k^\epsilon\,S(\bar\rho_k)$.
    
    The term $\mathrm{H}(p^\epsilon)$ cancels in \eqref{eq: chi decomposition}. Using the sector entropy \cref{lemma: sector entropy} and the mean $\sum_{k}k\,p_k^\epsilon=t\epsilon$ of the $\mathrm{Binomial}(t,\epsilon)$ law, we conclude the bound $        \chi=\sum_{k=0}^{t}p_k^\epsilon\,S(\bar\rho_k)\ \le\ 2m\sum_{k=0}^{t}k\,p_k^\epsilon\ =\ 2t\epsilon m$.
\end{proof}

\subsubsection{Proof of the main theorem} 
Finally, we are ready to combine everything and conclude the proof of \cref{theorem: copy lower bound}.

\begin{proof}[Proof of \cref{theorem: copy lower bound}]
    Suppose that a $t$-copy algorithm identifies the hidden subgroup with probability at least $2/3$ on every input state in the hard family. 
    Draw $(H,\alpha)$ from the distribution fixed above, and let $Z$ be the algorithm's output.
    
    By \cref{prop:fano-lb}, the worst-case success guarantee implies $I(H;Z)\ge \frac{m^2}{6}$.
    By \cref{prop:holevo-ub}, the same mutual information satisfies   $I(H;Z)\le 2t\epsilon m.$
    Therefore,
    \begin{align}
        \frac{m^2}{6}\ \le\ I(H;Z)\ \le\ 2t\epsilon m.
    \end{align}
    Solving for $t$ gives $t \in \Omega(\frac{m}{\epsilon})$ as claimed, for any $r \geq 0$ and $m\geq 1$. 
    Since for our hard family we have $m = \log|G/H|$, and any algorithm solving StateHSP solves this problem too, we conclude our StateHSP copy bound $\Omega({\log(|G/H|)}/{\epsilon})$. 
\end{proof}

\section{Applications}
\label{sec: applications}
Our algorithmic results immediately improve several symmetry-learning problems previously formulated as instances of the abelian StateHSP~\cite{hinsche2025povm}. The improvements are conceptually simple. 
First, even with access only to copies, the dependence on $\log |G|$ is replaced by one on $\log |G/H|$, so larger hidden subgroups can be learned more efficiently.
Then, whenever the corresponding state-preparation circuit and its inverse are available, the $1/\epsilon$ dependence of Fourier sampling is replaced by $1/\sqrt{\epsilon}$. 
The examples below illustrate three different forms of hidden structure captured by the StateHSP framework: tensor-product structure, stabilizer symmetries, and spatial translation symmetries.

Fourier sampling can be implemented either using the controlled representation
\begin{align}
    U_R=\sum_{g\in G}|g\rangle\!\langle g|\otimes R(g)
\end{align}
together with a quantum Fourier transform, or by applying a unitary that
jointly diagonalizes the representation and extracting the corresponding
character label~\cite{hinsche2025povm}.
Both implementations produce the samples required by our algorithms.
When a state-preparation unitary and its inverse are available, a coherent implementation of either procedure, retaining auxiliary registers and extracting character labels reversibly, also supports our coherent-access algorithms. 
We discuss the implementation costs of $U_R^{\pm1}$ for each application below; these are additional to the stated resource counts.

\subsection{Locating unentanglement}
Understanding the entanglement structure of multipartite quantum states is a central problem in quantum information. 
While tomography of large systems is generally computationally too demanding, many physically relevant states possess some hidden tensor-product structure that makes them easier to study. 
Given a quantum state, or access to its state-preparation unitary, the task is to identify the partition across which the state factorizes, without prior knowledge of the decomposition.

More formally, consider an $n$-qudit state vector of the form
\begin{align}
|\phi\rangle = |\phi^{(1)}\rangle_{C_1}\otimes\cdots\otimes|\phi^{(m)}\rangle_{C_m}
\end{align}
where the partition $C_1\sqcup\cdots\sqcup C_m=[n]$ is unknown, and one would like to recover the factors without reconstructing the full state.
We assume that each factor is separated from states admitting a further tensor-product decomposition. 
Specifically, let $\mathcal P(C_k)$ denote the set of pure-state density operators that factorize across at least one nontrivial bipartition of the $C_k$ factor. 
For every $k$ with $|C_k|\geq 2$, we assume
\begin{align}
\min_{\rho\in\mathcal P(C_k)}
\frac{1}{2}
\norm{|\phi^{(k)}\rangle\langle\phi^{(k)}|-\rho}_1
\geq \widetilde{\varepsilon}>0.
\end{align}
In particular, this ensures that $C_1,\ldots,C_m$ form the finest partition across which the state factorizes.

As observed in Refs.~\cite{bouland2024state,hinsche2025povm}, this problem can be formulated as an abelian StateHSP on $\mZ_2^n$ (regardless of the qudit dimensions) by considering the action of local SWAP operators on two copies of the state. 
The input state is $\ket{\varphi} = \ket{\phi}\ket{\phi}$, and a bitstring $g \in \mZ_2^n$ acts through $R(g)= \prod_{i\in [n]}\operatorname{SWAP}_{i}^{g_i}$, where $\operatorname{SWAP}_{i}$ exchanges the $i$th qudit between two copies.
The controlled representation
\begin{align}
    U_R=\sum_{g\in\mathbb Z_2^n}|g\rangle\!\langle g|\otimes R(g)
\end{align}
uses $n$ controlled-SWAP gates acting in parallel on disjoint triples.
For fixed local dimension and suitable connectivity, this gives a
constant-depth implementation.

The hidden subgroup consists precisely of the bitstrings whose $1$s select a union of factors of the hidden partition.
Indeed, swapping all qudits in a factor $C_k$ exchanges two identical copies of $|\phi^{(k)}\rangle$, leaving $\ket{\phi}^{\otimes 2}$ unchanged.
Swapping any collection of whole factors therefore also preserves the state.
However, swapping only a subset of qudits within a factor induces a fidelity gap $\epsilon\in \Omega(\tilde\epsilon\,^2)$.
Thus, writing $\mathbbm{1}_{C_k}$ for the bitstring with $1$s exactly at the indices in $C_k$, we have $H=\operatorname{span}(\mathbbm{1}_{C_1},\ldots,\mathbbm{1}_{C_m})$ and a promise gap $\Omega(\tilde \epsilon^2)$.

Consequently, if a preparation unitary $U_\phi$ and its inverse are available, \cref{theorem: stateHSP unified} gives
\begin{align}
O\!\left(
\frac{(n-m)+\log\frac{1}{\delta}}
{\sqrt\epsilon}
\right)=O\!\left(
\frac{(n-m)+\log\frac{1}{\delta}}
{\tilde\epsilon}
\right)
\end{align}
queries to $U_\phi$ and $U_\phi^{-1}$, while \cref{theorem: stateHSP copies z2} gives $O\!\left(
\frac{(n-m)+\log \frac{1}{\delta}}
{\tilde\epsilon^2}
\right)$ copies of the input state $\ket{\phi}$.
The dependence on $n-m$, rather than on the total number of qubits, follows since the hidden subgroup has size $2^m$, so $\log |G/H|=n-m$. 
Thus query access improves the scaling with the promise gap, and, regardless of the access model, states with many product factors are easier to identify: the algorithm only has to learn the nontrivial small quotient structure that distinguishes the possible cuts.

\subsection{Learning stabilizer groups}
Another important symmetry-learning problem with applications in quantum error correction, state certification, and quantum verification is to recover the stabilizer group of an unknown quantum state. 
Given an $n$-qu$d$it state, one seeks to identify the subgroup of Pauli operators that leave the state invariant, without reconstructing the state itself. This problem has been shown to admit a reduction to the abelian StateHSP over the Pauli group~\cite{hinsche2025povm}.

For prime local dimension $d$, let $W_x$ denote the standard Weyl operators,
indexed by $x\in\mathbb Z_d^{2n}$, and define the phaseless stabilizer subgroup
\begin{align}
H=\{x\in\mathbb Z_d^{2n}:|\langle\phi|W_x|\phi\rangle|=1\}.
\end{align}
We assume that
\begin{align}
|\langle\phi|W_x|\phi\rangle|\leq 1-\epsilon
\qquad\forall x\notin H.
\end{align}
Following Ref.~\cite{hinsche2025povm}, the corresponding StateHSP instance is
\begin{align}
G=\mathbb Z_d^{2n},\qquad
|\varphi\rangle=|\phi\rangle^{\otimes D},\qquad
R(x)=W_x^{\otimes D},
\end{align}
where $D=d$ for odd $d$ and $D=4$ for $d=2$.
This tensor power removes the projective phases, making $R$ a representation
and ensuring that $R(h)|\varphi\rangle=|\varphi\rangle$ for every $h\in H$.
For $x\notin H$,
\begin{align}
|\langle\varphi|R(x)|\varphi\rangle|    
\leq(1-\epsilon)^D,
\end{align}
so a valid StateHSP promise gap is
\begin{align}
\eta_D=1-(1-\epsilon)^D =\Theta\!\left(\min\{1,D\epsilon\}\right).
\end{align}

The controlled representation $U_R$ can be implemented using $O(nD)$
controlled qudit shift and phase gates. Assuming all-to-all connectivity,
operations on different sites can be performed in parallel, giving
depth $O(D)$ in the standard qudit gate model.

Each preparation of $|\varphi\rangle$ requires $D$ calls to $U_\phi$.
Accounting for both this cost and the amplified gap,
Theorem~\ref{theorem: stateHSP unified} identifies $H$, with failure probability
at most $\delta$, using
\begin{align}
O\!\left(
\left(n\log d + \log\frac{1}{\delta}\right)\,
\max\!\left\{d,\sqrt{\frac{d}{\epsilon}}\right\}
\right)    
\end{align}
queries to $U_\phi$ and $U_\phi^{-1}$.
Here the factor $\max\{d,\sqrt{d/\epsilon}\}$ comes from
$D/\sqrt{\eta_D}$, the number of queries per StateHSP input over the square root of the promise gap. 
The corresponding copy-access factor is
$D/\eta_D=\Theta(\max\{d,1/\epsilon\})$.
The stabilizer phases are recovered by measuring the Weyl operators
corresponding to generators of $H$ on additional copies of $|\phi\rangle$.

Finally, writing $|H|=d^r$, commutativity of the stabilizers implies
$r\leq n$, and hence
\begin{align}
    \log|G/H|=(2n-r)\log d=\Theta(n\log d).
\end{align}
Thus, for this application the quotient-size refinement gives at most a constant-factor improvement in the group-size term. Finally, we note that Allcock et al.~\cite{allcock2025reconqueringbellsamplingqudits} solve this problem using $O(n/\epsilon)$ copies for prime $d$. Our algorithm uses coherent access to $U_\phi$ and $U_\phi^{-1}$, improving the dependence on $\epsilon$ for fixed $d$, but with worse dependence on $d$. These sample and query bounds do not directly compare total gate costs, which also depend on state preparation and how qudit operations are implemented. Our improvement applies to the StateHSP reduction used here, but more efficient reductions may be possible.

\subsection{Identifying translation symmetries}
A third example is the identification of hidden translational symmetries in quantum many-body states. Let $T$ denote the cyclic translation operator acting on an $n$-site ring, and suppose that the state satisfies
\begin{align}
    T^k|\psi\rangle=|\psi\rangle
\qquad\forall k\in H,    
\end{align}
for some unknown subgroup $H\leq\mathbb Z_n$, while translations outside $H$ have overlap at most $1-\epsilon$ with the state.

This is directly an abelian StateHSP with $G=\mathbb Z_n$ and representation $R(k)=T^k$. The controlled representation
\begin{align}
    U_R=\sum_{k=0}^{n-1}|k\rangle\!\langle k|\otimes T^k
\end{align}
can be implemented by writing $k$ in its binary representation and applying the corresponding controlled powers $T^{2^j}$. 
Each power is a permutation of the site registers requiring $O(n)$ SWAPs. 
Thus, assuming all-to-all connectivity, a straightforward implementation without additional auxiliary qubits uses $O(n\log n)$ controlled-SWAP gates and depth $O(n\log n)$~\cite{hinsche2025povm}.

Hence, given a preparation unitary and its inverse, our results identify the full subgroup of translation symmetries using
\begin{align}
    O\!\left(
    \frac{\log (n/p)+\log \frac{1}{\delta}}
    {\sqrt{\varepsilon}}
    \right)
\end{align}
queries to the state-preparation unitary and its inverse, where $p=|H|$ is the number of translations leaving the state invariant. 
Thus, once again, query access improves the dependency on $\epsilon$ quadratically, and, regardless of the access model, the complexity depends on the number of translations that leave the quantum state invariant, rather than simply on the system size.

These examples show that the access-model improvement established for StateHSP is not restricted to an abstract subgroup-identification problem. 
It directly transfers to learning problems in many-body physics, finding applications in tasks involving hidden entanglement structure, stabilizer symmetries, and spatial symmetries.

\section{Discussion and outlook}
\label{sec: outlook}

In this paper, we determined the optimal sample and query complexity of the abelian StateHSP.
For constant success probability, the worst-case bounds are $\Theta(\log(|G/H|) /\epsilon)$ copies and $\Theta(\log (|G/H|)/\sqrt{\epsilon})$ queries to the state-preparation unitary and its inverse.
The dependence on the quotient size is achieved without prior knowledge of $|H|$, while the quadratic improvement in $\epsilon$ comes from coherently amplifying Fourier samples that enlarge the subgroup generated so far.
The matching lower bounds hold even with collective measurements in the copy model and with conjugate and controlled preparation queries in the query model. 
Together, these results identify both the advantage of coherent access and its limits.

The applications expose an important distinction between solving StateHSP optimally and obtaining optimal algorithms for the learning problems it captures.
While the improvements to the upper bounds propagate through the reductions in \cref{sec: applications}, our hard instances do not arise from hidden cuts, stabilizer groups, or translation symmetries. 
Establishing matching lower bounds for these problems in each access model, or exploiting their additional structure to improve on the StateHSP reductions, remains a natural next step.

One connection of broader interest is to quantum pseudorandomness and pseudoentanglement.
Bouland, Giurgi\c{c}\u{a}-Tiron, and Wright observed that locating unentanglement rules out a simple recursive construction: placing smaller 
pseudorandom states on the two sides of a randomly hidden product cut~\cite{bouland2024state}.
The resulting tensor-product structure can be recovered efficiently from copies. 
When the preparation unitary and its inverse are accessible, our result gives a sharper quantitative bound: under the hidden-cut promise, the partition can be recovered with $O(\log (|G/H|)/\sqrt{\epsilon})$ coherent queries, where $\epsilon$ is the induced StateHSP gap.
This improves the gap dependence of the detection algorithm under stronger access.
Such tests constrain candidate pseudorandom constructions by detecting structure that does not exist in Haar-random states.

The access model itself leaves further questions.
If only forward queries to $U_\varphi$ are allowed, does the worst-case complexity return to $\Theta(\log(|G/H)/\epsilon)$? 
The necessity of inverses for general amplitude amplification~\cite{TangWriteInputModels} suggests this possibility, but proving it for StateHSP requires a separate lower bound.
On the other hand, conjugate queries to $U^*_\varphi$ can help in other oracle problems~\cite{RandomPurificationTang}; however, the preparation unitaries in our query lower bound construction are real, so conjugate access cannot improve the worst-case scaling established here. 
Such an access may still help on some structured  families of instances.
Finally, a classical description of the state-preparation circuit exposes information beyond black-box queries. 
Understanding when that description permits more efficient (quantum or classical) symmetry learning would clarify end-to-end advantages for practical applications.

Overall, the mechanism behind the speedup also suggests a direction beyond the applications considered here. 
The subspace identifier tests whether a prospective sample makes progress over the information already collected, while fixed-point amplitude amplification increases the probability of that progress. 
The adaptive stopping rule then removes the need to know the size of the hidden subgroup in advance.
More broadly, this combination may be useful in other learning strategies that accumulate independent constraints or generators, provided that progress can be recognized efficiently and coherently.

\section*{Acknowledgements}
A.B. would like to thank Alessandro Barenghi and Gerardo Pelosi for useful discussions on lattice membership 
algorithms, Arjan Cornelissen for discussions on the adversary bound, and Ignacio J. Cirac for useful feedback throughout the project. 
Y. L would like to thank Xin Wang for raising the question of controlled-unitary access. 
The authors are also grateful to Marcel Hinsche for initial discussions. 
The work at MPQ is supported by the German Federal Ministry of Education, Research and Space (BMFTR) through the funded project ALMANAQC, grant number 13N17236 within the research program “Quantum Systems”, by THEQUCO as part of the Munich Quantum Valley, which is supported by the Bavarian state government with funds from the Hightech Agenda Bayern Plus, and by the Alexander von Humboldt Foundation. The Berlin team is supported by the BMFTR (Hybrid++, PasQuops, QSolid, MuniQC-Atoms), the Munich Quantum Valley, Berlin Quantum, the Quantum Flagship (Millenion, PasQuans2), the QuantERA, the European Research Council (DebuQC), the Clusters of Excellence (MATH+, ML4Q), and the DFG (CRC 183, SPP 2514, and BoLaCo). For the Munich Quantum Valley, this constitutes the result of 
fruitful joint work on quantum algorithms involving both Berlin and Munich.

\section*{AI Disclosure}
A.B. conceived this project during Adam Bouland’s short plenary talk on StateHSP at QIP 2025. The authors then joined forces to work on the problem, study it, and prove the bounds together. 
The main conceptual algorithmic ideas and worst-case families have been conceived by the authors and are the result of a long study. 
The LLMs of OpenAI’s ChatGPT 5.5 and Anthropic’s Fable assisted the authors in proving the lower bounds and refining some technical details. 
No Pro accounts were used, and no proof has been entirely generated by AI without supervision, careful revisions, and modifications. ChatGPT 5.6 Sol, 6 Astra, and Opus 5 assisted the authors in revising the document and provided feedback. 
The authors take full responsibility for the presentation of the results and their correctness.

\appendix

\section{Group-theoretical facts}
\label{app:group}
\begin{proposition}
\label{proposition: cardinality of perp}
Let $G$ be a finite abelian group and $H\leq G$ a subgroup. The cardinality of the dual subgroup $H^\perp$ equals $|G|/|H|$. 
\end{proposition}
\begin{proof}
    Since $G$ is an abelian group, $H$ as a subgroup of $G$ is a normal subgroup, and the quotient space $G/H$ is also an abelian group. The groups $G,H,G/H$ form an exact sequence, 
    \begin{equation}
    0\ra H \xrightarrow{\iota} G \xrightarrow{\pi} G/H\ra 0.
    \end{equation}
    By Pontryagin duality, the sequence
    \begin{equation}
    0\ra \widehat{G/H} \xrightarrow{\hat{\pi}} \hat{G} \xrightarrow{\hat{\iota}} \hat{H}\ra 0
    \end{equation}
    is also exact. 
    In particular, $\widehat{G/H}\cong \mathrm{im}(\hat{\pi})=\mathrm{ker}(\hat{\iota})$. 
    By definition, given $\lambda\in\hat{G}$ and $h\in H$, $\rho_{\hat{\iota}(\lambda)}(h)=\rho_\lambda(\iota(h))={\color{newtext}\chi_\lambda(h)}$, therefore $\mathrm{ker}(\hat{\iota})=H^\perp$, leading to $H^\perp\cong \widehat{G/H}$ and $|H^\perp|=|\widehat{G/H}|=|G/H|=[G:H]=|G|/|H|$. 
\end{proof}

\begin{proposition}
\label{proposition: duality inequalities}
    Let $G$ be a finite abelian group and $H\leq G, K\leq G$. If $K^\perp<H^\perp$, then $H<K$. 
\end{proposition}
\begin{proof}
    To prove this, we first prove that: if $H<K$, then $K^\perp<H^\perp$. This is straightforward by noting that if $\lambda\in\hat{G}$ satisfies $\chi_\lambda(k)=1$ for any $k\in K$, then by $H<K$ it must satisfies $\chi_\lambda(h)=1$ for any $h\in H$. Therefore, $K^\perp<H^\perp$ will lead to $(H^\perp)^\perp<(K^\perp)^\perp$; and by using $(H^\perp)^\perp=H, (K^\perp)^\perp=K$, leading to $H<K$. 
\end{proof}

\section{Linear algebra and integer matrix forms}

Throughout this appendix we fix $n\in\N$ and moduli $M_1,\dots,M_n\in\N$, and write
$G=\mZ_{M_1}\times\cdots\times\mZ_{M_n}$ additively. A group element is a vector
$a=(a_1,\dots,a_n)$ with $a_i\in\mZ_{M_i}$; addition is componentwise,
$(a+b)_i\equiv a_i+b_i\pmod{M_i}$, and the identity is the zero vector. 

\begin{definition}[Span and generating set]
\label{def: span}
    The \emph{span} of $a^{(1)},\dots,a^{(d)}\in G$ is
    \begin{align}
        \langle a^{(1)},\dots,a^{(d)}\rangle
        \defeq \Big\{\, \textstyle\sum_{j=1}^{d} x_j\,a^{(j)} \ \Big|\ x_1,\dots,x_d\in\mZ \,\Big\},
    \end{align}
    with addition componentwise modulo $M_1,\dots,M_n$. It is a subgroup of $G$. Given a subgroup $K\le G$, we call
    $a^{(1)},\dots,a^{(d)}$ a \emph{generating set of $K$} if
    $\langle a^{(1)},\dots,a^{(d)}\rangle = K$.
\end{definition}

\begin{definition}[Redundant and span-increasing element]
\label{def: redundant}
    An element $b\in G$ is \emph{redundant} with respect to $a^{(1)},\dots,a^{(d)}$ if
    $b\in\langle a^{(1)},\dots,a^{(d)}\rangle$; otherwise $b$ is \emph{span-increasing},
    i.e. $\langle a^{(1)},\dots,a^{(d)},b\rangle \supsetneq \langle a^{(1)},\dots,a^{(d)}\rangle$.
\end{definition}

\begin{remark}
    A non-trivial finite abelian group $G$ as a module over $\mathbb{Z}$ does not have a basis. Among rings of the form
    $\mZ_M$, the group $G$ is a $\mZ_M$-module admitting a basis only when all moduli
    coincide, $G=\mZ_M^{\,n}$, and its subgroup $H\leq G$ is guaranteed to admit a basis only when $M$ is prime. We therefore work
    with generating sets and the span-increasing notion of \cref{def: redundant}; the
    latter is exactly what drives the subgroup-chain bound of \cref{fact: log-generators}.
    Genuine linear independence is used only in \cref{apx: RREF}, where $G=\mZ_M^n$ with $M$
    prime is a vector space over the field $\mZ_M$.
\end{remark}

Given $a^{(1)},\dots,a^{(d)}\in G$, we stack these $d$ elements as the rows of a matrix
\begin{align}
    A=\row\!\big(a^{(1)},\dots,a^{(d)}\big)
    =\begin{bmatrix} a^{(1)}\\ \vdots\\ a^{(d)}\end{bmatrix}
    \in \big(\mZ_{M_1}\times\cdots\times\mZ_{M_n}\big)^{d},
\end{align}
and write $\mspan(A):=\langle a^{(1)},\dots,a^{(d)}\rangle $ for the subgroup generated by its rows. Since we also denote the $k$-th row of $A$ as $A_k$, then $A_k$ is simply $a^{(k)}$. 

\subsection{Reduced row echelon form}
\label{apx: RREF}
In this subsection we specialize to the uniform prime case $G=\mZ_M^n$ with $M$ prime, so that $\mZ_M$ is a field and $A\in\mZ_M^{i\times n}$ is an ordinary matrix over a field. (The mixed- and composite-modulus cases do not admit a field RREF and are handled by the Howell normal form of \cref{apx: howell normal form}, to which the mixed-moduli problem is reduced by lifting to $\mZ_{\lcm}^n$). The two elementary row operations
\begin{enumerate}
    \item permuting rows, and
    \item replacing a row $r$ by $r+x\,r'$ for an integer $x$ and a distinct row $r'$
    (componentwise modulo $M$),
\end{enumerate}
preserve the row span $\mspan(A)$, and Gauss--Jordan elimination uses them to bring $A$ into the RREF. 

\begin{proposition}[Membership testing]
\label{prop:membership-testing-RREF-app}
Let $M$ be prime, $1\leq i<n$, and $A\in\mZ_M^{i\times n}$ in RREF with no all-zero rows and pivot columns $j_1<\dots<j_i$. 
Then every $g\in \mZ_M^n$ decomposes uniquely as \begin{align}
\label{eq: group elem decomposition primes}
    g=\sum_{k=1}^{i} c_k\ A_k + q,\quad \mathrm{with}\quad q_{j_k}=0,
\end{align}
where $A_k$ is the $k$-th row of $A$ and $c_k\in[M]$.
Moreover, $g\in \mspan(A)$ iff $q=0$. Such a decomposition costs $O((n-i)i~\polylog M)$ classical binary operations. 
\end{proposition}

\begin{proof}
    The coefficients can be obtained by $c_k= g_{j_k}\,(A_{k,j_k})^{-1}\!\!\pmod{M}$, and $((A_{k,j_k})^{-1}$ exists because $M$ is prime. Then $q=g-\sum_{k=1}^{i} c_k\,A_k$. 

    We next prove $g\in \mspan(A)\iff q=0$. $q=0$ leads to $g\in\mathrm{span}(A)$ is by definition. 
    Showing that $q\neq 0$ leads to $g\notin \mathrm{span}(A)$ is equivalent to showing $q\neq 0$ leads to $q\notin \mathrm{span}(A)$. Suppose there exists $q\neq 0$ that $q_{j_k}=0$, and $q\in\mathrm{span}(A)$. Then one can expand $q=\sum_k c'_k A_k$ and at least one $c_k'$ is nonzero. Take a nonzero $c_k'$, then $q_{j_k}\neq 0$, contradicts with $q_{j_k}=0$.  
    
     Finally, to prove the decomposition is unique, suppose there exist two distinct decompositions $g=\sum_k c_k A_k+q$ and $g=\sum_k c'_k A_k+q'$. That is, $\sum_{k=1}^i (c_k-c_k') A_k+(q-q')=0$. Consider the $j_k$-th component, $c_k A_{k,j_k}=c_k'A_{k,j_k}\pmod{M}$. Applying the inverse $A_{k,j_k}^{-1}$ on both sides lead to $c_k=c_k'$. Since it holds for all $k$, then $q=q'$. 
\end{proof}

We next show how to update the RREF incrementally when a new element is adjoined.

\begin{theorem}[Incremental Gauss--Jordan elimination]
\label{theorem: apx gauss-jordan}
    Let $M$ be prime, $1\leq i<n$, and $A\in\mZ_M^{i\times n}$ in RREF with no all-zero rows. For a
    new vector $b\in\mZ_M^n$, the matrix obtained by adjoining $b$ can be returned in RREF
    using $O((n-i)i~\polylog M)$ binary operations.
\end{theorem}
\begin{proof}
    \emph{Procedure.} Denote the $k$-th row of $A$ as $A_k$ with pivot columns
    $j_1<\dots<j_i$. For each $k=1,\dots,i$ in increasing order, if $b_{j_k}\neq 0$ set
    $c_k\equiv -\,b_{j_k}\big(A_{k,j_k}\big)^{-1}\pmod{M}$ and replace $b$ by
    $b+c_k A_k$. This zeroes the entries of $b$ in the pivot columns of $A$ and turns
    $b$ into the residual $q$ of \eqref{eq: group elem decomposition primes}. If $q=0$ then
    $b\in \mspan(A)$ and the algorithm returns $A$; otherwise $q$ is span-increasing, and we
    append $q$, normalize its leading entry to a pivot, and clear that column in the other
    rows to restore RREF.

    \emph{Cost.} Each scalar operation in $\mZ_M$ costs $O(\polylog M)$. The first loop is
    $O(i)$ operations on length-$(n-i+1)$ vectors, i.e. $O((n-i) i~\polylog M)$; clearing the new pivot column costs $O((n-i)i~\polylog M)$; locating pivots and inserting $q$ in row
    order is of order $O(n ~\polylog M)$. The total is $O((n-i)i~\polylog M)$.
\end{proof}

\begin{remark}
    For $G=\mZ_2^n$ this simplifies: $A_{k,j_k}=1$, so $\big(A_{k,j_k}\big)^{-1}=1$
    and $-x\equiv x\pmod 2$, whence $c_k=b_{j_k}$.
\end{remark}

Adjoining one row at a time, a full Gauss--Jordan elimination of an $i\times n$ matrix
costs $O(n i^2~\polylog M)$ operations. From the RREF we can read off a generating set for the
kernel.

\begin{theorem}[Generating the kernel, $\mZ_M^n$ prime]
\label{theorem: apx kernel generation}
    Let $M$ be prime and $A\in\mZ_M^{i\times n}$, $i\le n$, in RREF with no all-zero rows.
    Its kernel $\ker(A)$ is generated by $n-i$ vectors with at most $i+1$ non-zero entries each,
    output in $O((n-i)n~\polylog M)$ binary operations, or in $O((n-i)i~\polylog M)$ binary operations if using sparse representation of the vectors.
\end{theorem}
\begin{proof}
    Let $j_1,\dots,j_i$ be the pivot columns and $f_1,\dots,f_{n-i}$ the free columns, and
    denote the $k$-th row of $A$ as $A_k$. For each $l\in\{1,\dots,n-i\}$ define
    $v^{(l)}\in\mZ_M^n$ whose components are
    \begin{align}    
        v^{(l)}_j \defeq \begin{cases}
            1, & j=f_l\\
         -(A_{k,j_k})^{-1} A_{k,f_l} \pmod M & j=j_k, k \in [i],\\
            0 & \text{otherwise}.
        \end{cases}
    \end{align}
    Then $A v^{(l)}$ has component $(Av^{(l)})_k=A_{k,f_l}-\sum_{k'=1}^i A_{k,j_k'} (A_{k',j_k'})^{-1} A_{k',f_l}=0$, 
    so $v^{(l)}\in\ker(A)$. Any $x\in\ker(A)$ is determined by its free coordinates: fixing
    them, the equations $Ax\equiv0$ fix the pivot coordinates; and the $v^{(l)}$ realize the
    free unit vectors, so they generate $\ker(A)$ and are independent. Under the above construction, each $v^{(l)}$ has at most $i+1$ non-zero entries. Constructing them
    takes $O((n-i)i)$ arithmetic operations, i.e. $O((n-i)i~\polylog M)$ binary operations.
\end{proof}

\subsection{Howell normal form}
\label{apx: howell normal form}

In this section, we show the details and proofs for the Howell normal form algorithm~\cite{storjohann1998fast} and the kernel algorithm. 
We first introduce the basic operations of $\mathbb{Z}_M$~\cite{storjohann1998fast} with $M\in\mathbb{Z}_+$ that will be used in the algorithms. Denote $S=\{0,1,\cdots,M-1\}$. Given $a,b\in S$, the following basic operations can be achieved with binary operations $O(\mathrm{polylog}(M))$,
\begin{itemize}
    \item $\mathrm{Gcdex}(a,b)$: returns $g,s,t,u,v\in S$ such that $g=\gcd(a,b)$, $g=sa+tb,0=ua+vb$, and $sv-tu=1$ in $\mathbb{Z}_M$. ($u=-b/g$ and $v=a/g$)
    \item $\mathrm{Quo}(a,b)$: when $b\neq 0$, returns $q\in S$ such that $a-qb=r$ with $0\leq r <b$. 
    \item $\mathrm{Ann}(a)$: returns $c\in S$ such that $c=M/\gcd(a,M)$ in $\mathbb{Z}_M$. 
    \item $\mathrm{Unit}(a)$: when $a\neq 0$, returns $c\in S$ such that $\gcd(c,M)=1$ and $ca=\gcd(a,M)$ in $\mathbb{Z}_M$.
\end{itemize}

We comment that for $c=\mathrm{Unit}(a)$, $\langle ca\rangle=\langle a\rangle$ since $c$ is invertible in $\mathbb{Z}_M$ due to $\gcd(c,M)=1$. Below, we restate~\cref{theorem: howell algorithm} and show an explicit algorithm.

\begin{theorem}[Howell normal form algorithm~{\cite{storjohann1998fast}}, restating~\cref{theorem: howell algorithm}]
    Let $A\in\mathbb{Z}_M^{i\times n}$ be a matrix over $\mathbb{Z}_M$.
    Then, Algorithm~\ref{alg: howell algorithm} brings $A$ into Howell normal form using $O(n^2\max(n,i) )$ arithmetic operations.
\end{theorem}

\begin{algorithm}[t!]
  \caption{Howell normal form}
  \label{alg: howell algorithm}
  \DontPrintSemicolon          
  \SetAlgoLined                
  \SetKwInOut{Input}{Input}
  \SetKwInOut{Output}{Output}
  \BlankLine          

  \Input{A matrix $A\in\mathbb{Z}_M^{i\times n}$}
  \Output{The Howell normal form of $A$}

  \If{$i<n$}{augment $A$ with zero rows to make it square, and set $i$ to $n$. }

  \tcp{Part A: Put $A$ in upper triangular form}
  \For{$k=1$ \KwTo $n$}{
    \For{$k'=k+1$ \KwTo $i$}{
        $(g,s,t,u,v):=\mathrm{Gcdex}(A_{k,k},A_{k',k})$

        $
        \begin{pmatrix}
            A_k\\
            A_{k'}
        \end{pmatrix}:=
        \begin{pmatrix}
            s & t\\
            u & v
        \end{pmatrix}\begin{pmatrix}
            A_k\\
            A_{k'}
        \end{pmatrix}
        $
    }  
  }  
  \tcp{$A$ is now a square matrix with dimension $n$}
  \tcp{Part B: Put $A$ in Howell form. }
  Truncate rows $n+1$ to $i$ (which are zeros). Then augment $A$ with one zero row. \; 
  \For{$k=1$ \KwTo $n$}{
    \tcp{Part B.1}
    \If{$A_{k,k}\neq 0$}{
        \tcp{Make $A_{k,k}$ divides $M$}
        $A_k:=\mathrm{Unit}(A_{k,k}) A_k$\; 
        \For{$k'=1$ \KwTo $k-1$}{
            $A_{k'}:=A_{k'}-\mathrm{Quo}(A_{k',k},A_{k,k})A_{k}$
        }
        $A_{n+1}:=\mathrm{Ann}(A_{k,k}) A_k$
    }\Else{
    $A_{n+1}:=A_k$
    }
    \tcp{Part B.2: Enforce the extended row in $\langle A_{k+1},\cdots,A_n\rangle$}
    \For{$k'=k+1$ \KwTo $n$}{
        $(g,s,t,u,v):=\mathrm{Gcdex}(A_{k',k'},A_{n+1,k'})$

        $
        \begin{pmatrix}
            A_{k'}\\
            A_{n+1}
        \end{pmatrix}:=
        \begin{pmatrix}
            s & t\\
            u & v
        \end{pmatrix}\begin{pmatrix}
            A_{k'}\\
            A_{n+1}
        \end{pmatrix}
        $
    }
  }
  Move all nonzero rows to the top of $A$, and truncate the 
  zero rows.\;
  \Return{$A$.}
\end{algorithm}

The cost can be obtained by counting the number of operations in the algorithm. With the Howell normal form, we now state the procedure for membership testing and generating the kernel. 

\begin{proposition}[Membership testing]
\label{prop:membership-testing-HNF-app}
    Let $A \in \mathbb{Z}_M^{i \times n}$ in Howell normal form with pivot columns $j_1 < \cdots < j_{i}$. Then any $g \in \mathbb{Z}_M^n$ decomposes 
    uniquely as
    \begin{equation}
        g = \sum_{k=1}^{i} c_k A_k + q,
    \end{equation}
    where $c_k \in [M/A_{k,j_k}]$ for $k = 1, \ldots, i$, and $q_{j_k}\in [A_{k,j_k}]$ for $k = 1, \ldots, i$. 
    Moreover, $g\in \mspan(A)\iff q=0$. Such a decomposition costs $O(ni~\polylog M)$ binary operations. 
\end{proposition}

\begin{proof}
    The coefficients $c_k$ can be obtained by the following procedure: set $q^{(0)}=g$; for $k=1$ to $i$, compute $c_k:=\lfloor q^{(k-1)}_{j_k}/A_{k,j_k}\rfloor$ and set $q^{(k)}= q^{(k-1)}-c_k A_k(\mathrm{mod} M)$. Finally, set $q=q^{(i)}$. This guarantees that $c_k \in [M/A_{k,j_k}]$ for $k = 1, \ldots, i$, and $q_{j_k}\in [A_{k,j_k}]$ for $k = 1, \ldots, i$, and also justifies the computational cost.

    We next prove $g\in \mspan(A)\iff q=0$. $q=0$ leads to $g\in\mathrm{span}(A)$ is by definition. Showing that $q\neq 0$ leads to $g\notin \mathrm{span}(A)$ is equivalent to showing $q\neq 0$ leads to $q\notin \mathrm{span}(A)$. Suppose there exists $q\neq 0$ that $q_{j_k}\in [A_{k,j_k}]$, and $q\in\mathrm{span}(A)$. Consider $q_{j_1}$, and because $j_1$ is the left-most pivot position, $q\in\mspan(A)$ leads to that $q_{j_1}$ must be a multiple of $A_{1,j_1}$, which together with $q_{j_1}\in [A_{1,j_1}]$, leads to $q_{j_1}=0$ and the first $j_1$ components of $q$ must also vanish. Therefore, $q\in\mspan(A_2,\cdots,A_i)$ by the extended-row property. Repeat the above reasoning on $j_2$ until $j_i$ leads to the conclusion that $q=0$. 

    Finally, to prove the decomposition is unique, suppose there exist two distinct decompositions $g=\sum_k c_k A_k+q$ and $g=\sum_k c'_k A_k+q'$. That is, $\sum_{k=1}^i (c_k-c_k') A_k+(q-q')=0$. Consider the $j_1$-th component, $c_1 A_{1,j_1}+q_{j_1}=c_1'A_{1,j_1}+q'_{j_1}\pmod{M}$. Since $c_1,c_1'\in[M/A_{1,j_1}]$ and $q_{j_1},q_{j_1}'\in[A_{1,j_1}]$, LHS and RHS are equal as integers (one can remove modulo $M$). Their Euclidean quotient and remainder upon division by $A_{1,j_1}$ are unique, so $c_1=c_1'$ and $q_{j_1}=q'_{j_1}$. One can next consider $\sum_{k=2}^i (c_k-c_k') A_k+(q-q')=0$ and take the $j_2$-th component. Repeating this procedure leads to $c_k=c_k'$ for all $k$, and therefore, $q=q'$.  
\end{proof}

The Howell normal form is also crucial for computing the kernel of the matrix $A$.

\begin{theorem}[Generating the kernel, $\mZ_M^n$]
\label{theorem: apx kernel-M}
    Let $M\in\mathbb{Z}_+$ and $A \in \mZ_{M}^{i\times n}$ with $i \leq n$, in Howell normal form with no all-zero rows.
    Its kernel $\mathrm{ker}(A)$ can be generated by $n$ vectors with $O(i)$ non-zero entries each, and there exists an algorithm that outputs these vectors in $O((ni^2+n^2)~\polylog M)$ classical binary operations, or in $O(n i^2~\polylog M)$ classical binary operations if using sparse representation of the vectors.
\end{theorem}

\begin{proof}
    For simplicity of notation, denote the pivot of row $k$ as $d_k$, and denote $\tilde{c}_k=M/d_k$. Denote the set of pivotal columns $P=\{j_1,\cdots,j_i\}$. Columns that are not pivot columns correspond to free variables, i.e., the set of free variables is $F=\{f_1,\cdots,f_{n-i}\}:=\{1,2,\cdots,n\}\setminus P$. Consider the following algorithm:

\begin{algorithm}[H]
  \caption{Generating set of $\ker(A)$}
  \label{alg:howell-kernel}
  \DontPrintSemicolon          
  \SetAlgoLined                
  \SetKwInOut{Input}{Input}
  \SetKwInOut{Output}{Output}
  \BlankLine          

  \Input{A matrix $A\in\mathbb{Z}_M^{i\times n}$ with $i\leq n$ in the Howell normal form with no zero rows, with pivotal set $P$ and free variable set $F$.}
  \Output{The generating set $S$ of $\ker(A)$. }
  $S=\{\}$. \;
 
  \If{$F$ is non-empty 
  }{ 
  \tcp{Part A: Obtain $n-i$ independent generators}
  \For{$l=1$ \KwTo $n-i$}{
    Set $x_{f_l}=1$ and $x_{f_{l'}}=0$ for all $l'\neq l$.\;
  \tcp{With chosen free variables $x_f,f\in F$, solve for $x_j$ with $j\in P$ from bottom to top}
  \For{$k=i$ \KwTo $1$}{
  \tcp{Solve the kernel condition for row $k$}
    Compute $x_{j_k}=-\frac{1}{d_k}(\sum_{j>j_k}A_{kj}x_j)$. 
  }
  Append $x$ to $S$. 
  }
 }
 
 \If{$P$ is non-empty}{ 
 \tcp{Part B: Obtain $i$ generators}
  \For{$l=i$ \KwTo $1$}{
    Set $x_f=0$ for all $f\in F$. Set $x_{j_{l'}}=0$ for all $l'>l$, and set $x_{j_l}=\tilde{c}_l$. \;
    \For{$k=l-1$ \KwTo $1$}{
        Compute $x_{j_k}=-\frac{1}{d_k}(\sum_{j>j_k}A_{kj}x_j)$. 
    }
  Append $x$ to $S$. 
  }}
  
  \Return{$S$.} 
\end{algorithm} 

To see the validity of the algorithm, note that at iteration $k$, $x$ has satisfied the kernel conditions for rows below $k$, i.e., $A_{k'}\cdot x=0\pmod M$ for all $k'>k$. By extended row property, $\tilde{c}_k A_k\in \langle A_{k+1}\cdots, A_i\rangle$, leading to $\tilde{c}_k A_k\cdot x=0\pmod M$, which is
    $
    \tilde{c}_k \left(d_k x_{j_k}+\sum_{j>j_k}A_{kj} x_j \right)=0\pmod M 
    $. In particular, it leads to $\sum_{j>j_k}A_{kj} x_j=0\pmod d_k$, and therefore $\frac{1}{d_k} \sum_{j>j_k}A_{kj} x_j$ must be an integer. 
Now solve the kernel condition for row $k$, $d_k x_{j_k}+\sum_{j>j_k} A_{kj}x_j=0\pmod M$ leads to the solution for $x_{j_k}$ as $x_{j_k}=-\frac{1}{d_k}(\sum_{j>j_k}A_{kj}x_j)+q \tilde{c}_k$, with $q=0,1,\cdots,d_k-1$. We choose $q=0$ for simplicity.

To see that the output generates the whole of $\ker(A)$, take any $v\in \ker(A)$. Subtracting $\sum_l v_{f_l} x^{(l)}$ where $x^{(l)}$ denotes the Part-A generator with unit coordinate at $f_l$, yields a kernel element $v$ supported on the pivot coordinates. We reduce $v$ to $\bzr$ from the bottom pivot upward. At stage $l=i,\cdots,1$, suppose $v\in \ker(A)$ vanishes on the free coordinates and on the pivots $j_{l'}$ with $l'>l$. The kernel condition of row $l$ reads $d_l v_{j_l}+\sum_{j>j_l} A_{lj} v_j=0$, leading to $d_l v_{j_l}=0$, i.e., $v_{j_l}$ is a multiple of $\tilde{c}_l$. Subtracting the corresponding multiple of Part-B generator $x^{(l)}$ (generated at iteration $l$ of the outer loop) makes $v_{j_l}=0$. After stage $l=1$ we are left with $v=\bzr$, so the $n$ output vectors generate $\ker(A)$.

From the counting in the algorithm, a direct result is: $\mathrm{ker}(A)$ is a subgroup of $\mathbb{Z}_M^n$ with cardinality $|\mathrm{ker}(A)|=M^{|F|}\prod_{k=1}^i d_k$. 

\textit{Cost.} Each vector $x$ produced above is of sparsity $O(i)$, with at most $i+1$ non-zero elements. Part A costs $O((n-i)i^2)$ arithmetic operations and part B costs $O(i^3)$. The total cost is then $O(n i^2)$ arithmetic operations or $O(n i^2~\polylog M)$ binary operations. Since there are $n$ vectors $x$, the cost of vector initialization in sparse representation is $O(ni)$ or $O(n^2)$ in dense representation. 
\end{proof}

The $n-i$ generators from part A and the $i$ generators from part B provide us a generating set with size $n$. We note that the $i$ generators from part B may not be independent (the generating set may be redundant), nevertheless, they are able to generate $\prod_{k=1}^i d_k$ different elements. For example, consider $M=4$, and when
\begin{equation}
A=\begin{pmatrix}
    2 & 1\\
    0 & 2
\end{pmatrix}
\end{equation}
the kernel generators from part B are $x=(2,0)$ and $x'=(1,2)$. They are not independent, since $x=2\cdot x'$. Yet together, they generate all 4 elements of the kernel, as the counting predicts.

\section{Fixed point amplitude amplification}
\label{apx: fixed point amp amp}
In this section, we connect the fixed point amplitude amplification to singular value decomposition, following the treatments of Refs.~\cite{yoder2014fixed,gilyen2019quantum}. We will first summarize the useful theorems, then present and prove the main result~\cref{prop:amp-psi0}. 

Given a polynomial $P\in\mathbb{C}[x]$ and an operator $A$, if $P$ is an odd polynomial and $A=W\Sigma V^\dagger$ is a singular value decomposition (SVD), then $P^{(SV)}(A):=W P(\Sigma) V^\dagger$. In Ref.~\cite{gilyen2019quantum}, it has been shown that: if $A$ allows a projected unitary encoding $A:=\tilde{\Pi} U \Pi$ where $U$ is a unitary and $\tilde{\Pi},\Pi$ are orthogonal projectors, and if $P\in\mathbb{R}[x]$ is a degree-$m$ odd polynomial that is bounded by 1 in absolute value on $[-1,1]$, then one can obtain the corresponding phases $\Phi\in\mathbb{R}^m$ and implement a unitary $U_\Phi$ such that $P^{(SV)}(A)=\tilde{\Pi}U_\Phi \Pi$. For our purpose, we will focus on the case where $P$ is a polynomial approximation of the sign function up to a phase, which we call the amplification function.

\begin{proposition}[Explicit phases for amplification function]
\label{prop:explicit-phase}
    Let $0<\gamma<1$ and $0<\delta<1$. Define a parametrized single-qubit reflection operator for all $x\in[-1,1]$ 
    \begin{equation}
        R(x):=\begin{pmatrix}
            x & \sqrt{1-x^2}\\
            \sqrt{1-x^2} & -x
        \end{pmatrix}.
    \end{equation}
    Then there exists an odd integer $m=O(\frac{\log(1/\delta)}{\gamma})$ and $\Phi=(\phi_1,\phi_2,\cdots,\phi_m)\in\mathbb{R}^m$ such that the polynomial defined by $P_\Phi(x):=\langle 0|M_\Phi(x)|0\rangle$ with 
    \begin{equation}
    \label{eqn:define-MPhi}
        M_\Phi(x):=\prod_{j=1}^m\left(e^{i\phi_j \sigma_z}R(x)\right)
    \end{equation}
    satisfies: (1) it is an odd polynomial with degree at most $m$; (2) $|P_\Phi(x)|\leq 1$ for $x\in[-1,1]$; and (3) $|P_\Phi(x)|^2\geq 1-\delta$ for $\gamma\leq |x|\leq 1$. The phases can be computed classically using $O(m)$ elementary-function evaluations, or $O(m~\polylog(m/\delta))$ binary operations. 
\end{proposition}

\begin{proof}
We prove by construction, that one can construct the phases such that $P_\Phi(x)$ is a polynomial approximation of the sign function up to a phase for the interval $x\in[-1,1]$. Choose the smallest odd integer $m\geq \frac{\log(4/\sqrt{\delta})}{\gamma}$ and denote $l=(m-1)/2$. Define 
    \begin{equation}
        \alpha_j = 2\arccot(\gamma \tan\frac{2\pi j}{m} ),\quad \beta_j = -\alpha_{l-j+1}, 
    \end{equation}
    and construct $\Phi=(0,-\frac{\alpha_l}{2},\frac{\beta_l}{2},\cdots,-\frac{\alpha_1}{2},\frac{\beta_1}{2})$. We now prove it satisfies the three properties. 

    (1) To show $P_\Phi$ is an odd polynomial degree of at most $m$, we note that using the explicit form of $R(x)$, one can show 
    $
    P_\Phi(x)=\sum_{j=0}^{(m-1)/2} c_j x^{m-2j} (1-x^2)^j
    $
    for some constants $c_j\in\mathbb{C}$. Each summand has degree at most $m$ and is odd. 

    (2) To show $|P_\Phi(x)|\leq 1$ on $[-1,1]$, note that $M_\Phi(x)$ is unitary on $[-1,1]$, and then $\langle 0|M_\Phi(x) M_\Phi^\dg(x)|0\rangle=1$, leading to $|P_\Phi(x)|^2+|\langle 0|M_\Phi(x)|1\rangle|^2=1$ and thus $|P_\Phi(x)|\leq 1$. 

    (3) To show $|P_\Phi(x)|^2\geq 1-\delta$ for $\gamma\leq |x|\leq 1$, note that for $0\leq x\leq 1$, 
    \begin{equation}
    \label{eqn:1-P}
        1-|P_\Phi(x)|^2=\frac{T_m\left(\frac{\sqrt{1-x^2}}{\sqrt{1-\gamma^2}}\right)^2}{T_m\left(\frac{1}{\sqrt{1-\gamma^2}}\right)^2}
    \end{equation}
    which is proven in Yoder-Low-Chuang~\cite{yoder2014fixed} (under the substitutions $L=m,\lambda=x^2,\gamma_{\mathrm{YLC}}=\sqrt{1-\gamma^2},\delta_{\mathrm{YLC}}=1/T_m(1/\sqrt{1-\gamma^2})$), and $T_m$ is the Chebyshev polynomial of the first kind. Now, for $\gamma\leq x\leq 1$, $0\leq \frac{\sqrt{1-x^2}}{\sqrt{1-\gamma^2}}\leq 1$, and for $x\in[-1,1]$, $T_m(x)=\cos(m \cos^{-1} (x))$, leading to $|T_m(x)|\leq 1$ for $x\in[-1,1]$. Therefore, the numerator of Eq.~\eqref{eqn:1-P} is at most 1. 

    For the denominator, $T_m(1/\sqrt{1-\gamma^2})=\cosh(m \tanh^{-1} \gamma)$. Using that $\tanh^{-1}(x)\geq x$ for $x\in[0,1]$, we have $T_m(1/\sqrt{1-\gamma^2})=\cosh(m \tanh^{-1} \gamma)\geq \cosh(m\gamma)\geq \frac{1}{2}e^{m\gamma}$. Therefore, $1-|P_\Phi(x)|^2\leq 4 e^{-2m\gamma}\leq \frac{\delta}{4}\leq \delta$. Finally, oddness of $P_\Phi$ gives $|P_\Phi(-x)|=|P_\Phi(x)|$ so the same bound holds for $-1\leq x\leq -\gamma$. This finishes the proof of the properties of $P_\Phi$. 

    Now, consider the case that we output each $\phi_i$ to $p$ bits, which costs in total $\tilde{O}(m(p+\log m))$. We denote the output phase $\hat{\phi}_i$, and $|\phi_i-\hat{\phi}_i|\leq 2^{-p}$, leading to $\|e^{i\phi_i \sigma_z}-e^{i\hat{\phi}_i \sigma_z}\|\leq 2^{-p}$, where the norm is operator norm (the Schatten infinite norm). 
    Using telescoping, one can show $\|M_\Phi - M_{\hat{\Phi}}\|\leq m 2^{-p}$. Since the operator norm of an operator is the largest singular value, we have $|P_\Phi-P_{\hat{\Phi}}|=|\langle 0| M_\Phi - M_{\hat{\Phi}}|0\rangle|\leq \|M_\Phi - M_{\hat{\Phi}}\|\leq m 2^{-p}$. And $| |P_{\Phi}|^2-|P_{\hat{\Phi}}|^2 |=(|P_{\Phi}|+|P_{\hat{\Phi}}|)(| |P_{\Phi}|-|P_{\hat{\Phi}}| |)\leq 2 | P_{\Phi}-P_{\hat{\Phi}}|\leq 2m 2^{-p}$. 
    Therefore, 
    \[
    1-|P_{\hat{\Phi}}|^2=1-|P_{\Phi}|^2 + (|P_{\Phi}|^2-|P_{\hat{\Phi}}|^2)\leq \frac{\delta}{4} + 2m 2^{-p}. 
    \]
    By choosing $p=\lceil \log_2 \frac{4m}{\delta}\rceil$, we obtain $1-|P_{\hat{\Phi}}|^2\leq \frac{\delta}{4}+\frac{\delta}{2} \leq\delta$, that is, all the properties still apply to $P_{\hat{\Phi}}$. This finishes the proof of classical binary operations being $O(m~\polylog(m/\delta))$. 
\end{proof}

We comment that the above proposition constructs the phases $\Phi$ explicitly for the amplification function $P_\Phi$. For a generic function $P_R\in\mathbb{R}[x]$, Ref.~\cite{gilyen2019quantum} provides a method to find the corresponding phases $\Phi$ for an approximate function $P\in\mathbb{C}[x]$. By a careful error propagation analysis, when requiring $|P_R-\mathrm{Re}[P]|<\delta$, the classical cost of finding the phases is $O(m^3~\polylog (m/\delta))$. In this work, we focus on the amplification function which has a lower classical computational cost of $O(m~\polylog (m/\delta))$. 

\begin{definition}[Phased alternating sequence]
    Let $\mathcal{H}$ be a finite-dimensional Hilbert space and let $U,\Pi,\tilde{\Pi}\in\mathrm{End}(\mathcal{H})$ be linear operators on $\mathcal{H}$ such that $U$ is a unitary, and $\Pi,\tilde{\Pi}$ are orthogonal projectors. Let $\Phi=(\phi_1,\phi_2,\cdots,\phi_m)\in\mathbb{R}^m$ with $m$ being an odd integer. We define the phased alternating sequence $U_\Phi$ as
    \begin{equation}
        U_\Phi=e^{i\phi_1 (2\tilde{\Pi}-\bo)} U \prod_{j=1}^{(m-1)/2} \left(e^{i\phi_{2j}(2\Pi-\bo)}U^\dagger  e^{i\phi_{2j+1}(2\tilde{\Pi}-\bo)}U\right).
    \end{equation}
\end{definition}

The unitary $e^{i\phi (2\Pi-\bo)}$ can be implemented using a single ancilla qubit, two $\C_\Pi\NOT$ gates, and an $e^{-i\phi\sigma_z}$ gate, and similarly for $e^{i\phi (2\tilde{\Pi}-\bo)}$. 

The following theorem states that when a polynomial $P$ can be expressed as $P=\langle 0|M_\Phi|0\rangle$ with $M_\Phi$ in Eq.~\eqref{eqn:define-MPhi}, then $\tilde{\Pi} U_\Phi \Pi$ performs the singular value transformation $P^{(SV)}(\tilde{\Pi}U\Pi)$. 
\begin{theorem}
\label{thm:U_Phi}
    Let $\mathcal{H}$ be a finite-dimensional Hilbert space and let $U,\Pi,\tilde{\Pi}\in\mathrm{End}(\mathcal{H})$ be linear operators on $\mathcal{H}$ such that $U$ is a unitary, and $\Pi,\tilde{\Pi}$ are orthogonal projectors. Let $\Phi\in\mathbb{R}^m$ and $P=\langle 0|M_\Phi|0\rangle$ with $M_\Phi$ defined in Eq.~\eqref{eqn:define-MPhi}, and $m$ being an odd integer. Then,
    \begin{equation}
        P^{(SV)}(\tilde{\Pi}U\Pi)=\tilde{\Pi} U_\Phi \Pi
    \end{equation}
    with $U_\Phi$ being the phased alternating sequence of $\Phi$. 
\end{theorem}

We note that the notations we adopt in this section are different from~\cite{gilyen2019quantum} by replacing $\delta$ with $\gamma$, and $\epsilon$ with $\delta$. 
With these theorems, we are now ready to present and prove the main result for our purpose. The idea is similar to fixed-point amplitude amplification, while the technical details are different. In particular, the projector $\tilde{\Pi}$ onto the ``good'' state is constructed through a flag qubit, so the construction of $\tilde{\Pi}$ does not require knowing $|\psi_G\rangle$. Furthermore, our construction strengthens~\cite{gilyen2019quantum} by showing that if the flag qubit is measured 0, the data register is not only close to the good state, but is \textit{exactly} the good state. This guarantee is important for our application to the state HSP problem. 

\begin{theorem}
\label{prop:amp-psi0}
    Given state $|\psi_0\rangle=|0\rangle^{\otimes n}\otimes|0\rangle$ and $U$ such that
    \begin{equation}
    U|\psi_0\rangle=a|\psi_G\rangle\otimes |0\rangle+\sqrt{1-a^2}|\psi_B\rangle\otimes|1\rangle,
\end{equation}
with $a\geq\gamma>0$.
Then there is an integer $m=O(\frac{\log(1/\delta)}{\gamma})$ with $\delta\in(0,1)$, and a $\Phi\in\mathbb{R}^m$ with a corresponding unitary $U_\Phi$, such that if measuring the second register of $U_\Phi|\psi_0\rangle$ yields 0, then the state is $|\psi_G\rangle$ (up to an irrelevant global phase). The success probability of measuring 0 in the second register is lower bounded by $1-\delta$. $U_\Phi$ can be implemented using a single ancilla qubit, with $m$ uses of $U,U^\dagger$, $m$ uses of $\C_\Pi\NOT$ and $\C_{\tilde{\Pi}}\NOT$, and $m$ single qubit gates, where $\Pi=|\psi_0\rangle\langle\psi_0|$ and $\tilde{\Pi}=\bo\otimes |0\rangle\langle 0|$. 
\end{theorem}
\begin{proof}
    Let $P_\Phi$ be the odd polynomial approximation of the amplification function of degree $m=O(\frac{\log(1/\delta)}{\gamma})$, expressed as $P=\langle 0|M_\Phi|0\rangle$ with $M_\Phi$ in Eq.~\eqref{eqn:define-MPhi}, and the phases $\Phi$ chosen as in the proof of~\cref{prop:explicit-phase}. Then, $|P(a)|^2\geq 1-\delta$. 
    
    Given such a $P_\Phi$ with phases $\Phi$, the corresponding phased alternating sequence $U_\Phi$ satisfies $\tilde{\Pi}U_\Phi\Pi=P^{(SV)}(\tilde{\Pi}U\Pi)$ (\cref{thm:U_Phi}). By definition, $\tilde{\Pi}U\Pi=a|\psi_G\rangle\otimes|0\rangle\langle\psi_0|$, and $P^{(SV)}(\tilde{\Pi}U\Pi)$ equals $P(a)|\psi_G\rangle\otimes|0\rangle\langle\psi_0|$. Therefore,  
    \begin{equation}
    \begin{aligned}
        U_\Phi|\psi_0\rangle&=\tilde{\Pi}U_\Phi\Pi|\psi_0\rangle +(\bo-\tilde{\Pi})U_\Phi\Pi|\psi_0\rangle\\
        &=P(a)|\psi_G\rangle\otimes|0\rangle +c|\mathrm{junk}\rangle\otimes|1\rangle
    \end{aligned}
    \end{equation}
    where $c$ is the normalization factor for $|\mathrm{junk}\rangle$, and we note that $|\mathrm{junk}\rangle$ may not be the same as $|\psi_B\rangle$. 
    If by measuring the second register we obtain state $|0\rangle$, then the first register is in state $|\psi_G\rangle$ up to a global phase. The probability of measuring 0 in the second register is $|P(a)|^2\geq 1-\delta$.
\end{proof}

Implementing $\C_\Pi \NOT$ costs $O(n^2)$ elementary gates, or $O(n)$ elementary gates with one ancilla qubit; and implementing $\C_{\tilde{\Pi}}\NOT$ costs $O(1)$ elementary gates. Therefore, the overall cost of elementary gates in the above proposition is $O(n \log(1/\delta)/\gamma)$ with access of one ancilla qubit. Computing the phases $\Phi\in \mathbb{R}^m$ for the amplification function costs $O(m ~\polylog(m/\delta))$ binary operations. 

\section{Subspace identifier inspired by Brassard and H{\o}yer's method}
\label{app:second-approach}
An alternative approach to identify a good $H^\perp \setminus K^\perp$ subspace builds on the exact Simon's method of \citet{brassard1997exact}.
We build a quantum circuit that acts on the group register containing the irrep labels $\ket{\lambda}$ and projects it outside the $K^\perp$ subspace. Equivalently, this circuit maps every label $\lambda \in K^\perp$ in the current span to the $\lambda=\bzr$ label, while it keeps elements that contain new span-increasing generators different from the $\lambda=\bzr$ label. 
Before proceeding, we first note the lemma on the decomposition of $H^\perp$. 
 
\subsection{A decomposition lemma}
 
\begin{lemma}[Decomposition of $H^\perp$]
\label{lemma:decompose-H-perp}
Let $G=\mZ_2^n$ and $H^\perp \leq \widehat{G}$. 
Let $K^{\perp}\le H^\perp$ be generated by $S^{\perp}=\{s^{\perp(1)},\ldots,s^{\perp(i)}\}$, and suppose that the encoding matrix $A \in\mathbb{Z}_2^{i\times n}$ with rows $s^{\perp(1)},\ldots,s^{\perp(i)}$ is in RREF with pivot columns $j_1,\ldots,j_i$, i.e., $A_{k,j_k}=s_{j_k}^{\perp(k)}=1  \text{ for } k=1,\ldots,i$. 
Then $H^\perp$ admits a direct-sum decomposition
$
H^\perp = K^{\perp} \oplus Q,
$
where
\begin{equation}
Q := \bigl\{\lambda=(\lambda_1,\cdots,\lambda_n)\in H^\perp \;\big|\; \lambda_{j_1}=\cdots=\lambda_{j_i}=0\bigr\}.
\end{equation}
\end{lemma}
\begin{proof}
    $K^{\perp}\cap~ Q=\{0\}$ is satisfied by definition. Proving that  any $\lambda\in H^\perp$ can be decomposed into $\lambda=\mu+ \nu$ where $\mu\in K^{\perp}$ and $\nu\in Q$ amounts to applying the procedure in membership testing~(\cref{prop:membership-testing-RREF-app}). 
\end{proof}
 
As a result, any $\lambda\in H^\perp$ can be written as 
    $\lambda=(\sum_{k=1}^{i}\iota_k s^{\perp(k)})+ q  $
for some $\iota_k \in\{0,1\}$ and $q\in Q$, where $0s^{\perp(k)}:=0$ and $1s^{\perp(k)}:=s^{\perp(k)}$. From the previous proof $\iota_k=\lambda_{j_k}$. 
 
We now state the main result of this section, the construction of a subspace identifier given $S^\perp$. 
 
\begin{figure}[t]
    \centering
    \scalebox{1}{ 
            \Qcircuit @C=1.2em @R=1.0em {
                \barrier[-0cm]{7} &  & \mbox{Step 1} & \barrier[-0cm]{7} &  & \mbox{Step 2} &  &  &  & \barrier[-0cm]{7} &  & \mbox{Step 3} &  &  \\
              \lstick{\ket{\lambda_{1}}} & \qw & \ctrl{4} & \qw & \qw & \targ & \qw & \qw & \qw & \qw & \qw & \qw & \ctrlo{1} & \qw \\
              \lstick{\ket{\lambda_{2}}} & \qw & \qw & \ctrl{4} & \qw & \qw & \qw & \qw & \targ & \qw & \qw & \qw & \ctrlo{1} & \qw \\
              \lstick{\ket{\lambda_{3}}} & \qw & \qw & \qw & \qw & \qw & \targ & \qw & \qw & \targ & \qw & \qw & \ctrlo{1} & \qw \\
              \lstick{\ket{\lambda_{4}}} & \qw & \qw & \qw & \qw & \qw & \qw & \targ & \qw & \qw & \qw & \qw & \ctrlo{3} & \qw \\
              \lstick{\ket{0}_{a_1}} & \qw & \targ & \qw & \qw & \ctrl{-4} & \ctrl{-2} & \ctrl{-1} & \qw & \qw & \qw & \qw & \qw & \qw \\
              \lstick{\ket{0}_{a_2}} & \qw & \qw & \targ & \qw & \qw & \qw & \qw & \ctrl{-4} & \ctrl{-3} & \qw & \qw & \qw & \qw \\
              \lstick{\ket{0}_f} & \qw & \qw & \qw & \qw & \qw & \qw & \qw & \qw & \qw & \qw & \qw & \targ & \qw
            }
    }
    \caption{Subspace identifier circuit $U_\perp$ for the binary case of \cref{lemma:single-shot-success}. As an example, we take $G=\mZ_2^4$ and let $K^\perp$ be generated by $s^{\perp(1)}=1011$ and $s^{\perp(2)}=0110$. }
    \label{fig:subspace-identifier-approach-1}
\end{figure}
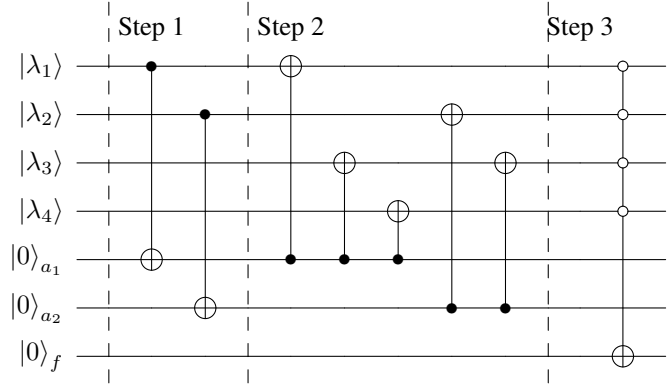

\begin{proposition}[Subspace identifier construction]
\label{lemma:single-shot-success}Let $G=\mZ_2^n$ and $H^\perp \leq \widehat{G}$. Let $S^\perp=\{s^{\perp(1)},\ldots,s^{\perp(i)}\}$ be a generating set for a subgroup $K^\perp \leq H^\perp$, and suppose that the encoding matrix $A\in\mathbb{Z}_2^{i\times n}$ with rows $s^{\perp(1)},\ldots,s^{\perp(i)}$ is in RREF with pivot columns $j_1,\ldots,j_i$.  Then one can construct a subspace identifier $U_\perp$ using $i$ auxiliary qubits and $O(ni)$ elementary gates. 
\end{proposition}
\begin{proof}
If $S^\perp\neq \emptyset$, then $U^\perp$ is the circuit that starts from the output state of the Fourier sampling circuit $|\varphi_F\rangle$, takes $i$ number of qubits in the ancilla register (third register), one flag qubit, and performs the following:
\begin{enumerate}
    \item Classically read the position of pivotal columns $j_1,\cdots,j_i$. 
Perform $i$ CNOT gates where the $j_k$-th qubit of the first register is the control and the $k$-th qubit of the ancilla register is the target; $k$ runs from $1$ to $i$. By~\Cref{lemma:decompose-H-perp}, $H^\perp$ admits a direct sum decomposition $H^\perp=K^\perp\oplus Q$, and the resulting state vector is 
\begin{equation}
\begin{aligned}
    &\sum_{\lambda\in H^\perp}\sqrt{P(\lambda)} |\lambda\rangle\otimes|\varphi_\lambda\rangle\otimes |\lambda_{j_1}\rangle|\lambda_{j_2}\rangle\cdots|\lambda_{j_{i}}\rangle\otimes|0\rangle_f\\
    =& \sum_{\{\iota\}} \sum_{q\in Q}\sqrt{P((\sum_{k=1}^{i}\iota_k s^{\perp(k)})+ q)} \left|(\sum_{k=1}^{i}\iota_k s^{\perp(k)})+ q\right\rangle \otimes |\varphi_{(\sum_{k=1}^{i}\iota_k s^{\perp(k)})+ q}\rangle\otimes |\iota_1\rangle|\iota_2\rangle\cdots|\iota_{i}\rangle\otimes|0\rangle_f.
\end{aligned}
\end{equation}
 
\item If the $k$-th qubit of the ancilla register is 1, apply $x\ra x + s^{\perp(k)}$ to the first register, $k$ runs from $1$ to $i$. 
This step requires at most $ni$ elementary gates. 
\item Apply a zero-controlled multi-qubit Toffoli, which applies an $X$ gate on the flag qubit, controlled on all irrep-register qubits being in $|0\rangle$. 
\end{enumerate}
 
As an illustration, we show an example of the circuit $U_\perp$ in~\cref{fig:subspace-identifier-approach-1}. The effect of applying $U_\perp$ is
\begin{equation}
  \begin{aligned}U_\perp|\varphi_F\rangle|0\rangle_a|0\rangle_f&=\sum_{\{\iota\}}\sum_{q\in Q,q\neq \bzr} \sqrt{P((\sum_{k=1}^{i}\iota_k s^{\perp(k)})+ q)} \left| q\right\rangle \otimes |\varphi_{(\sum_{k=1}^{i}\iota_k s^{\perp(k)})+ q}\rangle\otimes |\iota_1\rangle|\iota_2\rangle\cdots|\iota_{i}\rangle\otimes|0\rangle_f\\
  &+\sum_{\{\iota\}} \sqrt{P((\sum_{k=1}^{i}\iota_k s^{\perp(k)}))} \left| \bzr\right\rangle \otimes |\varphi_{(\sum_{k=1}^{i}\iota_k s^{\perp(k)})}\rangle\otimes |\iota_1\rangle|\iota_2\rangle\cdots|\iota_{i}\rangle\otimes|1\rangle_f.
  \end{aligned}
\end{equation}
The state in the first line is identified with $|\psi_G\rangle\otimes|0\rangle_f$ and the state in the second line is identified with $|\psi_B\rangle\otimes|1\rangle_f$. If the measurement of the flag register is 0, then measuring the irrep register will return a non-zero element of subgroup $Q$, which is, an element of $H^\perp\setminus K^\perp$. 

If $S^\perp=\emptyset$, then ignore steps 1 and 2 and only perform step 3 to obtain the circuit for $U^\perp$. 
\end{proof}

\subsection{Generalization to other abelian groups}
In this section, we generalize the algorithm to the case where $G=\mathbb{Z}_M^n$ with a generic $M\in\mathbb{Z}_+$. If $M$ is prime, the algorithm is identical to the case of $\mZ_2$ by bringing the encoding matrix $A$ to RREF. If $M$ is non-prime, we bring the encoding matrix $A$ to Howell normal form and use the following lemma. 

\begin{lemma}[Decomposition of $H^\perp$, generalized to other abelian groups]
    Let $G=\mathbb{Z}_M^n$, $K^{\perp} \leq H^\perp \leq \hat{G}$ be such that $K^{\perp} = \mspan(A)$
    for some $A \in \mathbb{Z}_M^{i \times n}$ in Howell normal form, with $i$ nonzero rows and 
    pivot columns $j_1 < \cdots < j_{i}$. Then any $\lambda \in H^\perp$ decomposes 
    uniquely as
    \begin{equation}
        \lambda = \sum_{k=1}^{i} \iota_k A_k + q,
    \end{equation}
    where $\iota_k \in [M/A_{k,j_k}]$ for $k = 1, \ldots, i$, and $q_{j_k}\in[A_{k,j_k}]$ for $k = 1, \ldots, i$. In particular, $q\notin K^{\perp}$ if $q\neq \mathbf{0}$. 
\end{lemma}

The proof is the procedure in membership testing~(\cref{prop:membership-testing-HNF-app}). 
With the above decomposition,~\cref{lemma:single-shot-success} can be implemented in a similar way. Specifically,  define $Q$ as 
\begin{equation}
    Q:=\{\lambda\in H^\perp| \lambda_{j_k}\in [A_{k,j_k}],\forall k=1,\cdots,i\}. 
\end{equation}
We note that $Q$ is in bijection with the quotient group $H^\perp / K^{\perp}$, although $H^\perp$ is in general not the direct sum of $K^{\perp}$ and $Q$. Using $Q$, we can again replace $\sum_{\lambda\in H^\perp}$ by $\sum_{\iota_1\cdots\iota_{i}} \sum_{q\in Q}$ and the rest of~\cref{lemma:single-shot-success} follows.  
\end{document}